\documentclass[aps, prx, superscriptaddress, reprint, floatfix]{revtex4-2}

\RequirePackage[l2tabu, orthodox]{nag}
\usepackage[normalem]{ulem} 
\usepackage[utf8]{inputenc}
\usepackage[T1]{fontenc}

\usepackage{amsfonts, amsmath, amsthm, amssymb} 
\usepackage{mathtools}

\usepackage[caption = false]{subfig}

\usepackage{bbm, bm}
\usepackage{MnSymbol}
\usepackage{enumitem}

\usepackage{float}

\usepackage[svgnames]{xcolor} % Specify colors by their svgnames', 
\usepackage{graphicx} % Required for including images
\definecolor{myurlcolor}{rgb}{0,0,0.7}
\definecolor{myrefcolor}{rgb}{0.8,0,0}

\usepackage[
	breaklinks,
	pdftex,
	pdfauthor={Hanna Wojewodka-Sciazko, Maciej Stankiewicz, Edgar A. Aguilar, Christopher Perry, Piotr Cwiklinski, Andrzej Grudka, Karol Horodecki, Michal Horodecki},
	pdftitle={Thermal Operations, in general, are not memoryless},
	pdfsubject={Quantum Information Theory},
	pdfkeywords={},
	pdfproducer={Latex with hyperref},
	pdfcreator={pdflatex},
	colorlinks=true, 
	linkcolor=myrefcolor,
	citecolor=myurlcolor,
	urlcolor=myurlcolor
]{hyperref}

\usepackage{orcidlink}

\newtheorem{theorem}{Theorem}[section]

\newtheorem{lemma}[theorem]{Lemma}
\newtheorem{remark}[theorem]{Remark}
\newtheorem{corollary}[theorem]{Corollary}

\newtheorem{definition}[theorem]{Definition}

\newtheorem{fact}[theorem]{Fact}

\DeclareMathOperator{\E}{\mathbb{E}}
 \def\blk{\color{black}}
 \newcommand{\la}{\langle}
 \newcommand{\ra}{\rangle}
 \newcommand{\ONE}{\mathbbm{1}}
 \newcommand{\swap}{\mathbbm{S}}
 \def\<{\langle}
 \def\>{\rangle}

\newcommand{\coww}{\textnormal{\textbf{CO}}_0}

\newcommand{\mcoww}{\textnormal{\textbf{MCO}}_0}

\begin{document}

\title{
%Thermal Operations cannot be simulated by classical-like Markovian dynamics
Thermal Operations in general are not memoryless
}

\author{Hanna Wojewódka-\'Sciążko \orcidlink{0000-0002-8248-5018}}
%\affiliation{Institute of Theoretical and Applied Informatics, Polish Academy of Sciences, Bałtycka 5, 44-100 Gliwice, Poland}
\affiliation{Institute of Mathematics, University of Silesia in Katowice, Bankowa 14, 40-007 Katowice, Poland}

\author{Maciej Stankiewicz \orcidlink{0000-0001-8288-3860}}
\affiliation{Institute of Informatics, Faculty of Mathematics, Physics and Informatics, University of Gdańsk, Wita Stwosza 57, 80-308 Gdańsk, Poland}
%\affiliation{International Centre for Theory of Quantum Technologies, University of Gdańsk, Wita Stwosza 63, 80-308 Gdańsk, Poland}

\author{Edgar A. Aguilar \orcidlink{0000-0002-1177-9246}}
\affiliation{Institute of Theoretical Physics and Astrophysics, National Quantum Information Centre, Faculty of Mathematics, Physics and Informatics, University of Gdańsk, Wita Stwosza 57, 80-308 Gdańsk, Poland}
\affiliation{Institute of Quantum Optics and Quantum Information, Austrian Academy of Sciences, Boltzmanngasse 3, 1090 Vienna, Austria}

\author{Christopher Perry~\orcidlink{0000-0001-5257-1209}}
\affiliation{QMATH, Department of Mathematical Sciences, University of Copenhagen, Universitetsparken 5, 2100 Copenhagen, Denmark}

\author{Piotr \'Cwikli\'nski \orcidlink{0000-0002-9315-2876}}
\affiliation{Institute of Theoretical Physics and Astrophysics, National Quantum Information Centre, Faculty of Mathematics, Physics and Informatics, University of Gdańsk, Wita Stwosza 57, 80-308 Gdańsk, Poland}

\author{Andrzej Grudka \orcidlink{0000-0002-8390-1458}}
\affiliation{Institute of Spintronics and Quantum Information, Faculty of Physics, Adam Mickiewicz University, 61-614 Poznań, Poland}

\author{Christopher Jarzynski \orcidlink{0000-0002-3464-2920}}
\affiliation{Department of Chemistry and Biochemistry, University of Maryland, College Park, MD 20742, USA}
\affiliation{Institute for Physical Science and Technology, University of Maryland, College Park, MD 20742, USA}
\affiliation{Department of Physics, University of Maryland, College Park, MD 20742, USA}

\author{Karol Horodecki \orcidlink{0000-0001-7540-4147}}
\affiliation{Institute of Informatics, National Quantum Information Centre, Faculty of	Mathematics, Physics and Informatics, University of Gdańsk, Wita Stwosza 57, 80-308 Gdańsk, Poland}
%\affiliation{International Centre for Theory of Quantum Technologies, University of Gdańsk, Wita Stwosza 63, 80-308 Gdańsk, Poland}

\author{Michał Horodecki \orcidlink{0000-0002-0446-3059}}
%\email[E-mail: ]{michal.horodecki@ug.edu.pl}
\affiliation{Institute of Theoretical Physics and Astrophysics, National Quantum Information Centre, Faculty of Mathematics, Physics and Informatics, University of Gdańsk, Wita Stwosza 57, 80-308 Gdańsk, Poland}
\affiliation{International Centre for Theory of Quantum Technologies, University of Gdańsk, prof.\ Marii Janion 4, 80-309 Gdańsk, Poland}
%\affiliation{International Centre for Theory of Quantum Technologies, University of Gdańsk, Wita Stwosza 63, 80-308 Gdańsk, Poland}

\date{\today}%! change to constant date before publishing

\begin{abstract}
	So-called Thermal Operations seem to describe the most fundamental and reasonable set of operations allowable for state transformations at an ambient inverse temperature $\beta$. 
	However, a priori, they require experimentalists to manipulate very complex environments and have control over their internal degrees of freedom.
	For this reason, the community has been working on creating more experimentally-friendly operations. 
	In [Perry, et al., Phys.\ Rev.\ X 8, 041049], it was shown that for states diagonal in the energy basis, Thermal Operations can be performed by so-called Coarse Operations, which need just one auxiliary qubit but are otherwise Markovian and classical in spirit. 
	In this work, by providing an explicit counterexample, we show that this one qubit of memory is necessary. 
	We also provide almost  full characterization (i.e., full, but with one intriguing exception that remains open) of possible transitions that do not require memory in the case where the system is a qubit.
	We do this by analyzing arbitrary control sequences comprising level energy changes and partial thermalizations in each step. 
	
%This complements the independent result of Korzekwa and Lostaglio [arXiv:2005.02403], who have shown that when quantum Markovian dynamics are considered one does not need the extra memory qubit. However, for fixed Hamiltonians and classical Markovian dynamics such memory is needed, thus there exists a quantum advantage for simulating thermal transitions. Our result establishes that the quantum advantage remains even in the general case when the system Hamiltonian is allowed to change.
	
	%..that one does not need the qubit memory, when quantum Markovian dynamics, while for fixed Hamiltonian and classical Markovian dynamics such memory is needed, obtaining a quantum advantage for simulating thermal transitions. Our result thus allows establishing such quantum advantage even in the general case of possibly changing the Hamiltonian of the system.
\end{abstract}

\maketitle

%===============================================================================
% Begin document body
%===============================================================================

\section{Introduction}
		
Due to the rapid growth of technology, the field of quantum thermodynamics is booming.
The advent of nanotechnology has a plethora of questions about quantum and discrete-size effects in thermodynamics. 
Already in the past couple of years, we can observe experimental progress in studying thermodynamics at the quantum level \cite{Rossnagel2016, Cottet2017, Klatzow2018, Goldwater2019,Aamir_2025,uusnakki2026autonomousquantumheatengine}. 
One of the fundamental questions in this field is which states can be converted into one another via a suitable set of operations. 
Of course, defining a suitable set of operations is no easy task, and much work has been devoted to addressing this issue. 
Of particular interest is the set of Thermal Operations (\textbf{TO}) \cite{Janzing00, Ruch76, Streater95, HO2013} 
which have been recently widely used within a~resource-theoretic approach to quantum thermodynamics \cite{Goold_2016, Vitagliano2019}. 
In a nutshell, Thermal Operations consist of appending a heat bath to the system, applying the energy-conserving unitary transformation, and then discarding the heat bath.  
		
A general criterion is not known which would allow to decide for any two pairs of states, whether transition between them is possible by Thermal Operations \cite{Cwiklinski2015, Korzekwa-TO-coh-2015}. 
Yet for states diagonal in the energy basis, the transformation $\rho\xrightarrow{\textbf{TO}} \sigma$ is allowed iff $\rho$ ``thermomajorizes'' $\sigma$, i.e., $\rho \succ_T \sigma$ \cite{HO2013, RuchSS78}. 
This ``thermomajorization'' condition can also be used to calculate how much deterministic work can be extracted from a state by appending a 2-level battery (work bit) and maximizing over the energy gap of the battery.
The problem with Thermal Operations, however, is that, as used in \cite{HO2013}, they are intractable for experiments. 
The energy-conserving unitary $U$ would have to be very delicately made and couple the system to the bath in extremely precise ways. 
Furthermore, the Hamiltonian of the environment itself is also tailor-made with the correct exponentially growing degeneracies. 
Hence, a search for more experimentally friendly operations was performed in \cite{Perry_etal}, in the spirit of \cite{Anders2013, Alicki79}, and the set of the so-called Coarse Operations was proposed. 

More recently, related Markovian approaches to thermodynamic state transformations
have been developed under the name of Markovian thermal processes. In particular,
continuous thermomajorization gives necessary and sufficient conditions for
population dynamics generated by Markovian master equations and shows that
elementary thermalizations form a universal set of controls
\cite{LostaglioKorzekwa2022_continuous_thermomajorization}. This was further
used to optimize thermalization protocols in tasks such as cooling, work extraction
and catalysis \cite{KorzekwaLostaglio2022_optimizing_thermalization}.

Unlike in Thermal Operations, where there is access to an arbitrary heat bath, with Coarse Operations, it is only permissible to append a two-level auxiliary system in thermal equilibrium.  
Moreover, contact with the reservoir is allowed in the form of partial thermalizations, and one can change the energy levels of the total system (consisting of the original system and the auxiliary one).
Such operations are experimentally friendly, as the fine-grained control is needed only for the system of interest and the auxiliary qubit system, while the contact with a large heat bath is Markovian, and it does not require any fine-grained control over heat bath degrees of freedom. 
Therefore Coarse Operations are much more in the spirit of traditional thermodynamics.

Since Coarse Operations allow for raising and lowering energy levels, putting work into a system is also an option so that pretty much any transformation is possible. 
It only needs to respect superselection rules (i.e., coherence cannot be created between different energies).
%Since Coarse Operations allow for raising and %lowering energy levels,
%they allow, in particular, to put work, so that pretty much any %transformation is possible, that respects superselection rules %(i.e., %coherence cannot be created between different energies).
One thus considers a subclass of Coarse Operations, denoted here by $\coww$, that does not use work, i.e., work can be borrowed, but in the end, should be returned with arbitrary accuracy. 
Also, Coarse Operations are almost Markovian -- the auxiliary qubit is the only memory here.
%\marginpar{the auxiliary qubit is the only memory here}

The main result of \cite{Perry_etal} was to show that \textbf{TO} $\subset$ $\coww$. 
I.e., if Thermal Operations allow for a state transformation $\rho\xrightarrow{\textbf{TO}} \sigma$, then the same transformation is possible under Coarse Operations without spending work: $\rho\xrightarrow{\coww} \sigma$.
		
In this article, we investigate whether auxiliary systems are necessary for $\coww$  to simulate Thermal Operations.
The purpose is to find a minimal set of experimentally friendly operations that can still produce the same transformations. This question is also closely related to recent work on memory-assisted
Markovian thermal processes, where ancillary thermal systems are used to promote
memoryless Markovian thermal dynamics to effectively non-Markovian dynamics
\cite{CzartowskiOliveiraKorzekwa2023_thermal_recall}. In contrast to that
general memory-assisted framework, our aim is to identify a concrete obstruction
already for a qubit system controlled by Coarse Operations without an auxiliary
memory.
% This would be preferable because it is questionable whether auxiliary systems of arbitrary energy gaps are attainable in practice. 
% \mh{I would remove the sentence above. while doing $\mcoww$ we allow for whatever energy gaps.
% }
% - or if this really does not incur any work cost.

%In other words, the main question is: %\textit{Do Thermal Operations require memory?}

In the paper, we show that auxiliary systems are indeed necessary for Coarse Operations  that do not spend work to reproduce Thermal Operations.
To this end, we \textit{characterize almost all transitions between diagonal qubit states that are possible by means of Coarse Operations without spending work and without using auxiliary systems.} 
We leave open only the following situation: when either the initial state or the final  state is the pure ground state. 
Apart from this exceptional case, the characterization is quite simple: the only transformations that can be implemented by Coarse Operations without auxiliary systems and without work expenditure are those that merely admix a Gibbs state, together with those for which the initial state is the pure excited state. 
Thus if a  state $\rho$, is less excited than the Gibbs state (and is not pure ground one), it cannot be transformed into a state that is more excited than the Gibbs state (unless it is pure excited state). %So for any such pair, if only transition is possible by Thermal Operations, constitutes example of transition which can be done by Coarse Opertations, but cannot be done by Coarse Operations without memory. 
Hence, for any such pair, any transition achievable by Thermal Operations constitutes an example of a transition that can be implemented by Coarse Operations, but cannot be implemented by Coarse Operations without memory.

Note, that Thermal Operations were shown to be achievable without memory \cite{KorzekwaL2020-TO-advantage}; however, their implementation did not proceed via Coarse Operations. Moreover the system–reservoir interaction allowed in \cite{KorzekwaL2020-TO-advantage} cannot be described within the standard Markovian dynamics in the weak-coupling regime. Specifically, the latter dynamics do not mix diagonal and off-diagonal elements of the density matrix -- a property that holds within Coarse Operations but is 
not obeyed in scenario of \cite{KorzekwaL2020-TO-advantage}.

\section{Results}\label{sec:main_results}

Coarse Operations, as defined in \cite{Perry_etal}, consist of the following basic bricks: 
% \begin{definition}
% % [Coarse Operations (\textbf{CO})]
% \leavevmode
	\label{def:COprim}
	\begin{enumerate}[align=left]
		\item [\textbf{CO1:}] A single 2-level system at thermal equilibrium may be appended as an auxiliary system.
		\item [\textbf{CO2:}] Partial thermalizations may be performed on a~chosen subset of energy levels.
		\item [\textbf{CO$\,'$3:}] Energy levels may be shifted up or down as long as the work cost is accounted for.
		\item [\textbf{CO4:}] The auxiliary system may be discarded.
		\item [\textbf{CO5:}] Any energy-conserving unitary operation $U$, such that $[U,H_S]=0$, may be applied.
	\end{enumerate}
% \end{definition}

There is one subtle issue about the above definition, which we have to discuss here.
Namely, the unitary transformation from item \textbf{CO5} may create coherences. 
When raising or lowering an energy level, it is possible to make two levels have equal energy. 
At this moment, any unitary that acts on these two levels will commute with the system Hamiltonian, $H_S$, and may generate coherences.

The problem is that shifting levels of a state with coherences does not lead to a well-defined random variable of work (see \cite{Acin-no-go-coherences-work}) which is needed if we want to account not only for average work but also for work fluctuations.
The latter is crucial for our results. 
Indeed, the phrase ``without spending work'' in \cite{Perry_etal}, and also in this paper, means that one can spend only an asymptotically vanishing amount of work with vanishing probability. 

To this end, in order to treat work as a random variable, we shall modify the operation \textbf{CO$\,'$3}. 
Namely, together with raising or lowering the energy levels, we shall dephase the suitable level, i.e., remove coherences between that level and the other ones. 

The modified item will thus be of the following form:
\begin{enumerate}[align=left]
    \item [\textbf{CO3:}] Energy levels may be shifted up or down as long as the work cost is accounted for. 
    Before the shift, the shifted level is dephased. 
\end{enumerate}

We are now in a position to define Coarse Operations in a way that we will use in the paper:

\begin{definition}[Coarse Operations (\textbf{CO})]
    \leavevmode
    Coarse Operations are defined as any composition of operations  
    \textbf{CO1}, 
    \textbf{CO2},
    \textbf{CO3},
    \textbf{CO4}, and
    \textbf{CO5} 
    described above.
\end{definition}
We now consider transitions that can be implemented using Coarse Operations which do not involve an auxiliary system and do not require work. To formalize this setting, we define below the class of Markovian Coarse Operations \textbf{MCO})

\begin{definition}\label{def:MCO_0}
    We say that a diagonal state $\rho$ can be transformed into a diagonal state $\sigma$ by Markovian Coarse Operations with no work expenditure ($\mcoww$) when for every $\epsilon >0$ and $\eta > 0$ there exists an operation $\Lambda$ composed of operations from the sets \textbf{CO2}, \textbf{CO3}, and \textbf{CO5} such that:
    \begin{equation}
        \Lambda(\rho)=\sigma,
    \end{equation}
    the final Hamiltonian is the same as the original one, and 
    \begin{equation}
        P(-W > \epsilon) \leq \eta,
    \end{equation}
    where the work random variable $W$ is defined in Section \ref{sec:notation}, Eq.\ \eqref{eq:work-def}.
\end{definition}

Our class \(\mcoww\) should not be confused with the most general Markovian
thermal processes generated by time-dependent master equations. The latter have
recently been characterized by continuous thermomajorization
\cite{LostaglioKorzekwa2022_continuous_thermomajorization} and, in a controlled
quantum setting, by reachable-set methods for Markovian quantum dynamics with
thermal resources
\cite{vomEndeMalvettiDirrSchulteHerbrueggen2023_controlled_markovian}. Here we
focus on the more restrictive operational model inherited from Coarse Operations:
level transformations, partial thermalizations, and energy-preserving unitaries at
degeneracies, together with a fluctuating-work no-expenditure condition.

We can now state our main result:
\begin{theorem}
    \label{thm:main} Consider two diagonal states $\rho$ and $\sigma$ (none of them pure), and let $\tau_\beta$ denote the thermal state corresponding to the initial Hamiltonian. 
    Then $\rho$ can be transformed into $\sigma$ by means of $\mcoww$ if and only if 
$\sigma$ is of the form
\begin{equation}
    \label{eq:sigma-thm}
    \sigma = (1-\lambda)\rho + \lambda \tau_\beta
    \quad \text{for some} \quad \lambda \in [0,1],
\end{equation}
\end{theorem}

\begin{proof}
Consider first the case, 
where \textbf{CO5} 
are not used. Then Theorem~\ref{thm:main_appendix} excludes from \textbf{MCO$_0$} all transitions apart from those mentioned above. 
    In Appendix  \ref{sec:appendix-with-swaps}, we argue that adding \textbf{CO5} will not change this state of affairs (see Theorem \ref{thm:main_appendix_swaps}).  On the other hand, 
    the above transition  can be done by $\mcoww$, since is just an example of \textbf{CO2}. We thus  obtain "if and only if". 
\end{proof}
Let us now discuss the cases that are not covered by the theorem, i.e. 
when either of the states is pure.

(i) Suppose first that the initial state is pure excited. Then one can achieve arbitrary final state by $\mcoww$. 
Indeed, by partial thermalization alone one can obtain any state more excited than the Gibbs state (and the Gibbs state itself is, of course, attainable by full thermalization).  
Also, one can obtain the ground state by means of  $\mcoww$, by lowering the upper level to zero, performing the swap, and raising the empty level back to  $E_0$. 
Then from the ground state by partial thermalization, we can obtain any state that is less excited than Gibbs. Thus all states are achievable in this manner.  

(ii) When the initial state is the ground state, we believe that still only states 
of the form \eqref{eq:sigma-thm} 
are achievable by $\mcoww$, but we do not have a proof. 

(iii) If the final state 
is excited, then we know that it cannot be obtained from any initial state by \textbf{TO}. So we believe the same is true for $\coww$ and therefore also for $\mcoww$. However we do not  have a formal proof, that $\coww$ is a subset of \textbf{TO}, although we believe so.

(iv) Finally, if the final state is ground, 
 then by \textbf{TO} one can achieve it only from initial pure excited state (and trivially from pure ground state), and we again believe that the same can be stated about $\coww$.

The nearly full characterization of transformations achievable by $\mcoww$ proves, in particular, that some transitions performed by Coarse Operations that do not use work ($\coww$) cannot be obtained by their Markovian version $\mcoww$. 
For instance, Thermal Operations allow for transitions from a state less excited than the Gibbs one to a state more excited than the Gibbs state. 
However, according to the above theorem, it is not possible to perform such a~transition by $\mcoww$. 

\noindent
\noindent

\textbf{Example:} \textit{Transition achievable by \textbf{TO} but not by $\mcoww$.} 
% Consider a qubit system with Hamiltonian $H = E\,|1\rangle\langle 1|$. Using a Thermal Operation, one can implement the transition
% \begin{equation}
%     |0\rangle\langle 0| \;\mapsto\; \left(1 - e^{-\beta E}\right)|1\rangle\langle 1| + e^{-\beta E}|0\rangle\langle 0|.
% \end{equation}
% Such a transformation can be realized by an extremal Thermal Operation (see, e.g., \cite{Mazurek_2018}). 

% {\color{red} The resulting state is more excited than the Gibbs state. Therefore, by Theorem~\ref{thm:main}, this transition cannot be achieved using $\mcoww$. On the other hand, since Coarse Operations with memory are equivalent to Thermal Operations, this transition is achievable when memory is allowed.} 

Consider a qubit system with Hamiltonian $H = E_0\,|1\rangle\langle 1|$.
Consider a thermal operation on diagonal states given by
\begin{align}
\Lambda =
\begin{pmatrix}
1 - e^{-\beta E_0} & 1 \\
e^{-\beta E_0} & 0
\end{pmatrix},
\end{align}
(cf.\ \cite{Mazurek_2018}).
Consider an almost ground diagonal state 
$\rho_\epsilon$
\begin{align}
\rho_\epsilon =
\begin{pmatrix}
1 - \epsilon \\
\epsilon
\end{pmatrix}.
\end{align}
Then
\begin{align}
\Lambda
\begin{pmatrix}
1 - \epsilon \\
\epsilon
\end{pmatrix}
=
\begin{pmatrix}
(1 - e^{-\beta E_0})(1 - \epsilon) + \epsilon \\
e^{-\beta E_0}(1 - \epsilon)
\end{pmatrix} 
\equiv \sigma.
\end{align}
We now find $\epsilon$ such that this output state $\sigma$ is more excited than the Gibbs state.
The latter means that
\begin{align}
e^{-\beta E_0}(1 - \epsilon)
&> \frac{e^{-\beta E_0}}{1 + e^{-\beta E_0}}.
\end{align}
This is equivalent to
\begin{align}
\epsilon < \frac{e^{-\beta E_0}}{1 + e^{-\beta E_0}}.
\end{align}
Thus, if
\begin{align}
0 < \epsilon < \frac{e^{-\beta E_0}}{1 + e^{-\beta E_0}},
\end{align}
then $\rho_\epsilon$ can be transformed by \textbf{TO} (and hence by $\coww$) into $\sigma$, which is not a mixture of $p_\epsilon$ and the Gibbs state (since its excited probability is larger than the Gibbs excited probability, while that of $\rho_\epsilon$ is smaller).

Thus the transition 
\begin{align}
    \rho_\epsilon \to \sigma
\quad \text{for} \quad 0<\epsilon<\frac{e^{-\beta E_0}}{1+e^{-\beta E_0}}    
\end{align}
can be obtained by means of $\coww$ 
but not $\mcoww$.

This example should be compared with extremal thermal operations appearing in
resource-theoretic analyses of non-Markovian thermodynamic protocols. In
particular, extremal thermal operations have been used to model non-Markovian
thermalization strokes in quantum heat engines, where memory effects can improve
both average performance and work fluctuations
\cite{Ptaszynski2022_nonmarkovian_thermal_engines}. Our result gives a
complementary structural statement: even in the simplest qubit setting, such
extremal population transfers cannot, in general, be reproduced by memoryless
Coarse Operations under the fluctuating-work no-expenditure requirement.

\section{Setup and Notation}
\label{sec:notation}

\subsection{System}
\label{subsec:system}

In this paper we consider a two-dimensional quantum system whose Hamiltonian is parameterized as 
\(H(E) = \operatorname{diag}(0, E)\) with \(E \ge 0\), 
which can be assumed without loss of generality.  
The initial (and final) Hamiltonian is given by $H(E_0)$. 
Both the initial state $\rho = (1-p_{\mathrm{IN}},p_{\mathrm{IN}})$ and the final state $\sigma = (1-p_{\mathrm{OUT}},p_{\mathrm{OUT}})$ are assumed to be diagonal in the Hamiltonian basis (i.e., there are no coherence terms). 
Likewise, we consider only positive (inverse) ambient temperatures $\beta$, which are fixed (and finite). 
For any given energy $E$, the partition function $Z_E$ is given by $Z_E = 1 + e^{-\beta E}$. 
The thermal state (also referred to as the equilibrium or the Gibbs state) is 
\begin{equation}
    \tau_E = (1-g(E),g(E))
\end{equation}
with 
\begin{equation}
    g(E) = \frac{1}{Z_E} e^{-\beta E}.    
\end{equation}
In the notation, we often do not make the inverse temperature $\beta$ explicit since it is constant. 
For any given probability $p\in(0,1/2]$, we shall define energy $E(p)$ such that the state $(1-p,p)$ is the Gibbs state with a Hamiltonian $H(E(p))$, i.e.,
\begin{equation}
    E(p) = -\frac{1}{\beta} \ln\left( \frac{p}{1-p} \right).
\end{equation}
Importantly, we then have that 
\begin{equation}
    \label{eq:gep}
    g(E(p))=p\quad\text{for any}\quad p\in(0,1/2].
\end{equation}

\subsection{Operations}
\label{subsec:operations}

An arbitrary Coarse Operation with an auxiliary qubit is a sequence of three types of operations. The operation \textbf{CO2}, which we will call a Partial Thermalization (PT), consists of admixing the Gibbs state with some probability.
Formally, we view a partial thermalization as a map that thermalizes the current state $\rho$ with a Hamiltonian $H(E)$ with probability $\lambda$, that is,
\begin{equation}\label{eq:partial-thermalization-gamma}
	\rho \mapsto (1-\lambda)\rho + \lambda \tau_E.
\end{equation}
The transformation keeps the Hamiltonian of the system intact. 

Operation \textbf{CO3}, called a Level Transformation (LT), shifts levels of the Hamiltonian of our two-level system. 
Without loss of generality, one can shift just the upper level. 
An energy increment in a single step will be denoted by $\Delta E_i$, i.e., $H(E_{i+1}) = H(E_i + \Delta E_i)$, and it may be different in each step. 

Finally, consider \textbf{CO5}.
Note that without loss of generality, we can dephase the qubit after \textbf{CO5}. 
One easily finds that the effect of such a transformation is the following: 
\begin{equation}
    \rho \mapsto (1-\zeta) \rho + \zeta \sigma_x \rho \sigma_x,
\end{equation}
where $\sigma_x $ is the Pauli-$X$ matrix. 
Thus with probability $\zeta$, a bit flip is performed on the qubit. 
For this reason, we will call it a Bistochastic Transformation (BT).

Further, note that it acts in a nontrivial way (i.e., with $\zeta\not=0$) only when the Hamiltonian of the system is trivial, i.e., the energy of the upper level is equal to the energy of the lower level. 

Now, let us note that the composition of several steps of a given type results in a single step of this type. 
E.g., several partial thermalizations applied one after another is a single partial thermalization. 
Therefore the most general Coarse Operation without access to an auxiliary qubit is the sequence of the following form: 
$\text{PT}_1\rightarrow\text{LT}_1
\rightarrow\text{BT}_1
\rightarrow\text{PT}_2
\rightarrow\text{LT}_2\rightarrow \text{BT}_2 \rightarrow\dots\rightarrow \text{PT}_N\rightarrow\text{LT}_N \rightarrow\text{BT}_N$.

Let us emphasize that we are interested in transitions between systems with the same Hamiltonian. Therefore, we consider only such sequences of operations for which, after the last step, the Hamiltonian of the system coincides with the initial one (in particular, any level transformation performed during the process must be undone at the end).

%Later, it will be shown that both for fixed finite $N$ and limiting cases, the result holds (cf \cite{Alicki79}).

Later it will be  convenient to think of a partial thermalization in each step $i$ as the map 
\begin{equation}
    \mathcal{C}_i = (1-\lambda_i) \ONE + \lambda_i \hat{\tau}_i,
\end{equation}
where $\ONE$ is the identity map, and $\hat{\tau}_i$ is a constant map that outputs a thermal state of the Hamiltonian $H(E_i)$.  
 
Similarly, we represent the bistochastic operation as a~mixture of two maps: 
\begin{equation}
	\mathcal{B}_i = (1-\zeta_i) \ONE + \zeta_i \swap,
\end{equation}
where $\swap$ is a map performing a~bit flip.

%{\color{red}In this paper, we show that -- with probability bounded away from zero -- any transition from $\rho$ to $\sigma$, excluding the specific cases already indicated in Theorem~\ref{thm:main} (namely, when $\rho$ is a pure state or when $\sigma = (1-\lambda)\rho + \lambda \tau_\beta$ for some $\lambda \in [0,1]$), necessarily requires a certain, bounded away from zero,  work expenditure.}

\subsection{Work}

Work is a random variable which, by convention, takes positive values whenever it is performed by the system, i.e.\ when the system's energy decreases ($\Delta E_i\coloneqq E_{i+1}-E_i < 0$). 
From the perspective of a user of the system, this corresponds to a gain of work.

Conversely, work is negative when the system's energy increases, i.e.\ when $\Delta E_i > 0$; in this case, we say that the system acquires work, and thus the user simultaneously experiences work expenditure.

More precisely, for each level transformation we define work as a random variable with the following distribution:
\begin{equation}
	P\left(W_i = -\Delta E_i\right) = p_i,
	\quad
	P\left(W_i = 0\right) = 1 - p_i,
\end{equation}
where $p_i \in [0,1]$ denotes the probability that the excited level is occupied, and $(1 - p_i)$ is the probability that the ground level is occupied.  
The total work performed in an $N$-step process is given by the sum of the work contributions from each step:
\begin{equation}
    \label{eq:work-def}
	W = \sum_{i=1}^N W_i .
\end{equation}
For convenience, let us additionally introduce the work-loss random variable, namely $-W$, which is, of course, fully determined by $W$ itself.

% {\color{blue}[PC: mamy tu niespójność znaków w oznaczeniu
% \(\Delta E_i\). W sekcji o pracy definiujemy \(\Delta E_i=E_{i+1}-E_i\), czyli jako
% wielkość ze znakiem. Wtedy, jeśli poziom wzbudzony jest obsadzony, praca wynosi
% \(W_i=-\Delta E_i\). Zatem przy obniżaniu poziomu mamy \(\Delta E_i<0\) oraz
% \(W_i>0\), czyli dodatnią pracę wyciągniętą z układu.
% Natomiast w protokole dla średniej pracy, w kroku (II), obniżamy energię z
% \(E(p_{\mathrm{IN}})\) do \(E(p_{\mathrm{OUT}})\), ale zapisujemy
% \[
%     \Delta E_i =
%     \frac{E(p_{\mathrm{IN}})-E(p_{\mathrm{OUT}})}
%          {n_{\mathrm{(II)}}}
%     >0 .
% \]
% To nie jest już przyrost ze znakiem \(E_{i+1}-E_i\), tylko dodatnia długość kroku.
% Ta niespójność znaków może potem przechodzić do równania
% \eqref{eq:optimal_work_average} oraz do dowodu lematu rozgrzewkowego.
% Proponowane rozwiązanie: rozdzielić dwie wielkości. Niech
% \[
%     \delta E_i := E_{i+1}-E_i
% \]
% oznacza przyrost ze znakiem, używany we wszystkich wzorach na pracę, natomiast
% \[
%     d_i := |\delta E_i|
% \]
% oznacza dodatnią długość kroku, używaną tylko w opisie geometrycznym ścieżki.
% Wtedy trzeba poprawić wzory tak, aby praca zawsze zależała od \(\delta E_i\),
% a nie od dodatniego \(d_i\).

% {\color{red}[HW-Ś: Piotrze, to co zauważyłam, to poprawiłam. Zerknij proszę, czy nic nie przeoczyłam. Notację wybrałam taką, żeby była zgodna z Appendixem.]}	
%     }

\section{Outline of the paper}

We will approach the problem in several steps, gradually increasing the complexity of the arguments.

Our main goal is to characterize transitions that necessarily require work, in the sense that, with probability greater than some positive constant $\eta >0$, a strictly positive and non-vanishing work loss $-W > \epsilon > 0$ must be incurred.

In Section~\ref{sec:av}, we will depict a situation in which we instead care just about the average work.  Namely,
we analyze a simple, specific energy-change cycle which shows that a transition from an initial state to a final state can be performed without ancillary system with zero average work cost  if the free energy of the initial state is larger than that of the final state.
We do it for states with probability of excited state smaller than $1/2$. 
Most likely, we have also "only if" because,
if initial free energy is larger than the final one, then even using Thermal Operations one cannot perform the transition. And we believe that $\textbf{CO}_0$ 
are equal to Thermal Operations. 

Subsequently, in Section~\ref{sec:worm-up}, we demonstrate that for the same cycle a fixed amount of work is necessarily expended with some fixed positive probability.

In Section~\ref{sec:arbitrary-cycle}, we outline the strategy for proving the same but in full generality (that is, for an arbitrary energy-change cycle). 
In particular, in Section~\ref{sec:tools}, 
we introduce the main tools used throughout the paper, including paths, Gibbs curves, and related notions. 
Then in Section~\ref{subsec:preparations}
we sketch the proof for the restricted case in which \textbf{CO5} operations are excluded.  

Finally, in Section~\ref{subsec:co5}, we extend these considerations to the fully general setting, allowing for \textbf{CO5} operations.

\section{Average work extraction}
\label{sec:av}
In this section, we will show that if we care only about the average work, then transition
$\rho \to \sigma$ can be performed without the expenditure of work provided that the free energy does not increase.
We present case where both states have probability of exicted state not greater than $1/2$.
Thus, not only can we implement transitions that are possible via \textbf{TO}, but we can even realize transitions that cannot be achieved by \textbf{TO}, as long as the free energy does not increase.

More precisely, we will see that one can always go from $\rho$ to $\sigma$ using, on average, an amount of work given by the free energy difference:
\begin{equation}
	\label{eq:avgwork}
	\langle W \rangle = F(\rho,E_0) - F(\sigma,E_0),
\end{equation}
where
$F(\Psi,E) = \langle E \rangle_{\Psi} - \beta^{-1} S(\Psi)$
denotes the Helmholtz free energy of the state $\Psi$.

% \textcolor{magenta}{[MS: Does the above sentence refers to completely arbitrary values of $\rho$ and $\sigma$ even that greater than 1/2? Because the protocol below do not includes CO5 therefore, do not allows to go above 1/2. This is not problem for the final proof but I think that the above statement about average work and free energy difference potentialy should be tweaked.]}
% \mh{ups, no faktycznie... trzeba pomysleć.} {\color{red}[HW-Ś: Gdybyśmy chcieli skończyć powyżej $1/2$, musielibyśmy dodać w części (III) swap. Jeśli chcecie, to mogę to rozpisać (robiłam podobny rachunek w~App.B). Jest szansa, że wynik wyjdzie dokłdnie taki sam. Inne pytanie, czy będziemy chcieli to przeliczenie dodawać do tekstu ;)]}
% \mh{Hm, szkoda roboty...}

We now formulate a protocol that has been constructed so as to achieve an optimal average work extraction (that is, equal exactly to the free energy difference, as in Eq.~\eqref{eq:avgwork}).
%\cred (it seems that this must have been known before, however we have not found such result stated explicitly in literature) \blk
The protocol consists of three stages. 
In stages~(I) and~(III) we perform only level transformations (whence stage~(I) ends exactly at the first thermalization and stage~(III) begins immediately after the last thermalization), while the intermediate stage consists of a sequence of small level transformations interlaced with full thermalizations (i.e.\ with $\lambda = 1$ in Eq.~\eqref{eq:partial-thermalization-gamma}).

More explicitly, we have the following scenario (see Fig.~\ref{fig:average-work}):
%\textcolor{magenta}{MS: I do not see a definition of ``stage'' anywhere above, we probably should define it before the first use.}
\begin{enumerate}[align=left]
	\item [I.] Level transformation raises the excited energy level from $E_0$ to $E(p_{\mathrm{IN}})$ while keeping the state constant. 
	\item [II.] The total energy increment (in stage (II)) is equal to $E(p_{\mathrm{IN}})-E(p_{\mathrm{OUT}})$, and in each step, it is given by a~constant value
    \begin{equation}
        \Delta_i:=\left|\Delta E_i\right|= \frac{E\left(p_{\mathrm{IN}}\right) - E\left(p_{\mathrm{OUT}}\right)}{n_{\mathrm{(II)}}},\quad i\in\left\{1,\ldots,n_\mathrm{(II)}\right\},
    \end{equation}
    where $n_{\mathrm{(II)}}$ denotes the number of steps in stage (II). In each $i$-th step, a full thermalization is performed, that is, $\lambda_i = 1$ for every $i \in \{1, \dots, n_{\mathrm{(II)}}\}$.
	\item [III.] Level transformation raises the excited energy level from $E(p_{\mathrm{OUT}})$ to $E_{0}$ while keeping the state constant. 
\end{enumerate}

\begin{figure}
	\centering
	\includegraphics[width=\linewidth]{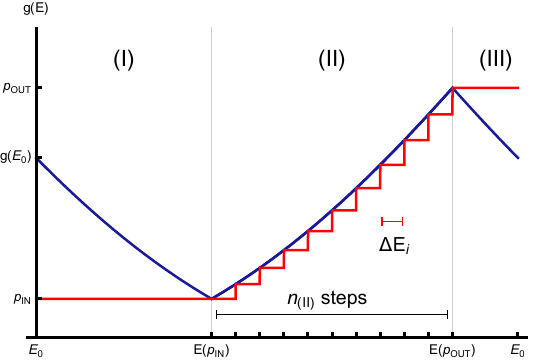}
	\caption{A specific protocol for which the average extracted work $\langle W \rangle$ is exactly equal to the free energy difference $F(\rho, E_0) - F(\sigma, E_0)$.}
	\label{fig:average-work}
\end{figure}

%\begin{figure}
%	\centering
%	\includegraphics[width=\linewidth]{figures/FigOptimalAverage.pdf}
%	\caption{\textcolor{magenta}{OLD image to be removed after verification of the new one. }{\color{red}A specific protocol for which the average extracted work $\langle W \rangle$ is exactly equal to the free energy difference $F(\rho, E_0) - F(\sigma, E_0)$.} {\color{teal}[HW-Ś: Notacja do poprawy, bo się pozmieniała: I, II, III zamieniamy na (I), (II), (III); $p^{\mathrm{out}}$, $p^{\mathrm{in}}$ zamieniamy na $p_\mathrm{OUT}$, $p_\mathrm{IN}$ oraz $n_\mathrm{II}$ i $\Delta E$ na $n_\mathrm{(II)}$, $\Delta E_i$.]}}
%	\label{fig:average-work}
%\end{figure}

In both stages (I) and (III), the excited energy level is raised so that in these stages work is lost on average in each step. 
In stage (II), the opposite is true, so work is gained in this part. 
The work average is then simply given by
\begin{equation}\label{eq:optimal_wor_average}
    \begin{split}
    	\la W \ra=& -p_{\mathrm{IN}}\left(E\left(p_{\mathrm{IN}}\right) - E_0\right) + \sum_{i=1}^{n_{\mathrm{(II)}}} g(E_i)\Delta_i \\
    	&- p_{\mathrm{OUT}}\left(E_0 - E\left(p_{\mathrm{OUT}}\right)\right).
    \end{split}
\end{equation}
Note that we have explicitly put the negative signs to indicate the work loss, where needed. 
While tending with the number of steps in stage (II) to infinity ($n_{\mathrm{(II)}}\to\infty$), the average work in this stage is equal to
\begin{equation}
    \label{eq:gIntegral}
	\begin{split}
        -\int_{E(p_{\mathrm{IN}})}^{E(p_{\mathrm{OUT}})} g(E)dE 
        =& \int_{E(p_{\mathrm{OUT}})}^{E(p_{\mathrm{IN}})} g(E)dE\\
    	=& F_{\beta}\left(E\left(p_{\mathrm{IN}}\right)\right) \\
        &- F_{\beta}\left(E\left(p_{\mathrm{OUT}}\right)\right),
	\end{split}
\end{equation}
where $F_{\beta}(E) = F(\tau_{E},E)$ is the free energy of the thermal state $\tau_E$ corresponding to its Hamiltonian $H(E)$. 
The above equation can be directly verified, by performing integration and using the definition of free energy. 
Notice that this integral is keeping track of the work performed by an isothermal process between energy $E(p_{\mathrm{IN}})$ and $E(p_{\mathrm{OUT}})$, hence the work required is simply the difference of free energies \cite{Aberg2013}.
Then, note that
\begin{equation}
    F_\beta (E)=g(E) E - \frac{1}{\beta} S(E)
\end{equation}
with $S(E)=S(\tau(E))$ being the entropy of the Gibbs state with Hamiltonian $H(E)$, i.e.,
\begin{equation}
    S(E)=S(\{ g(E),1-g(E)\}).
\end{equation}
Since, by definition, for any $p\in(0,1/2]$ we have $g(E(p))=p$ (see Eq.\ (\ref{eq:gep})), then
\begin{equation}
    \label{eq:F-Gibbs-p}
    F_\beta (E(p))=p E(p) - \frac{1}{\beta} S(\{1-p,p\}),
\end{equation}
so that 
\begin{equation}
    \label{eq:F-Gibbs-in-out}
    \begin{split}
        &F_\beta \left(E\left(p_{\mathrm{IN}}\right)\right)=p_{\mathrm{IN}} E\left(p_{\mathrm{IN}}\right) - \frac{1}{\beta} S(\rho), \\
        &F_\beta \left(E\left(p_{\mathrm{OUT}}\right)\right)=p_{\mathrm{OUT}} E\left(p_{\mathrm{OUT}}\right) - \frac{1}{\beta} S(\sigma).
    \end{split}
\end{equation}
Hence, keeping in mind Eqs. \eqref{eq:optimal_wor_average} and \eqref{eq:gIntegral}, we can write 
\begin{equation}
    \begin{split}
        \< W\> =&  -p_{\mathrm{IN}}\left(E\left(p_{\mathrm{IN}}\right) - E_0\right) + F_\beta\left(E\left(p_{\mathrm{IN}}\right)\right)\\
         & -F_\beta\left(E\left(p_{\mathrm{OUT}}\right)\right) - p_{\mathrm{OUT}}\left(E_0 - E\left(p_{\mathrm{OUT}}\right)\right).
        \end{split}
\end{equation}
Inserting Eq.\ \eqref{eq:F-Gibbs-in-out} into the above equation, and using
\begin{equation}
    \begin{split}
        F\left(\rho,E_0\right)=&p_{\mathrm{IN}} E_0 - \frac{1}{\beta} S(\rho),\\
        F\left(\sigma,E_0\right)=&p_{\mathrm{OUT}} E_0 - \frac{1}{\beta} S(\sigma),
    \end{split}
\end{equation}
we obtain that the process described above extracts the free energy difference, as indicated in Eq.\ \eqref{eq:avgwork}.

\section{Warm-up: beyond the average for a single energy-change cycle}
\label{sec:worm-up}
		
In this section, we will make the first steps toward analyzing work fluctuations instead of work averages (similarly to the previous section, we restrict our attention here to the case in which both the initial and final states have excited-state populations not exceeding \(1/2\)).  
Let us consider the process described in Section \ref{sec:av}, that is, one that extracts an optimal amount of work on average (equal to the free energy difference). 
We aim to show that, with probability bounded away from zero (that is, bigger than some constant $\eta>0$, which is independent of the protocol of transforming $\rho$ into $\sigma$), at least some amount of work is lost. 
The result may be formulated as follows:

\begin{lemma}
	\label{lem:simplecase}
	The probability of extracting negative work by the process defined and discussed in Section \ref{sec:av} (in the context of its average work extraction), is bounded away from zero 
    (regardless  of the number of its steps).  
	In particular, the probability that the work loss $-W$ is greater than or equal to
	\begin{equation}
	    \frac{E\left(p_{\mathrm{IN}}\right) - E\left(p_{\mathrm{OUT}}\right)}{2}>0
	\end{equation}
	is lower bounded by
	\begin{equation}
	    p_{\mathrm{IN}}\,p_{\mathrm{OUT}}\left(1-e^{-2\left(\frac{1}{2}-p_{\mathrm{OUT}}\right)^2}\right),
	\end{equation}
    where the above is positive, provided that and $p_\mathrm{IN},p_\mathrm{OUT}\in(0,1/2]$.
\end{lemma}

\begin{proof}
First, note that with probability $p_{\mathrm{IN}}$ the excited level is initially occupied (up to the first thermalization),
which leads to a  positive work loss of $E(p_{\mathrm{IN}})-E_0$ in stage~(I), that is, $-W_\mathrm{(I)}=E(p_{\mathrm{IN}})-E_0$.    

By the same logic, after the last thermalization (which ends stage~(II) and begins stage~(III)), the excited level is occupied with probability $p_{\mathrm{OUT}}$. 
As a result, the system remains in this state throughout stage~(III), giving us $-W_\mathrm{(III)}=E_0-E(p_{\mathrm{OUT}})>0$. 

Hence, with probability $p_{\mathrm{IN}}\, p_{\mathrm{OUT}}$, we obtain a strictly positive work loss in stages~(I) and~(III), whose exact value is
\begin{equation}
    -W_\mathrm{(I)} - W_\mathrm{(III)}
    =
    E(p_{\mathrm{IN}}) - E(p_{\mathrm{OUT}}).
\end{equation}
	
	The only way for this work to be recovered in stage (II) is the following: during every thermalization, the state shall end up at the excited level. 
	We can express this in terms of random variables, say $s_i$, $i\in\{1,\ldots,n_{\mathrm{(II)}}\}$, each of which takes the value $0$ if the ground state is occupied, and $1$ if the excited state is occupied. 
	Being more precise, any $s_i$ has the Bernoulli distribution with parameter $g(E_i)$, where $g(E_i)\leq p_{\mathrm{OUT}}$, by assumption.  
	Moreover, note that, due to thermalizing in each step (of stage (II)), the random variables $s_i$, $i\in\{1,\ldots, n_{\mathrm{(II)}}\}$, are independent. 
	The situation in which no work is lost in total (i.e., when the excited state is occupied in every step) is described by the event
$\{s = n_{\mathrm{(II)}}\}$, where $s = \sum_{i=1}^{n_{\mathrm{(II)}}} s_i$. 
Its probability can be easily estimated by
    %Such an event, however, has a vanishingly low probability (in the case $n_{\text{II}}$ going to infinity). Indeed, since $E_{\sigma}\leq E_i\leq E_{\rho}$, we obtain
	\begin{equation}\label{eq:prob_s=n}
		\text{P}\left(\left\{s=n_\mathrm{(II)}\right\}\right) = \prod_{i=1}^{n_\mathrm{(II)}} g(E_i) \leq \left(p_{\mathrm{OUT}}\right)^{n_{\mathrm{(II)}}}.
	\end{equation}
 For a finite number of steps $n_{\mathrm{(II)}}$, the extracted work is smaller than the required amount by at least
\begin{equation}
\Delta E_i = \frac{E(p_{\mathrm{IN}}) - E(p_{\mathrm{OUT}})}{n_{\mathrm{(II)}}},
\qquad i \in \{1,\ldots,n_{\mathrm{(II)}}\},
\end{equation}
with probability greater than or equal to
\begin{equation}
    p_{\mathrm{IN}}\, p_{\mathrm{OUT}}
    \left(1 - \left(p_{\mathrm{OUT}}\right)^{n_{\mathrm{(II)}}}\right).
\end{equation}
This follows from the fact that in at least one of the $n_{\mathrm{(II)}}$ steps the ground level is occupied instead of the excited one.

In the limiting case $n_{\mathrm{(II)}} \to \infty$, this probability increases further. However, at the same time the quantity $\Delta E_i$ tends to zero.
Consequently, in this regime it becomes necessary to evaluate the probability that the system occupies the ground state in a non-vanishing fraction of steps, rather than in just a single step. 
We will therefore evaluate the probability that, in at least three quarters of the steps of stage~(II), the ground state is occupied, which is formally described by the event $\{s \leq n_{\mathrm{(II)}}/4\}$. 
Note that, since all the steps are the same length, such an event means that the work gain along stage~(II) is at most one quarter of the total energy change, i.e.,
\begin{equation}
    W_\mathrm{(II)}\leq\frac{1}{4}\left(E\!\left(p_{\mathrm{IN}}\right) - E\!\left(p_{\mathrm{OUT}}\right)\right),
\end{equation}
which means that at most one quarter of the work lost in stages ~(I) and ~(III) can potentially be recovered, resulting in
\begin{equation}
    -W=-W_\mathrm{(I)}-W_\mathrm{(II)}-W_\mathrm{(III)}\geq \frac{3}{4}\left(E\!\left(p_{\mathrm{IN}}\right) - E\!\left(p_{\mathrm{OUT}}\right)\right).
\end{equation}
Now, to bound the probability of such an event from below, it is convenint to consider the auxiliary random variable 
$\tilde{s}=\sum_{i=1}^{n_{\mathrm{(II)}}} \tilde s_i$, where the variables $\tilde s_i$ are i.i.d.\ Bernoulli random variables with parameter $p_{\mathrm{OUT}}$. 
Since $g(E_i)\leq p_{\mathrm{OUT}}$ for all $i\in\{1,\ldots, n_{\mathrm{(II)}}\}$, it follows that
\begin{equation}
    \text{P}\!\left(s \leq \frac{n_{\mathrm{(II)}}}4\right) \;\geq\; 
    \text{P}\!\left(\tilde{s} \leq \frac{n_{\mathrm{(II)}}}4\right).
\end{equation}
We now apply Hoeffding's inequality (cf.\ \cite[Theorem 1]{Hoeffding}) in the following form: if $X_1,\dots,X_n$ are independent Bernoulli random variables with mean $\mathbb{E}[X_i]=p$, then for any $\delta>0$,
\begin{equation}
    \text{P}\!\left(\frac{1}{n}\sum_{i=1}^n X_i - p \geq \delta \right)
    \leq e^{-2n\delta^2}.
\end{equation}
In our case, setting $p=p_{\mathrm{OUT}}$ and $\delta = \frac{3}{4} - p_{\mathrm{OUT}}$ (which is positive by assumption), we obtain
\begin{align}
    \text{P}\!\left(\tilde{s} \geq \frac{3n_{\mathrm{(II)}}}{4}\right)
    &= \text{P}\!\left(\frac{\tilde{s}}{n_{\mathrm{(II)}}} - p_{\mathrm{OUT}} \geq \frac{3}{4} - p_{\mathrm{OUT}} \right)
    \\ \nonumber &\leq e^{-2n_{\mathrm{(II)}}\left(\tfrac{3}{4}-p_{\mathrm{OUT}}\right)^2}.
\end{align}
Taking the complement, we arrive at
\begin{align}
    \text{P}\!\left(\tilde{s} \leq \frac{n_{\mathrm{(II)}}}{4}\right)
    \geq 1 - e^{-2n_{\mathrm{(II)}}\left(\tfrac{3}{4}-p_{\mathrm{OUT}}\right)^2},
\end{align}
whence:
	\begin{equation}
    \begin{split}
	     \text{P}\left(\left\{s\leq \frac{n_{\mathrm{(II)}}}4\right\}\right)
	     &\geq 1-e^{-2n_{\mathrm{(II)}}\left(\frac{3}{4}-p_{\mathrm{OUT}}\right)^2}\\
	     &\geq 1-e^{-2\left(\frac{3}{4}-p_{\mathrm{OUT}}\right)^2}.
	\end{split}
    \end{equation}
	It implies that the probability of losing at least $3[E(p_{\mathrm{IN}}) - E(p_{\mathrm{OUT}})]/4$ of the work (in all the three stages) is greater than or equal to 
	\begin{equation}
	\label{eq:lower_bound}
	    p_{\mathrm{IN}}\,p_{\mathrm{OUT}}\left(1-e^{-2\left(\frac{3}{4}-p_{\mathrm{OUT}}\right)^2}\right).
	\end{equation}
	The proof is therefore complete.
\end{proof}

To summarize, in this section we have shown that -- for some particular protocol -- work is  lost (with some fixed positive probability) if we care not only about averages, but also about individual realizations. 
Note that we have assumed that all steps in the protocol are of equal length, i.e., in each step the energy is shifted by the same amount (and that each step entails a full thermalization). If one wishes to extend this result to the case of non-equal steps, one needs to ensure that the event $\{s \geq 3n_{\mathrm{(II)}}/4\}$ does not concentrate on a subset of steps of negligible length.  
While this can indeed be shown, rather than proving it directly, we proceed in the following sections to a general argument that applies to arbitrary protocols (including the possibility of partial thermalizations as well), not only to the specific one considered here.
		
% 		{\color{teal}[Hania: Słuchajcie, dodałam powyżej wątek częściowych termalizacji, ponieważ nasz szczególny protokół zakłada pełne termalizacje w każdym kroku, co pozwala traktować proces jako pojedynczą ścieżkę. Tymczasem w głównym dowodzie kluczowe jest ujęcie procesu jako mieszanki ścieżek, dlatego wydaje mi się, że warto to tutaj przynajmniej krótko zaznaczyć. Być może ten paragraf wymaga jeszcze dopracowania redakcyjnego i nieco lepszej narracji.
% ]}
\section{Arbitrary energy change cycle.}
\label{sec:arbitrary-cycle}
	
In this section, we will drop all the restrictions concerning the changes in energy levels.
Namely, these changes and the related Gibbs curve can now be completely arbitrary as long as they satisfy a single condition that the initial energy equals the final one.
The latter constraint means that the blue curve in Fig.\ \ref{fig:Gibbs-general} starts and ends at the same value. 
%\marginpar{This paragraph is one sentence! \cred MH: actually two :)  Now I made them three. \blk}

\begin{figure}
	\centering
	\includegraphics[width=\linewidth]{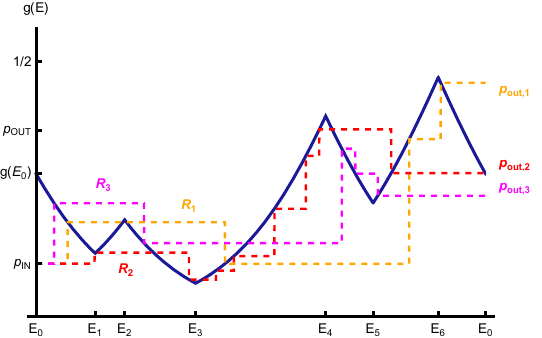}
	\caption{
	    \label{fig:Gibbs-general}
The Gibbs curve (solid), together with examples of different possible paths $R_1, R_2, R_3$ (dashed), where thermalization occurs at different moments along the process. 
Note that the $x$-axis represents the evolution of the process with energy levels changed in the order
$E_0 \rightarrow E_1 \rightarrow \ldots \rightarrow E_0$. 
While the initial states of all three paths are equal to $\rho = (1-p_\mathrm{IN}, p_\mathrm{IN})$, their final states may differ. The only requirement is that the appropriate mixture corresponding to the decomposition of the whole process, $\sum_{j \in I} \gamma_j R_j$, 
coincides with the target state $\sigma = (1-p_\mathrm{OUT}, p_\mathrm{OUT})$, i.e., 
$p_\mathrm{OUT} = \sum_{j \in I} \gamma_j p_{\mathrm{out},j}$. 
(For a more detailed explanation, see Appendix~\ref{sec:appendix_path}.)
}
\end{figure}

%\begin{figure}
%	\centering
%	\includegraphics[width=\linewidth]{figures/TailoredPaths.pdf}
%	\caption{
%	    \label{fig:Gibbs-general}
%\textcolor{magenta}{OLD image to be removed after verification of the new one. }The Gibbs curve (solid), together with examples of different possible paths $R_1, R_2, R_3$ (dashed), where thermalization occurs at different moments along the process. 
%Note that the $x$-axis represents the evolution of the process with energy levels changed in the order
%$E_0 \rightarrow E_1 \rightarrow \ldots \rightarrow E_0$. 
%While the initial states of all three paths are equal to $\rho = (1-p_\mathrm{IN}, p_\mathrm{IN})$, their final states may differ. The only requirement is that the appropriate mixture corresponding to the decomposition of the whole process, $\sum_{j \in I} \gamma_j R_j$, 
%coincides with the target state $\sigma = (1-p_\mathrm{OUT}, p_\mathrm{OUT})$, i.e., 
%$p_\mathrm{OUT} = \sum_{j \in I} \gamma_j p_{\mathrm{out},j}$. 
%(For a more detailed explanation, see Appendix~\ref{sec:appendix_path}.) \textit{\color{red}[HW-Ś: Trzeba zamienić $p^\mathrm{in}$, $q^\mathrm{out}$ na $p_\mathrm{IN}$, $p_\mathrm{OUT}$, dopisać numery ścieżek $R_1,R_2,R_3$ i oznaczyć ich końce $p_{\mathrm{out},1}, p_{\mathrm{out},2}, p_{\mathrm{out},3}$.]}}
%\end{figure}

% We shall prove that most transitions cannot be performed without work expenditure.  
% Then it will be very quick to see that all other transitions are achievable by $\mcoww$, apart from a single one - the one starting with pure ground stare and ending with state more excited than the Gibbs state.

\subsection{The main tools used in the proof}
\label{sec:tools}

The central tool underlying our entire proof is the notion of {paths}.

\textbf{Paths.} 
Recall that at each step of the process one applies a thermalization map
\[
\mathcal{C}_i = (1-\lambda_i)\,\ONE + \lambda_i\,\hat{\tau}_i,
\]
where $\hat{\tau}_i$ is a constant map replacing the state with the thermal state
$\tau_i$ corresponding to the Hamiltonian $H(E_i)$, and $\ONE$ denotes the identity map.

We keep track of both the maps applied and their associated probabilities.
For instance, in a two-step process $\mathcal{C}_2 \mathcal{C}_1$, one of the four maps
$\ONE \ONE$, $\ONE \hat\tau_1$, $\hat\tau_2 \ONE$, $\hat\tau_2 \hat\tau_1$ 
is applied, with respective probabilities
$(1-\lambda_2)(1-\lambda_1)$, $(1-\lambda_2)\lambda_1$, $\lambda_2 (1-\lambda_1)$ and $\lambda_2\lambda_1$ 
(though the exact values of these probabilities will not play an important role later.)

After $N$ steps, the process is therefore a mixture of various sequences of maps
$\ONE$ and $\hat{\tau}_i$. 
Representing $\hat{\tau}_i$ by the symbol $G$ (for ``Gibbs''), such sequences take the form
$\ONE GGG\ONE G\ONE \ONE GG \ldots$, 
where $\ONE$ indicates that the identity map was applied, and $G$ that a full thermalization occurred. 
Each individual sequence produces some final state $(1-p_{\mathrm{out}}, p_{\mathrm{out}})$;
however, the mixture of final states arising from all such sequences must coincide with the target state
$\sigma = (1-p_{\mathrm{OUT}}, p_{\mathrm{OUT}})$. 
Importantly, every sequence in the mixture starts from the same initial state
$\rho = (1-p_{\mathrm{IN}}, p_{\mathrm{IN}})$. 

This motivates the introduction of the notion of a~``path'' (see Fig.~\ref{fig:Gibbs-general}).
A path is specified by a sequence of the form
\begin{equation}
\Delta E_1, X_1, \Delta E_2, X_2, \ldots, \Delta E_N, X_N,
\end{equation}
where $\Delta E_i$ are the energy increments associated with the level transformations,
and $X_i \in \{\ONE, G\}$ denotes whether the $i$-th step is the identity or a full thermalization
($\hat{\tau}_i$).

Since the initial energy is $E_0$, we may equivalently describe a path by the sequence of energies
$\{E_i\}_{i=1}^N$ instead of the increments $\Delta E_i$.
Because the process forms a cycle, the final energy must satisfy
$E_N = E_0$.

Let us draw the reader’s attention to the fact that a~path can be represented by a polygonal curve, consisting of horizontal segments of lengths $|\Delta E_i|$ and vertical segments corresponding to thermalizations, which occur only when $X_i = G$; see Fig.~\ref{fig:Gibbs-general}.  
In particular, a~subsequence of the form
$E_1,\, \ONE,\, E_2,\, \ONE,\, \ldots,\, E_k,\, \ONE$
is represented by a single horizontal line composed of $k$ shorter segments.

\textbf{Gibbs curve.}
For us, it will be important to consider, along with the path, the ``Gibbs curve'', which is the curve whose value is $g(E)$, i.e., the population of the upper level for Gibbs state at energy $E$. 
In Fig.\ \ref{fig:Gibbs-general}, the Gibbs curve is the solid one. 
The curve has a property that the area below it over a given interval is given, up to a sign, by the difference of free energy of Gibbs states at the beginning and at the end of the interval.

\textbf{Path Shrinking.}
\label{sec:shrinking}
Once we have defined a path and the Gibbs curve, it is convenient to simplify both without affecting the generality of our considerations. 
The idea of \emph{Path Shrinking} is illustrated in Fig.~\ref{fig:shrink}, and, briefly speaking, amounts to removing the $\ONE$ symbols from the sequence determining the path.
\begin{figure}
    \subfloat[Original]{\includegraphics[width=0.66\linewidth]{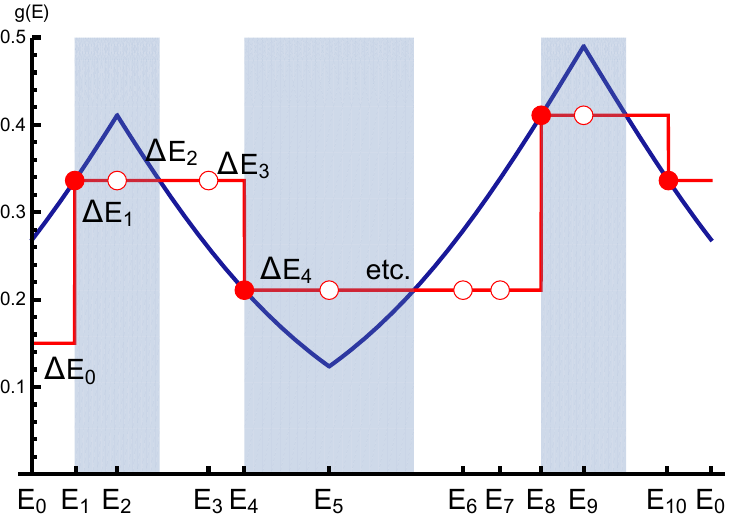}}
    \subfloat[Simplified]{\includegraphics[clip, trim=0 0 {175} 0, width=0.34\linewidth]{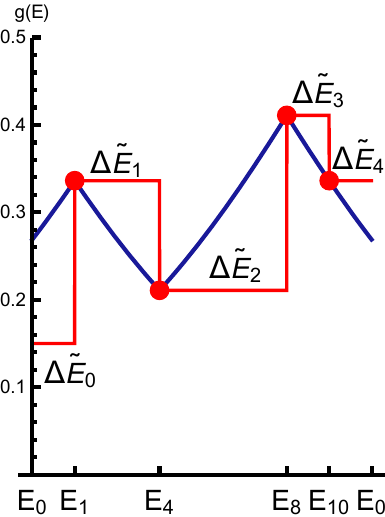}}
    \caption{
    \label{fig:shrink}
Simplifying paths (\emph{Path Shrinking}).
(a) A path consists of level transformations that change the energy by $\Delta E_i = E_{i+1}-E_i$, followed at each stage either by a thermalization (solid dot) or by no thermalization (white dot).
Sequences of consecutive level transformations without thermalizations are then combined into a single level transformation with the appropriate energy difference $\Delta \tilde{E}_i$.
Finally, unnecessary up--down excursions (shaded regions) are eliminated, without changing $\Delta \tilde{E}_i$.
(b) The resulting simplified process, shown with the original $x$-axis energy values retained for illustration.
    }
\end{figure}

Namely, we observe that in any path, a sequence of, say, $k$ consecutive $\ONE$’s means that no thermalization occurs between the corresponding level transformations in this part of the path.
Therefore, we can glue these level transformations together and obtain a single level transformation with energy increment $\sum_i \Delta E_i$, equal to the sum of the increments in this sequence.
Hence, without loss of generality, we may restrict attention to paths that contain no $\ONE$’s and consist solely of level transformations interlaced with full thermalizations.

In some cases, this procedure results in a level transformation in which the energy is first raised and then lowered (or vice versa). 
As illustrated in Fig.~\ref{fig:shrink}a, the Gibbs curve -- which encodes the changes of the energy levels -- goes up and then down, while the path itself contains no thermalization in between. If we now remove from the curve this round trip, it will 
clearly not affect the 
work random variable. 
Altogether we obtain a shrunk path, as depicted 
in Fig.~\ref{fig:shrink}b.
\blk
 % modification does not affect the work random variable. 
% This procedure corresponds precisely to shrinking the path, as depicted in Fig.~\ref{fig:shrink}b.

Thus, after applying the \emph{Path Shrinking} procedure, our paths never cross the Gibbs curve as a straight line. 
The only crossings are at thermalization points, and at such crossings, the path changes from vertical to horizontal. 

Also, after \emph{Path Shrinking}, paths correspond to sequences $E_1, G, E_2, G, E_3, G, \ldots, G, E_N$. 

\textbf{Division of a path into three stages.} 
There is a~natural division of any given path into three stages:
\begin{itemize}[align=left]
    \item[(I)] before the first thermalization,
    \item[(II)] between the first and the last thermalization,
    \item[(III)] after the last thermalization.
\end{itemize}

Note that, in the case of a path containing solely a single thermalization, only stages~(I) and~(III) are nontrivial, whereas stage~(II) is degenerate. On the other hand, the above decomposition is not applicable to paths without any thermalization. These special cases will therefore be discussed separately (at the end of Section \ref{subsec:co5}).

\subsection{Preparations for the formal proof}
\label{subsec:preparations}

For greater clarity, we shall first determine which transitions cannot be implemented within $\mcoww$ when operations from \textbf{CO5} are not allowed (that is, when an operation $\Lambda$ in Definition~\ref{def:MCO_0} of $\mcoww$ is composed only of operations from the sets \textbf{CO2} and \textbf{CO3}, excluding those from \textbf{CO5}). 
%At the end of this section, 
In Section \ref{subsec:co5} we then outline how the reasoning extends once operations from \textbf{CO5} are included.

As already mentioned above, the idea is to first analyze the possible behaviour of individual paths (already appropriately shortened by the \emph{Path Shrinking} procedure). Let us recall that the final state of a single path is denoted here by $(1 - p_\mathrm{out}, p_\mathrm{out})$ (to distinguish it from the final state of the whole process, denoted by $(1 - p_\mathrm{OUT}, p_\mathrm{OUT})$).  
The entire proof strategy is then based on considering separately two mutually exclusive scenarios:
\begin{itemize}
    \item[(a)] which assumes that \textbf{a single path ends above $g(E_0)$}, i.e.\ $g(E_0) < p_{\mathrm{out}}$;
    % guaranteeing that $-W_\mathrm{(III)}>0$ with some positive probability (cf.\ Fig.\ \ref{fig:cases_with_primes});
    \item[(b)] which assumes that \textbf{a single path ends below $g(E_0)$}, i.e.\ $p_{\mathrm{out}} < g(E_0)$.
    %covering settings (iii) and (iv) in \eqref{cases-i-iv} (for a single path $R$), which case shall be further examined within two (not disjoint!) sub-cases: 
%    \begin{itemize}
%        \item[(b1)]  $p_{\mathrm{out}}<q_0$ and $q_1<q_0$, guaranteeing that $-W_\mathrm{(I)}>0$ with some positive probability (cf.\ Fig.\ \ref{fig:cases_with_primes});
%        \item[(b2)] $p_{\mathrm{out}}<q_0$ and $p_{\mathrm{out}}<q_1$, guaranteeing that $-W_\mathrm{(II)}>0$ with some positive probability (cf.\ Fig.\ \ref{fig:cases_with_primes}).
%    \end{itemize}
\end{itemize}

\textbf{Scenario (a).}  
In scenario~(a), a work loss greater than some positive constant is incurred in stage~(III) with probability $p_\mathrm{out}$ (the reasoning is analogous to that in the proof of Lemma~\ref{lem:simplecase}). Moreover, with probability bounded away from zero, this work loss is not compensated by the contributions from stages~(I) and~(II). To prove this, we adopt the following strategy.

In stage~(I), we condition on the event ``0'' (occupation of the ground state) when the first level transformation lowers the energy, and on the event ``1'' (occupation of the excited state) when it raises the energy.  
It is then straightforward to argue that, conditioned on such an event, the amount of lost work (denoted by $-W_{\mathrm{(I)}}$) is no smaller than  the free energy difference $\Delta F_{\mathrm{(I)}}$ of the corresponding Gibbs state during stage~(I). 
This is established in Appendix~\ref{sec:appendA}, Lemma~\ref{lem:w1}.

In stage~(III), we perform the same conditioning.  
We then prove (which is relatively simple) that the amount of lost work (denoted by $-W_{\mathrm{(III)}}$) is, with probability $p_\mathrm{out}$, strictly greater than the free energy difference of the Gibbs state in stage~(III), i.e., $-W_\mathrm{(III)}\geq \Delta F_\mathrm{(III)}+\epsilon_\mathrm{(III)}$ (which value is strictly positive under the assumption $g(E_0) < p_\mathrm{out}$ defining case~(a) and causing the Gibbs curve to decrease along stage~(III)).  
The proof is again given in Appendix~\ref{sec:appendA}, Lemma~\ref{lem:w1}.

Combining the above two observations, we obtain
\begin{equation}\label{total_W(I)(III)}
    -W_{\mathrm{(I)}} - W_{\mathrm{(III)}}
    \geq
    \Delta F_{\mathrm{(I)}} + \Delta F_{\mathrm{(III)}} + \epsilon_{\mathrm{(III)}}.
\end{equation}

As far as stage~(II) is concerned, the key tool is the Jarzynski equality~\cite{Jarzynski1997}, which implies that
\begin{equation}
\label{eq:JarzynskiII}
    \left\langle e^{\beta W_{\mathrm{(II)}}} \right\rangle
    = e^{-\beta \Delta F_{\mathrm{(II)}}},
\end{equation}
where $\Delta F_{\mathrm{(II)}}$ denotes the free energy difference between the Gibbs states corresponding to the initial and final Hamiltonians of stage~(II)$\,$ (this comes from independence of works between each thermalization; see the proof of Lemma~\ref{lem:w2} in Appendix~\ref{sec:appendA} for an explicit derivation within our framework).

This identity, in turn, implies that
for any $\epsilon>0$ there exists $\eta>0$ such that 
\begin{equation}\label{total_W(II)}
    -W_{\mathrm{(II)}} \geq \Delta F_{\mathrm{(II)}} - \epsilon,
\end{equation}
with probability greater than $\eta$.   
We prove this statement in Lemma~\ref{lem:Jarzynski} at the end of the present subsection using a general framework. Its discrete counterpart is formulated in Lemma~\ref{lem:w2} in Appendix~\ref{sec:appendA}.

Now, note that since we are dealing with a cyclic process, i.e., the final Hamiltonian coincides with the initial one, we have
\begin{equation}
    \label{eq:total-deltaF}
    \Delta F_{\mathrm{(I)}} +
    \Delta F_{\mathrm{(II)}} +
    \Delta F_{\mathrm{(III)}} = 0
\end{equation}
(recall that this is a change of free energy of Gibbs states related to temporary Hamiltonian, rather than the actual states).

Finally, combining Eqs.~\eqref{total_W(I)(III)}, \eqref{total_W(II)}, and~\eqref{eq:total-deltaF}, we obtain
\begin{equation}
    -W_{\mathrm{(I)}} - W_{\mathrm{(II)}} - W_{\mathrm{(III)}}
    > \epsilon_{\mathrm{(III)}} - \delta > 0 .
\end{equation}
This inequality holds with probability equal to the product of the probabilities of the corresponding conditioning events: one arising from stage~(I), one from stage~(III), and the probability  associated with stage~(II) (greater than $\eta$). 
Let us emphasize that, since stages~(I), (II), and~(III) are separated by full thermalizations, the corresponding events are independent. 
Therefore, the probability of the total event is given by the product of the individual probabilities.

\textbf{Scenario (b).} In scenario~(b), no work loss along stage~(III) can be guaranteed with positive probability. 
Fortunately, the condition $p_{\mathrm{out}} < g(E_0)$ provides an opportunity to incur work loss either along stage~(I) or along stage~(II). 
It is therefore convenient to distinguish two subcases of case~(b). 
These subcases are not disjoint in general (although we later make them disjoint in the appendix, mainly for technical reasons; cf.\ Appendix~\ref{sec:idea}).

Before introducing them, let us note that the boundary between the first two stages (I) and (II) is determined by the first thermalization, i.e., by a point $(E_1, g(E_1))$ on the Gibbs curve. 
Depending on the position of this point relative to the end of a given path, that is, $p_\mathrm{out}$, as well as to the point $(E_0, g(E_0))$ (which, within case~(b), is separated from the endpoint $p_\mathrm{out}$ in the sense that $p_{\mathrm{out}} < g(E_0)$), we define the following subscenarios:
\begin{itemize}
    \item[(b1)] the path ends below $g(E_0)$, i.e., $p_{\mathrm{out}} < g(E_0)$, and the first thermalization is applied below the point $(E_0, g(E_0))$, that is, $E_1 > E_0$ and thus $g(E_1) < g(E_0)$;
    \item[(b2)] the path ends below $g(E_0)$, i.e., $p_{\mathrm{out}} < g(E_0)$, and the first thermalization $(E_1, g(E_1))$, which indicates the beginning of stage (II) of the path, lies above the endpoint of this stage $(E(p_{\mathrm{out}}), p_{\mathrm{out}})$, that is,  $p_{\mathrm{out}}<g(E_1)$.
\end{itemize}

In short, the reasoning within scenario~(b1) is analogous to that in case~(a), with the only difference that the roles of stages~(I) and~(III) are interchanged. 
Here, the work loss in the first stage is guaranteed by the assumption $g(E_1) < g(E_0)$, which forces the Gibbs curve to be decreasing along this segment.

On the other hand, in scenario~(b2) we cannot guarantee any work loss in either stage~(I) or stage~(III). 
Instead, we observe that the free energy increase\ increment within stage~(II) is positive, i.e., $\Delta F_{\mathrm{(II)}} > 0$ (since the starting and ending points of stage~(II), namely $(E_1, g(E_1))$ and $(E(p_{\mathrm{out}}), p_{\mathrm{out}})$, respectively, satisfy $p_{\mathrm{out}} < g(E_1)$). 
Therefore, recalling the already established consequence of the Jarzynski identity, namely Eq.~\eqref{total_W(II)}, we conclude that work loss is incurred in stage~(II) with some positive probability (bounded away from zero).

Simultaneously, if we again apply appropriate conditioning in stages~(I) and~(III) (this time only on the event ``0''), we obtain, with a fixed positive probability (greater than or equal to $(1-p_\mathrm{in})(1-p_\mathrm{out})$),
\begin{equation}\label{eq:total_W(I)(III)_simple}
    -W_{\mathrm{(I)}} - W_{\mathrm{(III)}} \ge 0 .
\end{equation}

Combining Eqs.~\eqref{total_W(II)}, \eqref{eq:total-deltaF}, and
\eqref{eq:total_W(I)(III)_simple}, we finally obtain that
\begin{equation}
    -W_{\mathrm{(I)}} - W_{\mathrm{(II)}} - W_{\mathrm{(III)}}
    > \Delta F_{\mathrm{(II)}} - \epsilon
\end{equation}
with probability bounded away from zero (which again follows from the independence of the events in the individual stages~(I), (II), and~(III), ensured by full thermalizations separating them). 
Since $\epsilon > 0$ can be chosen arbitrarily, this implies that the total work loss of the entire path is strictly positive with such a~probability.

\textbf{From a single path to the whole process.}
As already mentioned above, in our approach, the idea is to focus on a concrete path and show that, idependetly of which of the two disjoint  scenarios (a) or (b) realizes, the work loss is incured (more precisely, $-W$ is greater than some positive constant), with some probability bounded away from zero. 
Then, since the mixture of the ends of these paths must be equal to $p_{\mathrm{OUT}}$ (as the final state is to be $\sigma$), by Markov's inequality, we will get that the overall process must cost work with probability bounded away from zero.\newline
	
% \textcolor{magenta}{[MS: I have read the main text up to this point and it seems fine to me in general. I will look at the rest after reading Appendix B.]}

At the end of this section we provide the promised lemma.
\begin{lemma}
\label{lem:Jarzynski}
Let $\epsilon>0$. Then
\begin{align}
\label{eq:lemJarz}
P(-W_\mathrm{(II)}\geq \Delta F_\mathrm{(II)}-\epsilon)> 1-e^{-\beta\epsilon}\equiv \eta    
\end{align}
\end{lemma}
\begin{proof}
Let us denote \(W = -W_{\mathrm{(II)}}\) (which matches the convention commonly used in the Jarzynski equality) and \(\Delta F = \Delta F_{\mathrm{(II)}}\). Then Eq.~\eqref{eq:JarzynskiII} can be rewritten as
\begin{align}
    \left\langle e^{-\beta W} \right\rangle
    = e^{-\beta \Delta F}.
\end{align}
We then write
\begin{align}
    P(W \leq \Delta F - \epsilon)
    =
    \int_{-\infty}^{\Delta F - \epsilon} \rho(W)\, dW,
\end{align}
where \(\rho(W)\) is the probability density function of \(W\) (in the discrete case, it takes the form of a sum of Dirac delta distributions). 
Now, for all \(W\) within the above integration range, we have
\begin{align}
    e^{\beta (\Delta F - \epsilon - W)} \geq 1.
\end{align}
Hence,
\begin{align}
    P(W \leq \Delta F - \epsilon)
    &=
    \int_{-\infty}^{\Delta F - \epsilon} \rho(W)\, dW \nonumber \\
    &\leq
    \int_{-\infty}^{\Delta F - \epsilon}
    e^{\beta (\Delta F - \epsilon - W)} \rho(W)\, dW
    \nonumber \\
    &=
    e^{\beta (\Delta F - \epsilon)}
    \int_{-\infty}^{\Delta F - \epsilon}
    e^{-\beta W} \rho(W)\, dW
    \nonumber \\
    &\leq
    e^{\beta (\Delta F - \epsilon)}
    \int_{-\infty}^{\infty}
    e^{-\beta W} \rho(W)\, dW
    \nonumber \\
    &=
    e^{\beta (\Delta F - \epsilon)}
    \left\langle e^{-\beta W} \right\rangle
    =
    e^{-\beta \epsilon},
\end{align}
where in the last equality we used Eq.~\eqref{eq:JarzynskiII}. 
We have thus shown that 
\begin{align}
    P(W    \leq \Delta F- \epsilon) \leq  e^{-\beta \epsilon}
\end{align}
or equivalently 
\begin{align}
    P(W >\Delta F- \epsilon) 
    >1- e^{-\beta \epsilon},
\end{align}
which of course implies 
\begin{align}
    P(W \geq \Delta F- \epsilon) 
    > 1- e^{-\beta \epsilon}.
\end{align}
Returning to the original notation, we obtain the desired relation~\eqref{eq:lemJarz}.
\end{proof}

\subsection{Adding CO5}
\label{subsec:co5}
Let us now outline the fully general case in which swaps are allowed. 
As discussed in Section~\ref{subsec:operations}, without loss of generality, we may assume that after each unitary transformation the system is dephased. Such an operation can then be represented as a mixture of the identity operation and the swap operation (i.e., the bit-flip operation). 

Consequently, an arbitrary path is now determined by a sequence composed of three possible types of elements:
\begin{itemize}
    \item Level Transformations (represented by horizontal segments),
    \item Thermalizations (represented by vertical segments ending on the Gibbs curve),
    \item Swaps (represented by vertical segments ending above the Gibbs curve at the point $(0,1/2)$, corresponding to a reflection about the point $1/2$).
\end{itemize}

Thus, a path is described by a sequence of the form
\begin{equation}
    \Delta E_1, X_1, \Delta E_2, X_2, \ldots, \Delta E_n, X_n,
\end{equation}
where, as before, $\Delta E_i$ are the energy increments associated with the level transformations, while each $X_i$ is either $\ONE$, $G$ (i.e., $\hat{\tau}(E_i)$), or $S$, representing the $\swap$ operation.

Similarly as before, we may apply the \emph{Path Shrinking} procedure, which now consists not only in removing $\ONE$ operations, but also in eliminating pairs of $S$ operations, since the swap is an involution.

Hence, we may restrict attention to paths for which each $X_i$ is either $G$ or $S$. 
An example of such a path is depicted in Fig.~\ref{fig:path-swaps}.
\begin{figure}
	\centering
	\includegraphics[width=0.7\linewidth]{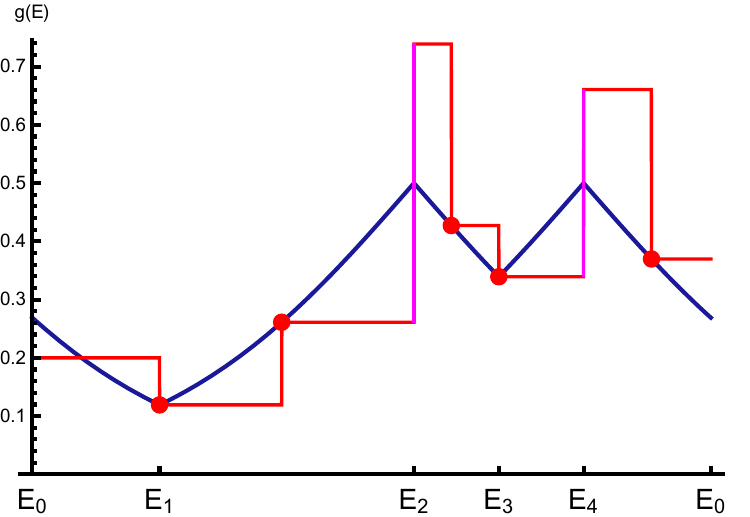}
	\caption{
    	\label{fig:path-swaps}
    	Path with swaps (indicated in pink). Note that swaps can occur only when the Gibbs curve attains the value $1/2$, i.e., when the energy is zero (at $E_2$ and $E_4$ in the example).
	}
\end{figure}

As before, the path can be represented by a polygonal curve. 
Again, $\Delta E_i$ means that there is a horizontal part, $G$ means that there is a vertical part ending up at a Gibbs curve, and $S$ means \textit{reflection} about point $1/2$. 
This last element can only appear when also Gibbs curve is at this point, i.e., when the energy of the upper level is $0$.  

To see that $\swap$ acts as a reflection, note that it transforms the state $(1-q,q)$ into $(1-q',q')$, where $q' = 1-q$. Hence, the values $q$ and $q'$ are symmetrically positioned with respect to $1/2$. More precisely, we have
\begin{equation}
    \frac{q+q'}{2}=\frac{q+(1-q)}{2}=\frac{1}{2}.
\end{equation}

Under a swap operation (similarly as under thermalization), the path may move either upward or downward; we shall loosely refer to these situations as \emph{swap up} and \emph{swap down}. Which of the two occurs depends on the state of the system immediately before the swap. 
More precisely, if the state is less excited than the corresponding Gibbs state (i.e., the path lies below the Gibbs curve), then the operation $\swap$ causes a \emph{swap up}: the path ends up above the Gibbs curve (in fact, above the level $1/2$). Conversely, if the path lies above the Gibbs curve (and therefore also above $1/2$), then a \emph{swap down} moves the path below the Gibbs curve. 
As already mentioned, this behaviour is analogous to the effect of thermalization, which causes a jump upward when the state lies below the Gibbs curve and a~jump downward when it lies above. The crucial difference is that thermalization moves the path directly onto the Gibbs curve, whereas $\swap$ produces a symmetric reflection about the level $1/2$.\blk

We shall now argue that 
a \emph{swap down} can occur exclusively in part~(I) of the path, i.e., before the first thermalization. Recall that a \emph{swap down} means that the swap operation moves the path from above the level $1/2$ to below it. Hence, immediately before such a swap, the path must lie above $1/2$. 
This can happen in only two ways. 

First, the path may initially start above the Gibbs curve (the initial state may have a population
of excited level greater than 1/2, i.e., $p_\mathrm{IN}>1/2$), before any thermalization is applied. Since the first thermalization always brings the path onto the Gibbs curve, and thus below or exactly to the level $1/2$, this possibility can occur only before the first thermalization.

Second, the path may previously undergo a \emph{swap up}. However, this is excluded after applying the \emph{Path Shrinking} procedure. Indeed, since the swap operation is an involution, multiple swaps between the same pair of thermalizations reduce to at most one effective swap. Consequently, any two swaps must be separated by a thermalization, and every thermalization brings the path below (or exactly to) the level $1/2$. Hence, after a \emph{swap up}, no subsequent \emph{swap down} can occur before another thermalization.

%Note that, in stage I, i.e., before the first thermalization, we may have a swap down without the previous swap up because the initial state may have a population of excited level greater than  $1/2$. And there can be only a single such swap up (the second one would need to be after some swap down, which is excluded due to \emph{Path Shrinking}). 

Let us also emphasize that, for a swap to appear along stage~(III) (and we have already explained that it can only be a \emph{swap up}), it is necessary that \(p_{\mathrm{out}} > \tfrac12\). In particular, scenario~(b), defined by the condition \(p_\mathrm{out} < g(E_0)\), excludes the possibility of a swap, since the Gibbs curve (and thus also \(g(E_0)\)) never exceeds the level \(1/2\). Indeed, any swap performed after the last thermalization necessarily moves the path above the level \(1/2\), which is possible only in scenario~(a). Moreover, performing such a swap is the only way to obtain a final state satisfying \(p_{\mathrm{out}} > 1/2\) (actually, this is also possible when \(p_\mathrm{in} \geq p_\mathrm{out} > 1/2\), but we do not need to consider this case separately, since it is a special instance of the scenario in which the initial and final states \(\rho\) and \(\sigma\) satisfy Eq.~\eqref{eq:sigma-thm} -- a scenario that allows one to transform \(\sigma\) into \(\rho\) by means of \(\mcoww\)).
Taking the above into account, we obtain the following rules:
\begin{itemize}
    \item There can be at most one swap (either a \emph{swap up} or a \emph{swap down}) in stage~(I).
    \item There can be at most one swap (necessarily a \emph{swap up}) in stage~(III), and this is possible only within scenario~(a), i.e., when $p_\mathrm{out} > q_0$.
    \item In stage~(II), a~\emph{swap down} cannot occur, whereas a~single \emph{swap up} may occur between any two neighboring thermalizations.
\end{itemize}
\blk

The proof now proceeds along the same general lines as before, with some additional technical complications. In particular, scenario~(a), which allows for a swap in stage~(III), requires a more subtle argument.

\textbf{Stages~(I) and~(III).}
Concerning stages~(I) and~(III), we proceed analogously to the case without swaps. We again use the rule that we condition on the event ``0'' when the Gibbs curve increases, and on the event ``1'' when it decreases. The idea behind this conditioning is that we always condition on the most ``work-losing'' event: when the energy level is raised, conditioning on ``1'' yields the maximal possible work loss (equal to the energy change), whereas when the level is lowered, conditioning on ``0'' yields zero work gain, i.e., the minimal possible contribution.

The effect of swaps will be discussed in detail in Appendix~\ref{sec:swaps(I)(III)} (cf. Lemmas \ref{lem:easy_swaps}-\ref{lem:w1_swaps}). Here we present only an illustrative example, showing how the presence of swaps modifies the reasoning in the proof. The example is depicted in Fig.~\ref{fig:swap-cond01}.

\begin{figure*}
	\centering
    \subfloat[Illustration by means of levels]{
        \includegraphics[width=0.95\linewidth]{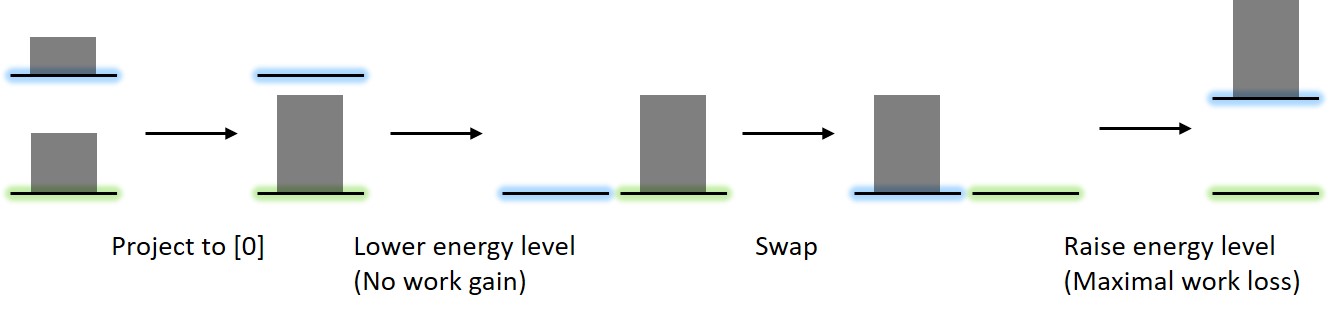}
    }
	\newline
    \subfloat[Illustration by means of a path]{
        \includegraphics[width=0.5\linewidth]{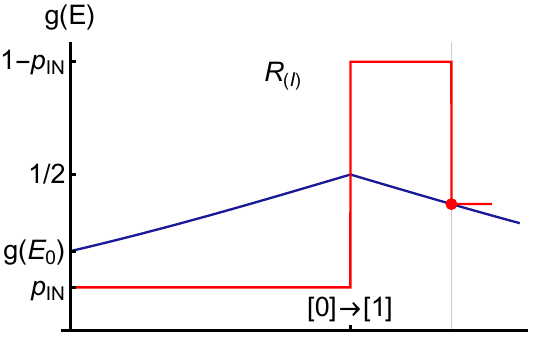}
    }
	\caption{
		\label{fig:swap-cond01}%
        Conditioning on ``0''  in the \emph{swap up} case along stage (I).
	% \mh{ja bym tylko zrobil to na szerokosc dwoch kolumn i powiekszył(b) 1.5 raza} 
 %    \textcolor{magenta}{[MS: czy tak będzie ok czy jeszcze zmieniać rozmiar (b)?]}
    }
\end{figure*}

Suppose that we lower the level to zero, perform a~swap, and then raise the level. 
Then, we see that conditioning on ``0'' gives the largest work loss: when the level is lowered, the work is zero (the smallest possible gain) as, due to our conditioning, the upper state is not populated. 
Then a~swap occurs, meaning that the upper level is now fully populated, and during raising the level, maximal work loss is incurred (equal to energy change).

Thus conditioning on ``0'' changes after a swap into conditioning on ``1'', so during the sequence lowering-swap-raising, we are all the time in the regime of largest possible work loss. 
	
\textbf{Stage (II).} 
Here it suffices to observe that all the conclusions concerning stage~(II), obtained in Appendix~\ref{sec:appendA} (cf. Lemma \ref{lem:w2}) %by means of the Jarzynski identity 
in the setting without swaps, remain valid after extending the analysis to paths that include swaps. This is established in Lemma~\ref{lem:w2_swaps}, proven in Appendix~\ref{sec:swaps(II)}.\newline\blk

As already announced in Section~\ref{sec:tools}, we also need to briefly discuss what happens when a given path cannot be naturally divided into the three stages~(I), (II), and~(III), due to containing at most one thermalization.

In the case where no thermalization is performed, we either have \(p_\mathrm{out} = p_\mathrm{in}\) (when only level transformations are applied) or \(p_\mathrm{out} = 1 - p_\mathrm{in}\) (when, after applying the \emph{Path Shrinking} procedure, the upper energy level is lowered to zero, a swap is performed, and finally the energy level is raised back to its initial value). The first case is a particular instance of scenario~\eqref{eq:sigma-thm} with \(\lambda = 0\). In the second case, we can guarantee that
\(
    -W \geq E_0
\)
with probability \(1 - p_\mathrm{in}\) (we condition on the event ``0'' while lowering the energy, and then the swap forces conditioning on the event ``1'' while raising the energy).

Let us now discuss an arbitrary path containing a single thermalization, which naturally divides the path into parts~(I) and~(III) (while part~(II) is degenerate). It suffices to observe that the first thermalization is simultaneously the last one, so the path necessarily falls into either scenario~(a) or scenario~(b1). These two scenarios are clearly disjoint and, in the present setting, exhaust all possibilities. Consequently, the previous arguments still apply, and we can guarantee a positive work loss either in stage~(I) or in stage~(III).

\section{Conclusions}

% \textit{\color{red}[HW-Ś: Proszę o zweryfikowanie poprawności tego fragmentu po zmianach dokonanych w artykule. 
% Potrzebna jest tu wiedza i intuicja fizyczna, więc zostawiam to w Waszych rękach ;)
% ]}
% \mh{To jeszcze zrobie w pewnym momencie}

In this paper, we have shown that without the use of a 2-level auxiliary system, the Coarse Operations that do not spend work cannot simulate Thermal Operations. 
It means that commutative Markovian dynamics cannot simulate Thermal Operations, as one bit of memory is needed. 
Thus such a~simulation must be a~hidden Markov process, with a~one-bit space of hidden states.
Note that the weak coupling limit is just an example of such commutative dynamics. Indeed, weak coupling limit leads to Markovian master equation known as Davies equation. 
The evolution according to such equation separates diagonal part of density matrix from off diagonal one.

This conclusion is consistent with the broader picture emerging from recent
studies of memory-assisted thermal processes: ancillary systems can be viewed as
explicit carriers of the memory needed to bridge the gap between Markovian
thermal dynamics and the full set of transitions allowed by Thermal Operations
\cite{CzartowskiOliveiraKorzekwa2023_thermal_recall}. Our contribution is to
show that, for Coarse Operations, this memory is not merely a technical
convenience but is necessary already for certain diagonal qubit transitions.

% We would like to emphasize that fully quantum Markovian dynamics in the case of nontrivial Hamiltonian and the two-level system is effectively noncommutative, as the evolution of the diagonal is decoupled from the evolution of the coherences.
% It is a~superselection rule that comes from energy conservation. 

In \cite{KorzekwaL2020-TO-advantage}, it was shown that abstract quantum Markovian dynamics  can simulate Thermal Operations. 
We have thus here established a difference in a continuous dynamics regime, analogous to the difference between Thermal Operations and Gibbs Preserving operations in the quantum operations regime of \cite{Faist_2015}.

% our result establishes a fully quantum advantage in a simulation of Thermal Operations.

An open question remains what if  the initial state
or final state are pure (apart from the case when initial state is excited and final is ground one).
Other possible open questions are about tightening quantitative bounds: namely, what is maximal work loss at a given probability of the loss? 
One can also consider other measures of the loss, like the average of the work over the region where work is negative, and minimize such value over all possible processes. 

Finally, one may consider an explicit battery and allow coherences to build up in the intermediate times while measuring the battery only at the end. 
It would be an even more general setup, where one may try to prove (or disprove) a similar no-go result.

\begin{acknowledgments}
    The work is partially supported by John Templeton Foundation through grant No.\ 56033.
	The work of \hbox{HW-\'S} is supported by the Foundation for Polish Science (FNP) under grant number POIR.04.04.00-00-17C1/18.
    MS and KH acknowledge National Science Centre, Poland, grant Sonata Bis 5, 2015/18/E/ST2/00327. 
	MS acknowledges National Science Centre, Poland, grant SHENG 1, 2018/30/Q/ST2/00625.
    MS acknowledges National Science Centre, Poland, grant OPUS 25, 2023/49/B/ST2/02468.
    KH and MH acknowledge National Science Centre, Poland, grant OPUS 9.\ 2015/17/B/ST2/01945.
	KH and MH acknowledge partial support by the Foundation for Polish Science (IRAP project, ICTQT, contract no.\ 2018/MAB/5, co-financed by EU via Smart Growth Operational Programme). 
    KH acknowledges National Science Centre, Poland, grant Opus 25, 2023/49/B/ST2/02468.
	AG was also partially supported by the Polish National Science Centre (NCN) under the Maestro Grant No. DEC-2019/34/A/ST2/00081.
	MH has been also supported by QuantERA II project ExTRaQT grant
No: 2021/03/Y/ST2/00178 that has received funding
from EU’s Horizon 2020.

	%The authors would like to thank 
	%
	%for useful comments.
		
    %Author Contributions: All authors researched, collated, and wrote this paper.
	
    %Competing interests: The authors declare that there are no competing interests.	
	
    %Data sharing is not applicable to this article as no datasets were generated or analyzed during the current study.	 
\end{acknowledgments}

	%----------FIGURE-------
	%-------------------------
	%-------------------------
	%\begin{figure}[h]
	%\begin{center}
	%\includegraphics[scale=0.5]{Energy.pdf}
	%\end{center}
	%\caption[]{Energy as a function of time for Optimal Process. Here $E=1$ , $p=0.1$}
	%\label{fig:Etime}
	%\end{figure}
	%-------END OF FIGURE---
	%-------------------------
	%-------------------------

	%----------FIGURE-------
	%-------------------------
	%-------------------------
	%\begin{figure}[h]
	%\begin{center}
	%\includegraphics[scale=0.5]{Process1.pdf}
	%\end{center}
	%\caption[]{Optimal Process (Blue), Dashed-Black line is the thermal Population at that instant.}
	%\label{fig:OptProcess1}
	%\end{figure}
	%-------END OF FIGURE---
	%-------------------------
	%-------------------------

	%----------FIGURE-------
	%-------------------------
	%-------------------------
	%\begin{figure}[h]
	%\begin{center}
	%\includegraphics[scale=0.5]{Process2.pdf}
	%\end{center}
	%\caption[]{Another Optimal Process (Red), Dashed-Black line is the thermal Population at that instant.}
	%\label{fig:OptProcess2}
	%\end{figure}
	%-------END OF FIGURE---
	%-------------------------
	%-------------------------

%================================================================
% Bibliography
%================================================================
	
%\bibliographystyle{ieeetr}
\bibliography{Ancillarefs}

%================================================================
% Appendices
%================================================================
	
\onecolumngrid
\appendix

\section{Characterization of transitions possible by means of \texorpdfstring{$\mcoww$ (without \textbf{CO5}}{MCO0 without CO5})}\label{sec:appendA}

%The aim here is to provide a full characterization of the set of transitions that can be done with no work cost by Coarse Operations \textbf{CO2} and \textbf{CO3}. 

In this section of the Appendix, we will provide a nearly full characterization -- with one intriguing exception that remains open (when
either the initial state or the final state is the pure ground
state) -- of the set of transitions that can be done with no work cost by Coarse Operations \textbf{CO2} and \textbf{CO3}. 
In the next section, we will expand it by \textbf{CO5}.

\subsection{Notation}\label{sec:App_not}

%\textbf{\textit{\color{teal}[HW-Ś: \newline 
%Currently, we use $P$ to denote probability and $\E$ to denote the expected value (in places where we refer to general random variables or probabilistic rules). Whenever we refer to the average work, we instead use $\left\langle W \right\rangle$, which is standard notation in thermodynamic contexts. I think this is perfectly fine in the thermalization setting, so let us keep it as it is ;) ]} }

%\textcolor{magenta}{[MS: in my opinion we should not use $E$ since we use it for energy and it would be harder to read so I preferred $\E$. Also if I remember correct in some part we use $\left\langle W\right\rangle$ so maybe we should change it to $\mathbf{E}W$ too. For probability I think $P$ is better (and probably $\mathrm{P}$ would be the best). $\mathbb{P}$ seams odd to me, I don't remember seeing this form and it looks like some set theory notation. It is some time used for set of irrational numbers or set of primes. If you agree I could fix this later.]}

For the reader's convenience, we begin by recalling several definitions and pieces of notation used in the main part of the paper.

As we have already mentioned in Section \ref{subsec:system}, we consider a two-dimensional quantum system with Hamiltonian $H$ parameterized as \hbox{$H(E)=\text{diag}(0, E)$}, where $E\geq 0$. 
The initial and final Hamiltonians are given by $H(E_0)$. 
We also assume that the initial state $\rho = (1-p_{\mathrm{IN}},p_{\mathrm{IN}})$ and the final state $\sigma = (1-p_{\mathrm{OUT}},p_{\mathrm{OUT}})$ are diagonal in the Hamiltonian basis. Throughout Appendix~\ref{sec:appendA}, we assume that \(p_\mathrm{IN} \in (0,1)\) and \(p_\mathrm{OUT} \in (0,1/2]\). The case \(p_\mathrm{OUT} \in (1/2,1)\) is analyzed in Appendix~\ref{sec:appendix-with-swaps}, while the cases of pure initial or final states are briefly discussed in the main text (see Section~\ref{sec:main_results}).
Furthermore, we consider an arbitrarily fixed  ambient  inverse temperature $\beta >0$. For any given energy $E$, the partition function $Z_E$ is given by $Z_E = 1 + e^{-\beta E}$, and the thermal state $\tau_E$  is equal to
\begin{equation}
    \tau_E = (1-g(E), g(E))
\end{equation}
with 
\begin{equation}\label{def:g}
    g(E) = \frac{1}{Z_E} e^{-\beta E}.    
\end{equation}
For any given probability $p\in(0,1/2]$, the energy $E(p)$ is defined so that the state $(1-p,p)$ is the Gibbs state with Hamiltonian $H(E(p))$, i.e.,
 \begin{equation}\label{def:E(p)}
     E(p) = -\frac{1}{\beta} \ln\left( \frac{p}{1-p} \right).
 \end{equation}
As we have already observed before, this immediately gives
\begin{equation}
    g(E(p))=p.
\end{equation}

\subsubsection{Free energy and work}
Let us recall that 
 $F(\Psi,E) = \la E \ra_{\Psi} - \beta^{-1} S(\Psi)$ is the Helmholtz Free Energy of the state $\Psi$ (cf.\ Section \ref{sec:av}). 
Additionally, the free energy of the thermal state $\tau_E$ corresponding to its Hamiltonian $H(E)$ is given by the formula 
\begin{equation}
    F(E) \coloneqq F\left(\tau_E,E\right) = g(E) E - \frac{1}{\beta} S(E)
\end{equation}
with $S(E) \coloneqq S\left(\tau_E\right)$ being the entropy of the Gibbs state with Hamiltonian $H(E)$. 
Let us indicate that the difference  $F(E') - F(E)$ can be expressed as
\begin{equation}
\label{eq:deltaTildeF}
    F(E') - F(E) \coloneqq \int\limits_{E}^{E'}g(u) \mathrm{d}u.
\end{equation}

In Section \ref{subsec:operations}, we have discussed the notion of the \textbf{CO3} operation. 
Recalling the notation introduced there, we write $\Delta E_i = E_{i+1} - E_i$ to denote the energy increment in a single \textbf{CO3} operation performed in the $i$-th step, i.e., $H(E_{i+1}) = H(E_i + \Delta E_i)$ for every $i\in\{0,\ldots, N-1\}$. 
We have also introduced
\begin{equation}\label{def:Delta_i}
	\Delta_i \coloneqq |\Delta E_i| = |E_{i+1} - E_{i}|
\end{equation}
and
\begin{equation}\label{def:qi}
	q_i \coloneqq g(E_i)= \frac{\mathrm{e}^{-\beta E_i}}{1 + \mathrm{e}^{-\beta E_i}}.
\end{equation}

Additionally, we have defined work as a random variable which, by convention, takes positive values whenever it is performed by the system, i.e.\ when the system's energy decreases ($\Delta E_i < 0$). 
From the perspective of a user of the system, this corresponds to a gain of work.

Conversely, work is negative when the system's energy increases, i.e.\ when $\Delta E_i > 0$; in this case, we say that the system acquires work, and thus the user simultaneously experiences work expenditure.

Formally, the work performed in the $i$-th step of a given process (thanks to which initial state $\rho$ is transformed into the final state $\sigma$) is described by a random variable $W_i$ with the following distribution:
\begin{equation}
    P(W_i = -\Delta E_i) = p_i,
    \quad
    P(W_i = 0) = 1 - p_i,
\end{equation}
where $p_i$ is the probability that the excited level is occupied, and $(1 - p_i)$ is the probability that the ground level is occupied. 

The total work for an $N$-step process is then simply the sum of the work performed in each step, that is, 
\begin{equation}\label{def:W_as_a_sum}
	W=\sum_{i=0}^{N-1} W_i.
\end{equation}
Let us emphasize that later in Appendix \ref{sec:appendA} (especially in Appendix \ref{sec:WII}) we shall analyze an isothermal process between energies  $E(p_{\mathrm{in}})$ and $E(p_{\mathrm{out}})$, in which case random variables $W_i$ have the following distributions: 
\begin{equation}\label{def:Wi}
    P(W_i=-\Delta E_i)=q_i\quad\text{and}\quad P(W_i=0)=1-q_i.
\end{equation}

Within this setting, where the primary goal is to characterize operations that can be performed with no work expenditure, it is convenient to introduce the work loss random variable, namely $-W$, which is, of course, fully determined by $W$ itself.

By saying ``no work expenditure'',  we formally mean that the event $\{-W > \epsilon\}$, for some $\epsilon > 0$, has vanishing probability. 
Conversely, whenever we show that this probability is strictly positive for some $\epsilon > 0$, this implies that given transition from an initial state $\rho$ to a final state $\sigma$ cannot be achieved with no work expenditure, and thus, in particular, is unattainable by means of $\mcoww$ (see Definition~\ref{def:MCO_0}).

\subsubsection{Path}\label{sec:appendix_path}
Referring to Section \ref{sec:tools}, recall that at each step of an $N$-step process we apply a thermalization map  
\begin{equation}
    \mathcal{C}_i \;=\; (1-\lambda_i)\,\ONE \;+\; \lambda_i\, \hat{\tau}_i, \qquad i \in \{0,1,\ldots,N-1\},
\end{equation}
where $\hat{\tau}_i$ is the constant map that replaces the state with the thermal state $\tau_i$ corresponding to the Hamiltonian $H(E_i)$, and $\ONE$ is the identity map. The parameters $(1-\lambda_i),\lambda_i \in [0,1]$ specify the probabilities with which either of the two behaviours (doing nothing or fully thermalizing) is realized.

As already noted in Section~\ref{sec:arbitrary-cycle}, an arbitrary process can be represented as a mixture of \emph{paths} $\{R_j\}_{j\in I}$, 
each occurring with probability~$\gamma_j$. By a path $R_j$ we mean a simplified, deterministic version of the process in which, at every step $i$, the corresponding thermalization parameter satisfies  
\begin{equation}
    \lambda_{i,j} \in \{0,1\},
\end{equation}
so that at step $i$ either nothing happens or a full thermalization is performed.

Each path $R_j$ can be depicted as a polygonal curve (cf. Fig. \ref{fig:Gibbs-general}) consisting of:
\begin{itemize}
    \item {horizontal segments}, each of length $\Delta_i$, representing steps in which the Hamiltonian is changed but no thermalization occurs, and  
    \item {vertical segments}, representing thermalizations, which appear exactly at those steps where $\lambda_{i,j} = 1$.
\end{itemize}

We will use $p_{\mathrm{in},j}$ and $p_{\mathrm{out},j}$ to denote the starting and ending points of each path $R_j$, respectively. 
    For the process that transforms state $\rho \coloneqq (1-p_{\mathrm{IN}}, p_{\mathrm{IN}})$ into state $\sigma \coloneqq (1-p_{\mathrm{OUT}}, p_{\mathrm{OUT}})$  the mixture of paths $\{ R_j\}_{j \in I}$ has to fulfill 
	\begin{equation}
		\mathop\forall_{j \in I} p_{\mathrm{in},j} = p_{\mathrm{IN}}
	\end{equation}
	and
	\begin{equation}
		\sum_{j \in I} \gamma_j p_{\mathrm{out},j} = p_{\mathrm{OUT}}.
	\end{equation}

%\textit{\textcolor{magenta}{[MS: Is $p_{\mathrm{in},j}$ some new notation for $p_{0,j}$, and -- similarly -- for out and N-1? If yes maybe we could ad his equality to clarify. Perhaps we could use some figure? \newline \textbf{HW-Ś: I do not remember ever using $p_0$ in this article, and also -- in general -- we have $q_0\neq p_\mathrm{in}$. In the latter case, it is true: whenever a single path $R$ consists of $N$ steps altogether, then $p_\mathrm{out}=q_{N-1}$. BUT, at this point it is irrelevant. I only need here to indicate how we proceed from thinking about a whole general process, which starts at $p_\mathrm{IN}$ and ends at $p_\mathrm{OUT}$ to thinking about a single path $R$ with starting and ending point denoted by $p_\mathrm{in}$ and $p_\mathrm{out}$, respectively.}]}}\textcolor{magenta}{[MS: w takim razie może zostać tak jak było]} {\color{red}[HW-Ś: Maćku, tak jak było, czy tak, jak jest teraz? Bo nie wiem, czy zmieniać coś, czy jest ok ;)]}
    
	Let $Y$ be the random variable defined as $P(Y = p_{\mathrm{out},j}) = \gamma_j$. Then
	\begin{equation}
        \E Y = \sum_{j \in I} \gamma_j p_{\mathrm{out},j} = p_{\mathrm{OUT}}.
	\end{equation}
    \vspace{\baselineskip}

For notational simplicity, we drop the index~$j$ whenever possible and refer to a single path simply as~$R$. 
Similarly, we use the symbols $p_\mathrm{in}$ and $p_\mathrm{out}$ to denote the starting and ending points of this path, respectively. 

Any such path $R$ can be decomposed into three parts:
\begin{itemize}
    \item $R_{\mathrm{(I)}}$, the segment up to the first thermalization,  
    \item $R_{\mathrm{(II)}}$, the segment between the first and the last thermalization, and  
    \item $R_{\mathrm{(III)}}$, the segment starting immediately after the last thermalization.
\end{itemize}
This decomposition is illustrated in Fig.~\ref{fig:proofR}.  Let us emphasize that we consider here only paths containing at least two thermalizations (paths with at most one thermalization are briefly discussed in the main text; see Section~\ref{subsec:co5}).
\vspace{\baselineskip}

\begin{figure}
	\centering
	\includegraphics[width=0.9\linewidth]{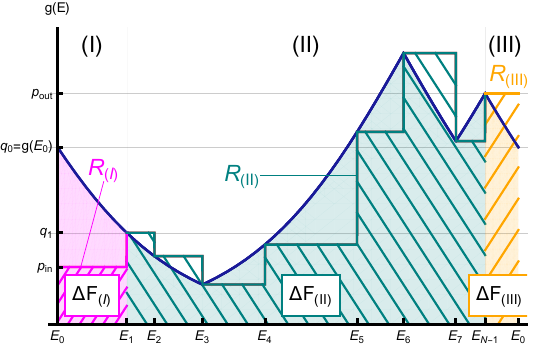}
	\caption{
        Visualization of the relevant concepts.
The path $R$ is divided into three parts $R_{\mathrm{(I)}}, R_{\mathrm{(II)}}, R_{\mathrm{(III)}}$, shown in magenta, teal, and orange, respectively.
Part (I) begins at $(E_0, p_{\mathrm{in}})$ and ends at $(E_1, q_1)$, where the first thermalization takes place at energy level $E_1$.
Part (II) comprises the portion of the path from the first thermalization at $E_1$ to the last thermalization at $E_{N-1}$ (which is $E_8$ in the diagram).
Part (III) starts at $(E_{N-1}, q_{N-1}=p_{\mathrm{out}})$ and continues until the energy returns to $E_0$, i.e., to the point $(E_0, p_{\mathrm{out}})$.
The total area under the Gibbs curve is
$\Delta F = \Delta F_{\mathrm{(I)}} + \Delta F_{\mathrm{(II)}} + \Delta F_{\mathrm{(III)}}$,
with each contribution shown in a~different color: violet for part (I), green for part (II), and yellow for part (III). 
% {\color{red}[HW-Ś: Rysnuek jest super, brakuje tylko zaznaczenia $q_1$ na osi. No i zmieniła się wartość $N$, ale to poprawiłam już w tekście.]}\textcolor{magenta}{[MS: dodane, zobacz czy w dobrym miejscu]}
}
	\label{fig:proofR}
\end{figure}

%\begin{figure}
%	\centering
%	\includegraphics[width=0.9\linewidth]{figures/R_divided_into_3_parts.pdf}
%	\caption{\textcolor{magenta}{Old drawing to be remove later after checking the above figure. }
%        Visualization of the relevant concepts.
%The path $R$ is divided into three parts $R_{\mathrm{(I)}}, R_{\mathrm{(II)}}, R_{\mathrm{(III)}}$, shown in magenta, teal, and orange, respectively.
%Part (I) begins at $(E_0, p_{\mathrm{in}})$ and ends at $(E_1, q_1)$, where the first thermalization takes place at energy level $E_1$.
%Part (II) comprises the portion of the path from the first thermalization at $E_1$ to the last thermalization at $E_{N-1}$ (which is $E_6$ in the diagram).
%Part (III) starts at $(E_{N-1}, q_{N-1}=p_{\mathrm{out}})$ and continues until the energy returns to $E_0$, i.e., to the point $(E_0, p_{\mathrm{out}})$.
%The total area under the Gibbs curve is
%$\Delta F = \Delta F_{\mathrm{(I)}} + \Delta F_{\mathrm{(II)}} + \Delta F_{\mathrm{(III)}}$,
%with each contribution shown in a~different color: violet for part (I), green for part (II), and yellow for part (III).}
%	\label{fig:proofR}
%\end{figure}

In what follows, let $\Delta F_i$ denote the difference between 
$F(E_{i+1})$ and $F(E_i)$.
Due to \eqref{eq:deltaTildeF} 
it can be written as 
\begin{equation}\label{def:DeltaFi}
    \Delta F_i \coloneqq \int\limits_{E_i}^{E_{i+1}}g(E) \mathrm{d}E.
\end{equation}
For the parts (I), (II), (III), described above, we, therefore, obtain
\begin{align}\label{eq:DeltaF(I)(II)(III)}
    \Delta F_\mathrm{(I)} 
    %= \int\limits_{\mathrm{(I)}}g(E) \mathrm{d}E 
    = \int\limits_{E_0}^{E_{1}}g(E) \mathrm{d}E,\quad
    \Delta F_\mathrm{(II)} 
    %= \int\limits_{\mathrm{(II)}}g(E) \mathrm{d}E 
     = \int\limits_{E_1}^{E_{N-1}}g(E) \mathrm{d}E,\quad
     \Delta F_\mathrm{(III)} 
    %= \int\limits_{\mathrm{(III)}}g(E) \mathrm{d}E 
    = \int\limits_{E_{N-1}}^{E_0}g(E) \mathrm{d}E,
\end{align}
where $E_1$ is the energy at which the first thermalization is performed, indicating the end of $R_{\mathrm{(I)}}$ and the beginning of $R_{\mathrm{(II)}}$, and $E_{N-1}$ is the energy at which the last thermalization is performed, indicating the end of $R_{\mathrm{(II)}}$ and the beginning of $R_{\mathrm{(III)}}$ (see Fig.\ \ref{fig:proofR}).  
Furthermore, let us indicate that
\begin{equation}\label{eq:deltaF}
    \Delta F = \sum_{i=0}^{N-1} \Delta F_i 
    % = \sum_{j \in \{\mathrm{(I), (II), (III)}\}} \Delta F_j =
    =\Delta F_\mathrm{(I)}+
    \Delta F_\mathrm{(II)}+\Delta F_\mathrm{(III)}
    \blk
    =\int\limits_{E_0}^{E_0}g(E) \mathrm{d}E = 0. 
\end{equation}

We will denote the work of the path as $W$ (similarly to the case of the work of the whole process -- see Eqs.~\eqref{def:W_as_a_sum} and~\eqref{def:Wi} ; importantly, it will always be clear from the context whether we refer to the entire process or to a single path), where $W =  W_\mathrm{(I)}+ W_\mathrm{(II)}+ W_\mathrm{(III)}$, and $W_\mathrm{(I)}, W_\mathrm{(II)}$, $W_\mathrm{(III)}$ stand for the work of each of the parts (I), (II), (III), respectively. 

% \begin{figure}
% 	\centering
% 	\includegraphics[width=0.9\linewidth, angle=-90]{figures/skanowanie0129.pdf}
% 	\caption{
%         An example of the graphical visualization of the values $\langle W_i\rangle$ {\color{red}(hatched areas)} and $\Delta F_i$ {\color{red}(shaded areas)}, as well as relations between them -- in two different cases: (1) $E_1<E_2$, meaning that the Gibbs curve over the interval $[E_1,E_2]$ is decreasing, and thus $\Delta F_1>0$ and $\langle W_1\rangle<0$; {\color{red}we have $\Delta F_1\leq - \langle W_1\rangle$} (2) $E_2<E_3$, meaning that the Gibbs curve over the interval $[E_2,E_3]$ is increasing, and thus $\Delta F_2<0$ and $\langle W_2\rangle >0$; {\color{red}we have $\langle W_2\rangle\leq - \Delta F_2$}.
%         }
% 	\label{fig:fig0}
% \end{figure}

% \begin{figure}
%     \centering
% \includegraphics[width=1\linewidth]{fig8.png}\caption{\textcolor{magenta}{[MS: porównująć z odręcznym rysunkiem prawa czerwona pozioma linia chyba powinna być wyżej na poziome $p_{out}$ oraz przy $\Delta F_2$ chyba powinien być minus. KH: Zobacz na następnej stronie poprawkę ]}}
%     \label{fig:placeholder}
% \end{figure}
\begin{figure}
    \centering
    \includegraphics[width=0.75\linewidth]{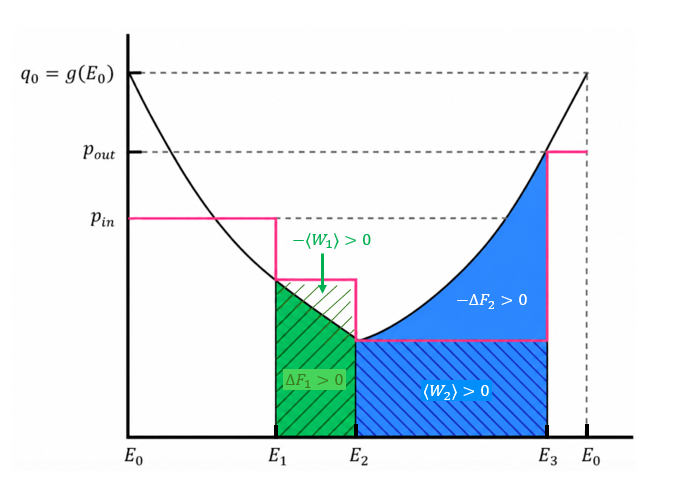}
    \caption{
    An example of the graphical visualization of the values $\langle W_i\rangle$ (hatched areas) and $\Delta F_i$ (shaded areas), as well as relations between them -- in two different cases: (1) $E_1<E_2$, meaning that the Gibbs curve over the interval $[E_1,E_2]$ is decreasing, and thus $\Delta F_1>0$ and $\langle W_1\rangle<0$; we have $\Delta F_1\leq - \langle W_1\rangle$ (2) $E_2<E_3$, meaning that the Gibbs curve over the interval $[E_2,E_3]$ is increasing, and thus $\Delta F_2<0$ and $\langle W_2\rangle >0$; we have $\langle W_2\rangle\leq - \Delta F_2$.
    }
    \label{fig:fig0}
\end{figure}
% \begin{figure}
%     \centering
%     \includegraphics[width=0.75\linewidth]{fig8skopanyGibbs.png}
%     \caption{
%         An example of the graphical visualization of the values $\langle W_i\rangle$ (hatched areas) and $\Delta F_i$ (shaded areas), as well as relations between them -- in two different cases: (1) $E_1<E_2$, meaning that the Gibbs curve over the interval $[E_1,E_2]$ is decreasing, and thus $\Delta F_1>0$ and $\langle W_1\rangle<0$; we have $\Delta F_1\leq - \langle W_1\rangle$ (2) $E_2<E_3$, meaning that the Gibbs curve over the interval $[E_2,E_3]$ is increasing, and thus $\Delta F_2<0$ and $\langle W_2\rangle >0$; we have $\langle W_2\rangle\leq - \Delta F_2$.
%     }
%     \label{fig:fig0}
% \end{figure}

Another important important observation that we need to make here is as follows:
\begin{equation}\label{eq:<Wi>_DeltaFi}
    \left\langle W_i\right\rangle\leq -\Delta F_i\quad \text{for every}\quad i\in\{0,\ldots, N-1\},
\end{equation}
where the average work $\langle W_i\rangle$ in the $i$th step can be calculated directly from Eq.\ \eqref{def:Wi}, i.e., $\langle W\rangle=-\Delta E_i\,q_i$.

%\textit{\textcolor{magenta}{[MS: in my opinion below paragraph is very important. I would even write it in more details pointing out when work and F are positive/negative areas, depending on the monotonicity of Gibs curve. ALso we could add justification why it works like this based on the order of integral limits. Furthermore, it would also be useful to clearly state how the monotonicity of the Gibbs curve depends on the change of energy.\newline
%\textbf{HW-Ś: From my point of view, a short description together with a simple figure is sufficient. 
%If someone truly wishes to understand the details, they will have to sit down and think it through anyway, or even draw their own picture to see how it works. 
%The task should be manageable for a student, so I believe it is reasonable to leave it to the reader. 
%On the other hand, feel free to add something if you have a good idea for how to describe it more clearly. \newline
%Furthermore, I have added above a short note on how to calculate the average work, although I do not consider it necessary.}]}} {\color{red}[HW-Ś: Czy moje argumenty Was przekonują?]}

To see it (cf.\ Fig.\ \ref{fig:fig0}), note that for every $i\in\{0,\ldots, N-1\}$ the value $\langle W_i\rangle$ is equal to the area under the part of the path $R$ which starts at $E_i$ and ends at $E_{i+1}$ (with either plus or minus sign), while $\Delta F_i$ is equal to the area under the corresponding part of the Gibbs curve, that is, the one over the interval $[E_i, E_{i+1}]$, (again, with either plus or minus sign). Observe that whenever the Gibbs curve is increasing (decreasing) over $[E_i, E_{i+1}]$, the value $\langle W_i\rangle$ is positive (negative) and, simultaneously, $\Delta F_i$ is negative (positive). %Moreover, in such a case, we have 
%\begin{equation}
%    \left|\left\langle W_i\right\rangle\right|\leq \left|\Delta F_i\right|\quad \Big(\;\left|\left\langle W_i\right\rangle\right|\geq \left|\Delta F_i\right|\;\Big),
%\end{equation}
%provided that the first thermalization has already been performed ($i\in\{1,\ldots, N-1\}$). 
\vspace{\baselineskip}

In the end, let us also recall that $q_0$, which, according to \eqref{def:qi}, is given by $q_0=e^{-\beta E_0}/{(1+ e^{-\beta E_0})}$, describes the probability of occupying the excited level of the Gibbs curve at the beginning and at the end of the process (corresponding to the energy level $E_0$). Similarly, the values 
\begin{equation}\label{def:q1_gN-1}
    q_1= \frac{e^{-\beta E_1}}{1+ e^{-\beta E_1}} \quad\text{and} \quad q_{N-1}= \frac{e^{-\beta E_{N-1}}}{1+ e^{-\beta E_{N-1}} }
\end{equation}
describe the probabilities of occupying the excited level of the Gibbs curve at the moment of the first and the last thermalization, respectively, that is, the beginning of parts $R_\mathrm{(II)}$ and $R_\mathrm{(III)}$ of the path (note that $q_{N-1}=p_\mathrm{out}$, and thus  $E_{N-1}=E(p_{\mathrm{out}})$).

\subsection{Idea of the proof}\label{sec:idea}

%\textit{\textcolor{magenta}{[MS: we could use in the text here in the following places that $p_{\mathrm{OUT}}$ is in between $q_0$ and $,p_{\mathrm{IN}}$ to clarify the formulas.\newline
%\textbf{\color{red}HW-Ś: I do not understand the comment.}]}}

First of all, let us point out that whenever 
$p_{\mathrm{OUT}}\in[q_0,p_{\mathrm{IN}}]$ (assuming $q_0 \leq p_{\mathrm{IN}}$), or 
$p_{\mathrm{OUT}}\in[p_{\mathrm{IN}},q_0]$ (if $p_{\mathrm{IN}} \leq q_0$),
the transition can be achieved with no work loss. 
To see this, let us consider the first case (the second one is analogous). 
Since $p_{\mathrm{OUT}}\in[q_0,p_{\mathrm{IN}}]$, we can write
\begin{equation}
    p_{\mathrm{OUT}} = (1-\gamma)\, p_{\mathrm{IN}} + \gamma q_0,
\qquad \gamma \in [0,1].
\end{equation}
One can observe that such a process consists of a single step, with no energy-level transformation. 
Formally, this corresponds to setting $E_1 = E_0$ (note that in this case parts~(II) and~(III) of the process degenerate and become trivial). 
Namely, with probability $1-\gamma$ nothing happens, which corresponds to a path $R_1$ that degenerates to a single point at $p_{\mathrm{in},1}=p_{\mathrm{IN}}$, and with probability $\gamma$, a full thermalization is performed at the very beginning, which corresponds to a path $R_2$ that immediately moves vertically from $p_{\mathrm{in},2}=p_{\mathrm{IN}}$ to $q_0$. 
Since path $R_1$ leads to the final state $p_{\mathrm{out},1}=p_{\mathrm{IN}}$, 
and path $R_2$ leads to the final state $p_{\mathrm{out},2}=q_0$, 
the overall process results in the convex mixture
\begin{equation}
   p_{\mathrm{OUT}} = (1-\gamma)p_{\mathrm{out},1} + \gamma p_{\mathrm{out},2}=(1-\gamma)p_{\mathrm{IN}} + \gamma q_0, 
\end{equation}
which is the desired output state.
\vspace{\baselineskip}

The aim of this paper is to deal with the opposite case, that is 
$p_{\mathrm{OUT}}\notin[q_0,p_{\mathrm{IN}}]$ (assuming $q_0 \leq p_{\mathrm{IN}}$), and  
$p_{\mathrm{OUT}}\notin[p_{\mathrm{IN}},q_0]$ (if $p_{\mathrm{IN}} \leq q_0$), which gives the following possible settings: 
\begin{equation}\label{cases-i-iv}
    \text{(i)}\quad p_{\mathrm{IN}}\leq q_0<p_{\mathrm{OUT}},\quad
    \text{(ii)}\quad q_0<p_{\mathrm{IN}}<p_{\mathrm{OUT}},\quad
    \text{(iii)}\quad p_{\mathrm{OUT}}<p_{\mathrm{IN}}\leq q_0,\quad
    \text{(iv)}\quad p_{\mathrm{OUT}}<q_0< p_{\mathrm{IN}}    
\end{equation}
to consider. 

% \sout{\color{red}
% \textbf{Another important assumption made in this paper is that $p_\mathrm{IN} \neq 0$, which means that we exclude the case where $\rho = (1 - p_\mathrm{IN}, p_\mathrm{IN})$ is the pure ground state. 
% In other words, we do not consider the case $\rho = (1,0)$ (or, equivalently, $\rho = |0\rangle\langle 0|$ in Dirac notation), which remains an open question.
% }}

%\newline

Let us now describe the \textbf{idea of the proof}. Since we can think of a process as a probabilistic mixture of paths, we shall start by analyzing the work loss $-W$ of a single path $R$. 
Let us here highlight that, to make the reasoning more clear and compact, we identify with each other all those paths which, after being transformed in accordance with the \emph{Path Shrinking} method, are identical. 
The method itself, as well as the rationale for considering only those paths to which the \emph{Path Shrinking} method has already been applied, is given in Section \ref{sec:shrinking}, so we do not repeat it here.

We will first separately consider the possible work losses: $-W_\mathrm{(I)}$, $-W_\mathrm{(II)}$, and $-W_\mathrm{(III)}$ in each of the parts: (I), (II), and (III) of an arbitrary path $R$.

The necessary conclusions regarding the work loss $-W_{\mathrm{(II)}}$ follow primarily from the Jarzynski equality~\cite{Jarzynski1997} 
% \textit{\textbf{[HW-Ś: Should we add (/ repeat from the main text) some short intuitive description regarding this equality?]}}\blk, \mh{nie umiem tego tak jasno wytlumaczyc, to by bylo dlugo, dla mnie wystarczy tak jak teraz jest w glownym tekscie}
and they are collected in Section~\ref{sec:WII}.

Section~\ref{sec:WI_III}, in turn, is devoted to estimating the work losses $-W_\mathrm{(I)}$ and $-W_\mathrm{(III)}$.

The results obtained in Sections \ref{sec:WII} and \ref{sec:WI_III} are then combined in Section \ref{sec:entire_path} in order to conclude about the work loss $-W$ of the entire single path $R$, which is done  separately in two mutually exclusive cases:

\begin{itemize}
    \item[(a)] $q_0< q_{N-1}=p_{\mathrm{out}}$, covering settings (i) and (ii) in \eqref{cases-i-iv} (for a single path $R$), and guaranteeing that $-W_\mathrm{(III)}>0$ with probability bounded away from zero (cf.\ Fig.\ \ref{fig:cases_with_primes});
    \item[(b)] $q_{N-1}=p_{\mathrm{out}}<q_0$,  covering settings (iii) and (iv) in \eqref{cases-i-iv} (for a single path $R$), which case shall be further examined within two (not disjoint!) sub-cases: 
    \begin{itemize}
        \item[(b1)]  $p_{\mathrm{out}}<q_0$ and $q_1<q_0$, guaranteeing that $-W_\mathrm{(I)}>0$ with probability bounded away from zero (cf.\ Fig.\ \ref{fig:cases_with_primes});
        \item[(b2)] $p_{\mathrm{out}}<q_0$ and $p_{\mathrm{out}}<q
_1$, guaranteeing that $-W_\mathrm{(II)}>0$ with probability bounded away from zero (cf.\ Fig.\ \ref{fig:cases_with_primes}).
    \end{itemize}
\end{itemize}
%\textcolor{magenta}{[NEW: MS: we believe that we should have weak ineequality instead of strong ones. We also thing it is enought since in the proof we use weak versions. But it would be good to double check this.]} {\color{red}\textbf{[HW-Ś: Maćku, rozumiem, że mówisz o tej zależności pomiędzy $p_\mathrm{out}$ a $q_0$? Jeśli tak, to moim zdaniem wystarczą jednak ostre nirówności, bo zobacz, że przypadek z równością wchodzi zakres, o którym piszemy, że go nie roważamy (patrz danie nad (A24)).]}}

More specifically, \textbf{in case~(a)} we show that the positive work loss $-W_{\mathrm{(III)}}$ scales appropriately with respect to the (possibly negative) contribution $-W_\mathrm{(I)}-W_\mathrm{(II)}$, with probability bounded away from zero. 
By this we mean that the strictly positive (i.e., greater than some positive constant) total work loss
\begin{equation}
  -W = -W_{\mathrm{(I)}} - W_{\mathrm{(II)}} - W_{\mathrm{(III)}}
\end{equation}
occurs with probability greater or equal to some positive constant.

Note that, under the assumption $q_0 < p_\mathrm{out}$, the Gibbs curve corresponding to the last part of the path is decreasing and -- as we will argue in the course of the proof -- this ensures that the work loss $-W_{\mathrm{(III)}}$ in part~(III) is strictly positive and non-vanishing with probability greater or equal to some positive constant.

The main challenge is then to compare it with the term $-W_{\mathrm{(I)}} - W_{\mathrm{(II)}}$. 
We can easily prove that $-W_{\mathrm{(I)}} \ge \Delta F_{\mathrm{(I)}}$ with probability bounded away from zero. 
Let us note here that if we could additionally show that the work loss in part~(II) is at least $\Delta F_{\mathrm{(II)}}$, we would immediately obtain the desired result.
 Unfortunately, this cannot be proven in general. 
However, we can show that, with probability bounded away from zero, the difference between $-W_\mathrm{(II)}$ and $\Delta F_{\mathrm{(II)}}$ is smaller than the corresponding difference in part~(III), that is,
\begin{equation}
    \Delta F_{\mathrm{(II)}} + W_\mathrm{(II)} < -W_\mathrm{(III)} - \Delta F_{\mathrm{(III)}}.
\end{equation}

To complete the proof, it then suffices to combine the results for all parts of the path and observe that the total free energy difference over the entire path is always zero (cf.~\eqref{eq:deltaF}).
\vspace{\baselineskip}

%\textit{\textcolor{magenta}{[MS: In the first part of the second sentence below it would be good to also mention $W_{III}$. Also after that maybe it could be better to divide that sentence into two.] \textbf{[HW-Ś: Maćku, po przeczytaniu Twojej uwagi,  zauważyłam, że źle (a raczej nieprecyzyjnie) poprowadziłam narrację, i dlatego opisy zarówno dla przypadku (a), jak i (b) zostały zmienione. {\color{red}-- Czy to, co zapisałam, ma sens? Jest intuicyjne dla Was?}]}}}

\textbf{In case (b)}, the reasoning is slightly more involved. Here, the key point to prove is either that the positive term $-W_{\mathrm{(I)}}$ scales appropriately relative to the (possibly negative) contribution $-W_{\mathrm{(II)}} - W_{\mathrm{(III)}}$ (with probability bounded away from zero), or that $-W_{\mathrm{(II)}}$ itself is positive while simultaneously $-W_{\mathrm{(I)}} \ge 0$ and $-W_{\mathrm{(III)}} \ge 0$ (again with probability bounded away from zero). The first occurs whenever $q_1 < q_0$ (case (b1)), whereas the second occurs if $p_{\mathrm{out}} < q_1$ (case (b2)).

For certian technical reasons, and to improve the clarity of the presented arguments, we shall introduce specific subcases (b1$^{\prime}$) and (b2$^{\prime}$) of cases (b1) and (b2), respectively.  
Note that case (b) is determined by the assumption $p_\mathrm{out} < q_0$, which allows us to fix a point $q_*$ between these two values; for example,  
\begin{equation}
    q_* = p_{\mathrm{out}}+\frac{q_0 - p_\mathrm{out}}{2}=\frac{p_\mathrm{out}+q_0}{2}.
\end{equation}
We then define the following cases:
\begin{itemize}
    \item[(b1$^{\prime}$)] $p_{\mathrm{out}} < q_0$ and $q_1 \leq q_*$,
    \item[(b2$^{\prime}$)] $p_{\mathrm{out}} < q_0$ and $q_1 > q_*$,
\end{itemize}
which are mutually disjoint and together cover the entire case (b) (see Fig. \ref{fig:cases_with_primes}).  

%\textit{\textcolor{magenta}{[MS: In my opinion arrows in the figure are not fully clear, and it took me a while to understand the concept of the drawing. On the other hand, Karol got it very fast. Nevertheless, we could think more about this aspect of the figure later when we will be preparing the final version of figures.]}}

Most importantly, after this modification, all limiting cases involving a possible vanishing of the positive contribution $-W_{\mathrm{(I)}}$ (in the limit $q_1 \nearrow q_0$ in case (b1)) or a possible vanishing of the positive contribution $-W_{\mathrm{(II)}}$ (in the limit $q_1 \searrow p_{\mathrm{out}}$ in case (b2)) are now eliminated (see Sections \ref{sec:WI_III} and \ref{sec:entire_path} for details).
 
Proceeding to our argumentation, we can then conclude the following:
\begin{itemize}
    \item If we can ensure that $-W_\mathrm{(I)}$ is strictly positive and non-vanishing with probability bounded away from zero (\textbf{case (b1$^{\prime}$)}), we then have to -- similarly as in scenario (a) -- show that it scales appropriately with the possibly negative term $-W_\mathrm{(II)} - W_\mathrm{(III)}$. The reasoning in this case is analogous to that described above while discussing scenario (a).
    \item If, on the other hand, we can prove that $-W_\mathrm{(II)}$ is strictly positive and non-vanishing with probability bounded away from zero (\textbf{case (b2$^{\prime}$)}), then the claim follows immediately, since it is straightforward to verify that $W_\mathrm{(I)} = W_\mathrm{(III)} = 0$ can simultaneously hold with some fixed probability.
\end{itemize}

% \begin{figure}
% 	\centering
% 	\includegraphics[width=\linewidth/2, angle=-90]{figures/skanowanie0128.pdf}
% 	\caption{
%     \textcolor{red}{remove after checking}
%     }
	
% \end{figure}
\begin{figure}
    \centering
    \includegraphics[width=0.8\linewidth]{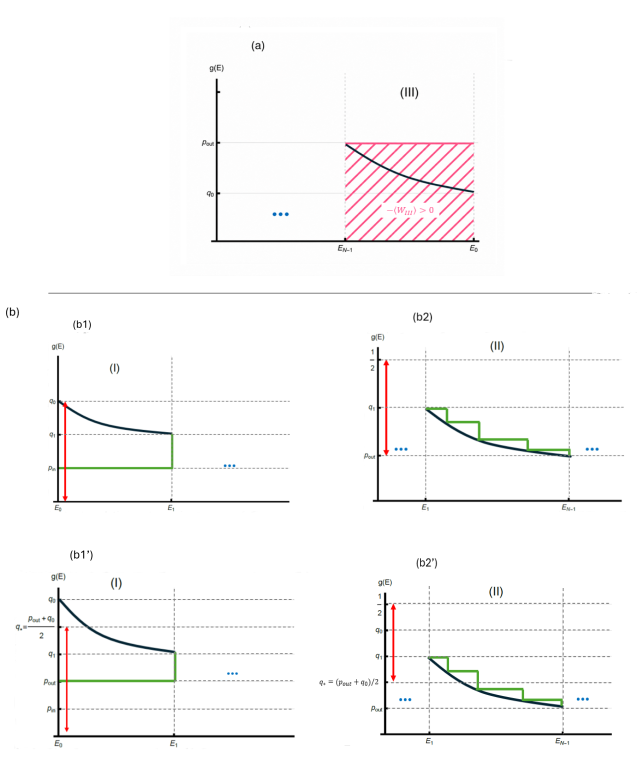}
    \caption{
    \underline{On the top}: an exemplary realization of \textbf{scenario (a)}: $q_0<p_\mathrm{out}$, in which we can guarantee a positive work loss $-W_\mathrm{(III)}$ in the third part of the path, with probability greater or equal to some positive constant. This follows from the fact that in this case the Gibbs curve associated with the last part of the path is decreasing.\newline    
        \underline{On the bottom}: 
        possible realizations of \textbf{scenario (b)}: $p_\mathrm{out}<q_0$, including: \textbf{case (b1)}  $p_\mathrm{out}<q_0$ and $q_1<q_0$ (\(q_1\) can take any value in the red range between \(0\) and \(q_0\)), which guarantees that $-W_\mathrm{(I)} > 0$ with probability of  $p_\mathrm{in}$ (the first thermalization occurs below $q_0$, making the Gibbs curve over part~(I) decreasing); and \textbf{case (b2)}: $p_\mathrm{out}<q_0$ and $p_\mathrm{out}<q_1$ (\(q_1\) can take any value in the red range between \(p_\mathrm{out}\) and \(1/2\)), which ensures that $-W_\mathrm{(II)} > 0$ with some positive probability (because the first thermalization $q_1$ lies above the last one $p_\mathrm{out}$). In contrast, neither part~(I) nor part~(III) guarantees a positive work loss, occuring with non-zero probability, in case~(b2). \newline
Moreover, \underline{the special instances of scenarios (b1) and (b2)} are also depicted:  
\textbf{case (b1$^{\prime}$)}: $p_\mathrm{out} < q_0$ and $q_1 < q_* = (p_\mathrm{out} + q_0)/2$ (\(q_1\) can take any value in the red range between \(0\) and \(q_*\)), and  
\textbf{case (b2$^{\prime}$)}: $p_\mathrm{out} < q_0$ and $q_1 > q_*$ (\(q_1\) can take any value in the red range between \(q_*\) and \(1/2\)).  
In contrast to the original cases (b1) and (b2), these refined subcases already are mutually exclusive.
    }
    \label{fig:cases_with_primes}
\end{figure}
Finally, in Section \ref{sec:paths_into_process}, we make conclusions about the work loss of the whole process, using what we know about the work loss of any single path.

\subsection{\texorpdfstring{Work of $R_{\mathrm{(II)}}$}{Work of R 2}}\label{sec:WII}

Let us recall that, according to the notation introduced above, $W_\mathrm{(II)} = \sum_{i=1}^{N-2} W_i$ and $\Delta F_\mathrm{(II)} = \sum_{i=1}^{N-2} \Delta F_i$ (see, for instance, Fig.\ \ref{fig:proofR}, where $N - 2 = 7$).

%\textit{\textcolor{magenta}{[MS: note to myself: check indexes above n-2,1]}}

As already mentioned in Section \ref{sec:idea}, the key fact that allows us to appropriately estimate the work loss $-W_{\mathrm{(II)}}$ along part (II) of the path $R$ is the Jarzynski identity~\cite{Jarzynski1997}. In our setting, this fundamental identity applied to a~single $i$-th step takes the form
\begin{equation}\label{eq:jarzynski}
    \left\langle e^{\beta W_i} \right\rangle = e^{-\beta \Delta F_i}, \quad i \in \{0,1,\ldots, N-1\},
\end{equation}
where $\beta > 0$ is an arbitrarily fixed (inverse) ambient temperature.  
Using this identity, we can readily establish the following result:

%\textit{\textcolor{magenta}{[MS: note to myself: this lemma is useful only for some specific range of epsilon, keep it in mind when we will see its usage.]}}

\begin{lemma}\label{lem:w2}
    Let $\epsilon_\mathrm{(II)} > 0$. With probability of at least
    \begin{equation}
        1-e^{-\beta \epsilon_\mathrm{(II)}},
    \end{equation}
    we have
    \begin{equation}
        -W_{\mathrm{(II)}} \geq \Delta F_{\mathrm{(II)}} - \epsilon_{\mathrm{(II)}}.
    \end{equation}
\end{lemma}

\begin{proof}
One can easily observe that the Jarzynski identity \eqref{eq:jarzynski}, together with the mutual independence of the random variables $W_i$ (as well as its further consequences), further implies
\begin{equation}\label{eq:jarzynski_(II)}
        \left\langle e^{\beta W_{\mathrm{(II)}}} \right\rangle
        = \left\langle e^{\beta \sum_{i=1}^{N-2} W_i} \right\rangle
         = \left\langle \prod_{i=1}^{N-2} e^{\beta W_i} \right\rangle
         = \prod_{i=1}^{N-2} \left\langle e^{\beta W_i} \right\rangle 
         = \prod_{i=1}^{N-2} e^{-\beta \Delta F_i}
         = e^{-\beta \sum_{i=1}^{N-2} \Delta F_i}
         = e^{-\beta \Delta F_{\mathrm{(II)}}}.
\end{equation}

Now, since $W_{\mathrm{(II)}}$ (and therefore also $-W_{\mathrm{(II)}}$) is, in our setting, a discrete random variable (specifically,
$W_{\mathrm{(II)}} = \sum_{i=1}^{N-2} W_i$, where the variables $W_i$, $i \in \{1,2,\ldots, N-2\}$, are mutually independent, and satisfy  
$P(W_i = -\Delta E_i) = q_i$, $P(W_i = 0) = 1 - q_i$), it follows that for any $a \in \mathbb{R}$ we have
\begin{equation}
    P\left( -W_{\mathrm{(II)}} \leq a \right)
    = \sum_{j \in J :\; x_j \leq a} P\left( -W_{\mathrm{(II)}} = x_j \right),
\end{equation}
where $J$ is a finite index set, and $\{x_j\}_{j \in J}$ denotes the set of all possible values of $-W_{\mathrm{(II)}}$.  
We may therefore, in particular, write
\begin{equation}\label{eq:W(II)_discrete}
    P\left( -W_{\mathrm{(II)}} \leq \Delta F_{\mathrm{(II)}} - \epsilon_{\mathrm{(II)}} \right)
    = \sum_{j \in J :\; x_j \leq \Delta F_{\mathrm{(II)}} - \epsilon_{\mathrm{(II)}}} 
      P\left( -W_{\mathrm{(II)}} = x_j \right).
\end{equation}

Note that, under the assumption
\begin{equation}
    x_j \leq \Delta F_{\mathrm{(II)}} - \epsilon_{\mathrm{(II)}},
\end{equation}
we clearly have
\begin{equation}
    0 \leq \beta\left( \Delta F_{\mathrm{(II)}} - \epsilon_{\mathrm{(II)}} - x_j \right),
\end{equation}
and consequently,
\begin{equation}\label{eq:simple}
    1 = e^{0} \leq e^{\beta\left( \Delta F_{\mathrm{(II)}} - \epsilon_{\mathrm{(II)}} - x_j \right)}.
\end{equation}

Combining \eqref{eq:W(II)_discrete} and \eqref{eq:simple}, and using the discrete version of the Laplace transform, we obtain
\begin{align}
\begin{aligned}
    P\left( -W_{\mathrm{(II)}} \leq \Delta F_{\mathrm{(II)}} - \epsilon_{\mathrm{(II)}} \right)
    &\leq \sum_{j\in J:\; x_j \leq \Delta F_{\mathrm{(II)}} - \epsilon_{\mathrm{(II)}}}
        e^{\beta\left( \Delta F_{\mathrm{(II)}} - \epsilon_{\mathrm{(II)}} - x_j \right)}
        P\left( -W_{\mathrm{(II)}} = x_j \right)\\
    &= e^{\beta\left( \Delta F_{\mathrm{(II)}} - \epsilon_{\mathrm{(II)}} \right)}
        \sum_{j \in J:\; x_j \leq \Delta F_{\mathrm{(II)}} - \epsilon_{\mathrm{(II)}}}
        e^{-\beta x_j}\, P\left( -W_{\mathrm{(II)}} = x_j \right)\\
    &\leq e^{\beta\left( \Delta F_{\mathrm{(II)}} - \epsilon_{\mathrm{(II)}} \right)}
        \sum_{j \in J}
        e^{-\beta x_j}\, P\left( -W_{\mathrm{(II)}} = x_j \right)\\
    &= e^{\beta\left( \Delta F_{\mathrm{(II)}} - \epsilon_{\mathrm{(II)}} \right)}
       \left\langle e^{-\beta\left( -W_{\mathrm{(II)}} \right)} \right\rangle \\
    &= e^{\beta\left( \Delta F_{\mathrm{(II)}} - \epsilon_{\mathrm{(II)}} \right)}
       \left\langle e^{\beta W_{\mathrm{(II)}}} \right\rangle.
\end{aligned}
\end{align}

Applying Eq.\ \eqref{eq:jarzynski_(II)}, which follows directly from the Jarzynski identity, we get
\begin{equation}
    P\left( -W_{\mathrm{(II)}} \leq \Delta F_{\mathrm{(II)}} - \epsilon_{\mathrm{(II)}} \right)
        \leq e^{\beta\left( \Delta F_{\mathrm{(II)}} - \epsilon_{\mathrm{(II)}} \right)} 
             e^{-\beta \Delta F_{\mathrm{(II)}}}
        = e^{-\beta \epsilon_{\mathrm{(II)}}},
\end{equation}
whence we finally obtain

\begin{equation}
    P\left( -W_{\mathrm{(II)}} \geq\Delta F_{\mathrm{(II)}} - \epsilon_{\mathrm{(II)}} \right)\geq P\left( -W_{\mathrm{(II)}} >\Delta F_{\mathrm{(II)}} - \epsilon_{\mathrm{(II)}} \right)
        = 1 - P\left( -W_{\mathrm{(II)}} \leq \Delta F_{\mathrm{(II)}} - \epsilon_{\mathrm{(II)}} \right)
        \geq 1 - e^{-\beta \epsilon_{\mathrm{(II)}}},
\end{equation}
which completes the proof.

\end{proof}

\subsection{\texorpdfstring{Work of $R_{\mathrm{(I)}}$ and $R_{\mathrm{(III)}}$}{Work of R 1 and R 3}}\label{sec:WI_III}
While analysing \textbf{scenario (b2)}, i.e.\ the case where  $p_{\mathrm{out}} < q_0$ and $p_{\mathrm{out}} < q_1$, described in Section~\ref{sec:idea} (including its special instance~(b2$^{\prime}$)), it will be sufficient to use the following general result:

\begin{lemma}\label{lem:easy}
Let $p_\mathrm{in},p_\mathrm{out}\in(0,1)$ (in fact, the general assumption of the Appendix \ref{sec:appendA}, namely the exclusion of the possibility of performing any operation \textbf{CO5}, gives $p_\mathrm{out}\in(0,1/2]$).
With positive probability of at least $1-p_\mathrm{in}$ we have $-W_\mathrm{(I)} \geq 0$, and with positive probability of at least $1-p_\mathrm{out}$ we have $-W_\mathrm{(III)} \geq 0$.
\end{lemma}
\begin{proof}
    It suffices to observe that in parts $R_\mathrm{(I)}$ and $R_\mathrm{(III)}$ of the path $R$, either the ground states or the excited states are constantly occupied. 
    Indeed, the first thermalization is the endpoint of $R_\mathrm{(I)}$, so if the ground (excited) state is occupied at the beginning of the path, which happens with probability $1-p_\mathrm{in}$ ($p_\mathrm{in}$), it is occupied for the entire length of $R_\mathrm{(I)}$. 
    Similarly, if the ground (excited) state is occupied at the beginning of  $R_\mathrm{(III)}$, which happens with probability $1-p_\mathrm{out}$ ($p_\mathrm{out}$), it is occupied for the entire length of $R_\mathrm{(III)}$. 
    We therefore get
    \begin{align}
        P\left(-W_\mathrm{(I)}\geq 0\right)\geq P\left(-W_\mathrm{(I)}=0\right)=1-p_\mathrm{in}\qquad\text{and}\qquad 
        P\left(-W_\mathrm{(III)}\geq 0\right)\geq P\left(-W_\mathrm{(III)}=0\right)=1-p_\mathrm{out},        
    \end{align}
    and the proof is completed.
\end{proof}

In the remaining cases, a slightly more elaborate argument is required.
We analyse separately the following situations:
\begin{itemize}
    \item $q_0 < q_{N-1} = p_\mathrm{out}$ (that is, \textbf{scenario~(a)} introduced in Section~\ref{sec:idea}),
    \item $p_{\mathrm{out}} < q_0$ and $q_1 < q_0$ (that is, \textbf{scenario~(b1)}, as well as its special instance~(b1$^{\prime}$), both characterised in Section~\ref{sec:idea}).
\end{itemize}
Before formulating the appropriate lemmas, let us define the constants
\begin{equation}
	\label{eq:eps3}
	\epsilon_\mathrm{(III)}\left(p\right) 
    \coloneqq 
    \frac{1}{\beta} \ln \left(\frac{1+e^{\beta E_0}}{1+e^{\beta E\left(p\right)}}\right)
\end{equation}
and 
\begin{equation}
	\label{eq:eps1}
	\epsilon_\mathrm{(I)}\left(p\right) \coloneqq 
    \frac{1}{\beta} \ln \left(\frac{1+e^{\beta E(p)}}{1+e^{\beta E_0}}\right)
\end{equation}
for any given probability value $p\in(0,1/2]$, where $p\mapsto E(p)$ is defined via \eqref{def:E(p)} (as a~reminder,  $g(E_i)\coloneqq q_i$ for $i\in\{0,1, \ldots, N-1\}$, and we also use $E(p_{\mathrm{out}})=E_{N-1}$). Importantly, one can easily observe that $\epsilon_{\mathrm{(III)}}(p) = -\epsilon_{\mathrm{(I)}}(p)$, which implies that the positivity of one immediately implies the negativity of the other.

\begin{lemma}
	\label{lem:w1}
    Let $p_\mathrm{in}\in(0,1)$ and $p_\mathrm{out}\in(0,1/2]$.
	If $q_0 < q_{N-1} = p_\mathrm{out}$, then, with positive probability of at least $\min\{p_{\mathrm{in}},1-p_{\mathrm{in}}\}$ we have $ -W_\mathrm{(I)} \geq \Delta F_\mathrm{(I)}$, and with positive probability of at least $p_\mathrm{out}$ we have 
    \begin{equation}
		  -W_\mathrm{(III)} 
		\geq \Delta F_\mathrm{(III)} +\epsilon_\mathrm{(III)}(p_{\mathrm{out}}).
    \end{equation}
    where 
\begin{align}\label{def:eps3_out}
\epsilon_\mathrm{(III)}(p_{\mathrm{out}})
= 
\frac{1}{\beta} \ln \left(\frac{1+e^{\beta E_0}}{1+e^{\beta E\left(p_{\mathrm{out}}\right)}}\right).
    \end{align}
    Moreover, $\epsilon_\mathrm{(III)}(p_{\mathrm{out}}) > 0$.
\end{lemma}
\begin{proof}
    As it was already observed in the proof of Lemma \ref{lem:easy}, in parts $R_\mathrm{(I)}$ and $R_\mathrm{(III)}$ of the path $R$ either the ground states or the excited states are constantly occupied. 
    Moreover, upon previously applying the \emph{Path Shrinking} method to a given path, the corresponding Gibbs curve is monotonic in both part (I) and part (III). 
    
    Let us first analyze part (I). 
    In case of an increasing Gibbs curve and decreasing ergo 
    % \mh{co to jest to "ergo" tutaj?} 
    energy ($\Delta F_\mathrm{(I)}<0$), we have $-W_\mathrm{(I)}=0>\Delta F_\mathrm{(I)}$ with probability $1-p_{\mathrm{in}}$ (cf. Fig. \ref{fig:fig6}(ii)). 
    In the opposite case (see Fig. \ref{fig:fig6}(i)) we, in turn, get $-W_\mathrm{(I)}=\Delta E_\mathrm{(I)}$ with probability $p_\mathrm{in}$. 
    Obviously, the energy difference $\Delta E_\mathrm{(I)}$ is equal to the rectangular area below the horizontal line at height $1$, and so one can easily observe that it is not smaller than the free energy difference $\Delta F_\mathrm{(I)}$, equal to the area below the corresponding part of the Gibbs curve (cf. Fig. \ref{fig:proofR}). 
    In this connection, we further have $-W_\mathrm{(I)}=\Delta E_\mathrm{(I)}\geq \Delta F_\mathrm{(I)}$ with probability $p_\mathrm{in}$. 
    The proof of the first part of this lemma is, therefore, complete.

    Proceeding to part (III), we first observe that, since $R_{\mathrm{(III)}}$ starts at the Gibbs curve (more precisely, at point $(E_{N-1},q_{N-1})=(E(p_{\mathrm{out}}),p_{\mathrm{out}})$), and ends above it (by assumption of this lemma), the Gibbs curve has to be decreasing in this part, with path $R_\mathrm{(III)}$ entirely above it (see Fig.\ \ref{fig:fig6}(iii)).  
    Note that, with probability $p_{\mathrm{out}}$, the excited level is occupied for the entire length of $R_\mathrm{ (III)}$. 
    As a consequence, with probability $p_{\mathrm{out}}$, we have 

%\textit{\textcolor{magenta}{[MS: it seams that in two formulas below we somehow use A6, A7 and potetialy A4 and A5, also some more clarification would be usefull.] \textbf{[HW-Ś: Ja nie widzę, żeby tyle odnośników było potrzebnych, w moim odczuciu (A18) wystarczy, ale jeśli czujecie, że dodanie czegokolwiek jest ważne, to tak zróbcie ;)]}}}
    
	\begin{equation}
	    -W_\mathrm{(III)} =  \Delta E_\mathrm{(III)} 
	    =  E_0 - E(p_{\mathrm{out}}) 
	    =  E_0 - E(p_{\mathrm{out}}) + \Delta F_\mathrm{(III)} - \Delta F_\mathrm{(III)},
	\end{equation}
	where (due to \eqref{eq:DeltaF(I)(II)(III)})

	\begin{equation}
		\Delta F_\mathrm{(III)} = 
        E_0-E\left(p_{\mathrm{out}}\right)
        - \frac{1}{\beta} \ln \left(\frac{1+e^{\beta E_0}}{1+e^{\beta E(p_{\mathrm{out}})}}\right),
	\end{equation}
	and so 
	\begin{equation}
		-W_\mathrm{(III)} = 
	    \Delta F_\mathrm{(III)} + \frac{1}{\beta} \ln \left(\frac{1+e^{-\beta E_0}}{1+e^{-\beta E(p_{\mathrm{out}})}}
    		\right)= \Delta F_\mathrm{(III)} + \epsilon_\mathrm{(III)}\left(p_{\mathrm{out}}\right).
	\end{equation}
    %with $\epsilon_\mathrm{(III)}(p_{\mathrm{out}})$ given by \eqref{eq:eps3}.  

It remains to justify that $\epsilon_{\mathrm{(III)}}(p_{\mathrm{out}})$ is strictly positive.
This follows immediately from the fact that, in part~(III), the Gibbs curve is decreasing while the ergo energy is increasing (in particular, $q_0 < p_{\mathrm{out}}$, and therefore $E_0 > E(p_{\mathrm{out}})$).
Consequently, $\epsilon_{\mathrm{(III)}}(p_{\mathrm{out}}) \geq 0$.
Moreover, equality can occur only in the degenerate case $E_0 = E(p_{\mathrm{out}})$, that is, when $p_{\mathrm{out}} = q_0$, which does not apply here.
\end{proof}

\renewcommand{\thesubfigure}{\roman{subfigure}}
\begin{figure}
    \subfloat[]{
        \includegraphics[width=0.48\linewidth]{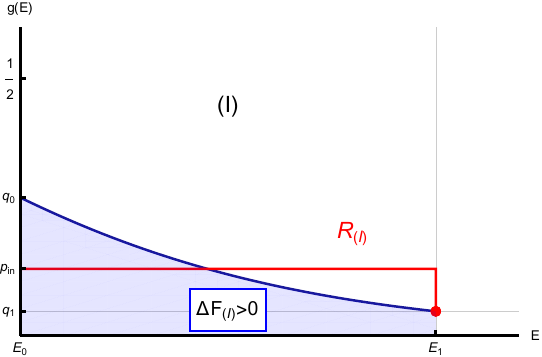}
    }
    \subfloat[]{
        \includegraphics[width=0.48\linewidth]{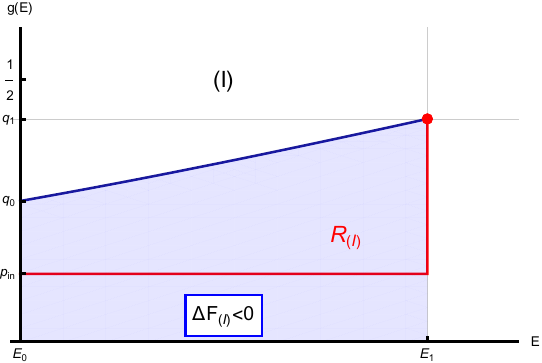}
    }\\
    \subfloat[]{
        \includegraphics[width=0.48\linewidth]{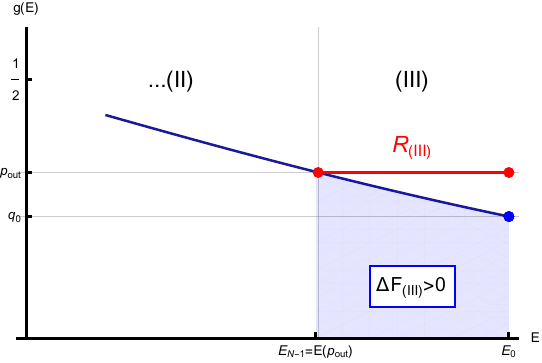}
    }
    \subfloat[]{
        \includegraphics[width=0.48\linewidth]{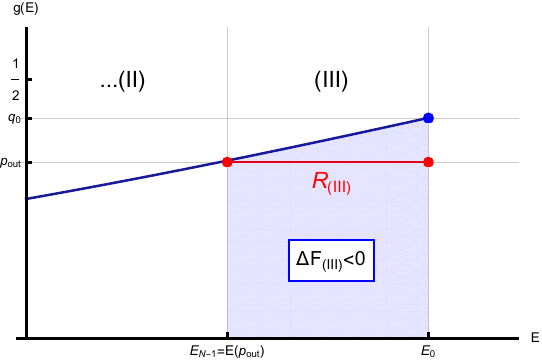}
    }
    \caption{
        \label{fig:fig6}
        Possible behaviors of a given path $R$ in its first and last part: 
        (i) the behavior of $R_{\mathrm{(I)}}$ in the case of decreasing Gibbs curve; 
        (ii) the behavior of  $R_{\mathrm{(I)}}$ in the case of increasing Gibbs curve; (iii) the behavior of $R_{\mathrm{(III)}}$ in the case of decreasing Gibbs curve; (iv) the behavior of $R_{\mathrm{(III)}}$ in the case of increasing Gibbs curve.  
        % {\color{red}{[HW-Ś: Tu jest generalnie ok koncepcyjnie, tylko dla mnie krzywa Gibbsa ma mało Gibbsowy kształt? Albo mi się wydaje. Wygląda na tych dolnych rysunkach prawie jak funkcja liniowa.]}}  
    }
\end{figure}

%\begin{figure}
%    {\includegraphics[width=0.9\linewidth]{figures/R(I)R(III).pdf
%    }}
%    \caption{\textcolor{magenta}{OLD image to be removed after verification of the new one. }
%        Possible behaviors of a given path $R$ in its first and last part: 
%        (i) the behavior of $R_{\mathrm{(I)}}$ in the case of decreasing Gibbs curve; 
%        (ii) the behavior of  $R_{\mathrm{(I)}}$ in the case of increasing Gibbs curve; (iii) the behavior of $R_{\mathrm{(III)}}$ in the case of decreasing Gibbs curve; (iv) the behavior of $R_{\mathrm{(III)}}$ in the case of increasing Gibbs curve. }
%    \label{fig:fig6}
%\end{figure}

\begin{lemma}
	\label{lem:w3}
	Let $p_\mathrm{in}\in(0,1)$ and $p_\mathrm{out}\in(0,1/2]$.
    If $p_{\mathrm{out}} < q_0$ (which is the opposite of what we assume in Lemma \ref{lem:w1}) and $q_1< q_0$, then, with positive probability of at least $(1-p_{\mathrm{out}})$ we have $ W_\mathrm{(III)} \leq - \Delta F_\mathrm{(III)}$, and with positive probability of at least  $p_\mathrm{in}$ we have 
     \begin{equation}\label{eq:WI}
    		 - W_\mathrm{(I)} 
		\geq \Delta F_\mathrm{(I)} +\epsilon_\mathrm{(I)}\left(q_1\right),
    \end{equation}
    where $\epsilon_\mathrm{(I)}(q_1)$  
    is given by 
    \begin{equation}\label{def:eps(I)(q1)}
        \epsilon_\mathrm{(I)}\left(q_1\right) \coloneqq 
    \frac{1}{\beta} \ln \left(\frac{1+e^{\beta E_1}}{1+e^{\beta E_0}}\right).
    \end{equation}
    Moreover, $\epsilon_\mathrm{(I)}(q_1) > 0$. 
\end{lemma}
\begin{proof}
The proof is analogous to that of Lemma~\ref{lem:w1}. 
Indeed, using similar arguments, we can justify that, since $p_\mathrm{out} < q_0$, the Gibbs curve in the third segment is increasing (see Fig.~\ref{fig:fig6}(iv)), and thus $\Delta F_\mathrm{(III)} < 0$. 
Moreover, under this assumption, we have $-W_\mathrm{(III)} = 0 > \Delta F_\mathrm{(III)}$ with probability $(1 - p_{\mathrm{out}})$.

Proceeding to part~(I), we note that the monotonicity of the Gibbs curve in this segment, together with the assumed condition $q_1 < q_0$, implies that it is decreasing (see Fig.~\ref{fig:fig6}(i)).  
Consequently, with probability $p_\mathrm{in}$, we have  
\begin{equation}
    -W_\mathrm{(I)} = \Delta E_\mathrm{(I)} = \Delta F_\mathrm{(I)} + \epsilon_\mathrm{(I)}(q_1).
\end{equation}

Finally, since the Gibbs curve is decreasing in part~(I) (in particular, $E_0 < E_1$), we conclude that $\epsilon_{\mathrm{(I)}}(q_1) \geq 0$.  
Equality occurs only in the degenerate case $E_0 = E_1$, i.e., when $q_0 = q_1$, which is not the case here.
\end{proof}

\begin{remark}
    Note that the inequalities $-W_{\mathrm{I}}\geq \Delta F_{\mathrm{I}}$ and $-W_{\mathrm{III}}\geq \Delta F_{\mathrm{III}}$ in Lemmas \ref{lem:w1} and \ref{lem:w2}, respectively, cannot be made strict, since we can thermalize just at the beginning (or respectively end), in which case part (I) (or respectively (III)) becomes trivial, and $-W_{\mathrm{(I)}}=\Delta F_{\mathrm{(I)}} = 0$ (respectively $-W_{\mathrm{(III)}}=\Delta F_{\mathrm{(III)}} = 0$). 
    % \textbf{\color{red}[HW-Ś: Właściwie to chyba nie jest prawda. Bo przecież zakłądamy $p_\mathrm{out}\neq q_0$, więc może powinniśmy zostawić komentarz tylko do części (I), a w lemacie dotyczącym części (III) dać ostrą nierówność? Czy dobrze myślę? {\color{blue}Maćku, Karolu, odnieście się proszę i nanieście stosowne poprawki, jeśli się ze mną zgadzacie.}]}
\end{remark}

\begin{remark}\label{rem:cases_with_primes}
As has already been highlighted in Section \ref{sec:idea}, cases (b1) and (b2) are not disjoint. The goal now is to redefine them so that they become disjoint. This procedure will -- most importantly --  eliminate the undesirable case of vanishing $\epsilon_{\mathrm{(I)}}$ in Lemma \ref{lem:w3} (when $q_1 \nearrow q_0$). Moreover, it will remove the dependence of the constant $\epsilon_{\mathrm{(I)}}$ on the parameter $q_1$, which is not fixed. Consequently, all constants in the subsequent statements will depend only on the given constants $q_0$, $p_{\mathrm{in}}$, and $p_{\mathrm{out}}$. 

The primary assumption in case (b) is that $p_{\mathrm{out}} < q_0$. We can therefore choose any point between these two values, that is, $q_* \in (p_{\mathrm{out}}, q_0)$; for example, let
\begin{equation}
q_* \coloneqq  \frac{p_\mathrm{out}+q_0}{2}.
\end{equation}
The redefined cases (b1) and (b2) are then given by:
\begin{itemize}
\item[(b1$^{\prime}$)] $p_{\mathrm{out}} < q_0$ and $q_1 \leq q_*$,
\item[(b2$^{\prime}$)] $p_{\mathrm{out}} < q_0$ and $q_1 > q_*$.
\end{itemize}
Let us note that cases (b1$^{\prime}$) and (b2$^{\prime}$) are particular instances of (b1) and (b2), respectively (see Fig. \ref{fig:cases_with_primes}). 
\end{remark}

Keeping in mind Remark \ref{rem:cases_with_primes}, we observe that Lemma \ref{lem:w3}, originally formulated under the assumptions of scenario (b1), also naturally holds in the case where $p_\mathrm{out}<q_0$ and $q_1 < (p_{\mathrm{out}}+q_0)/2$, which defines scenario (b1$^{\prime}$). As a~consequence, we obtain the following corollary of Lemma \ref{lem:w3}.

\begin{corollary}\label{corollary:w3_prime_case}
    Let $p_\mathrm{in}\in(0,1)$ and $p_\mathrm{out}\in(0,1/2]$. 
    If $p_{\mathrm{out}} < q_0$ and $q_1<(p_{\mathrm{out}}+q_0)/2$ (\textbf{case (b1$^{\prime}$)}), then with positive probability of at least $(1-p_{\mathrm{out}})$ we have $-W_{\mathrm{(III)}}\geq \Delta F_{\mathrm{(III)}}$, and with positive probability of at least $p_{\mathrm{in}}$ we have
    \begin{equation}
        -W_{\mathrm{(I)}}\geq \Delta F_{\mathrm{(I)}}+\overline{\epsilon}_{\mathrm{(I)}}\left(p_{\mathrm{out}}\right),
    \end{equation}
    where

     \begin{equation}\label{def:eps(I)}
        \overline{\epsilon}_{\mathrm{(I)}}\left(p_{\mathrm{out}}\right)\coloneqq\epsilon_\mathrm{(I)}\left(\tfrac{p_{\mathrm{out}}+q_0}{2}\right) = 
    \frac{1}{\beta} \ln \left(\frac{1+e^{\beta E\left(\tfrac{p_{\mathrm{out}}+q_0}{2}\right)}}{1+e^{\beta E_0}}\right).
    \end{equation}
    Moreover, $\overline{\epsilon}_{\mathrm{(I)}}(p_{\mathrm{out}}) > 0$.
\end{corollary}

\begin{proof}
   The proof of the first part of the assertion of Corollary~\ref{corollary:w3_prime_case} is identical to the corresponding argument in Lemma~\ref{lem:w3}. 
To prove the second part, observe first that the assumption $q_1 < (p_{\mathrm{out}}+q_0)/2$ clearly implies 
\begin{equation}
    q_1 <\tfrac{p_{\mathrm{out}}+q_0}{2}< q_0,
\end{equation}
and hence
\begin{equation}
    E_0 \leq E\left(\tfrac{p_{\mathrm{out}}+q_0}{2}\right) \leq E_1.
\end{equation}
Recalling the general definition of the function 
\(
    (0, q_0) \ni p \mapsto \epsilon_{\mathrm{(I)}}(p)
\)
(cf. Eq.~\eqref{eq:eps1}), we note that this function is decreasing, and takes positive values for $p\in(0,1/2]$ such that $E_0\leq E(p)$. 
Consequently,
\begin{equation}
    \epsilon_{\mathrm{(I)}}(q_1) > \epsilon_{\mathrm{(I)}}\left(\tfrac{p_{\mathrm{out}}+q_0}{2}\right)
        = \overline{\epsilon}_{\mathrm{(I)}}\!\left(p_{\mathrm{out}}\right)>0.
\end{equation}
Therefore,
\begin{equation}
    P\!\left(
        -W_{\mathrm{(I)}} 
        \geq \Delta F_{\mathrm{(I)}} 
             + \overline{\epsilon}_{\mathrm{(I)}}\!\left(p_{\mathrm{out}}\right)
    \right)
    \geq
    P\!\left(
        -W_{\mathrm{(I)}} 
        \geq \Delta F_{\mathrm{(I)}} 
             +\epsilon_{\mathrm{(I)}}(q_1)
    \right)
    \geq  p_{\mathrm{in}},
\end{equation}
which completes the proof.

\end{proof}

\subsection{Work of the entire path \texorpdfstring{$R$}{R}}\label{sec:entire_path}

We are now ready to estimate the work $W$ of the entire path $R$. 

Let us first analyze \textbf{case (a): $q_0<p_\mathrm{out}$}, pointed out in Section \ref{sec:idea}. 
Referring to Lemmas \ref{lem:w2} and \ref{lem:w1}, we obtain the following result:
\begin{theorem}\label{thm:case(a)altogether}
Let $p_\mathrm{in}\in(0,1)$ and $p_\mathrm{out}\in(0,1/2]$, and consider a path $R$ such that $q_0<p_\mathrm{out}$ (\textbf{case(a)}). Then 
\begin{equation}\label{def:eps(a)}
    \epsilon_{(a)}\left(p_\mathrm{out}\right)\coloneqq\frac{1}{2}\epsilon_\mathrm{(III)}\left(p_\mathrm{out}\right)=\frac{1}{2\beta} \ln \left(\frac{1+e^{\beta E_0}}{1+e^{\beta E\left(p_{\mathrm{out}}\right)}}\right)
\end{equation}
(cf. \eqref{def:eps3_out} in Lemma \ref{lem:w1}) is positive, and 
the work $W$ of $R$ fulfills

\begin{equation}
    -W\geq \epsilon_{(a)}(p_\mathrm{out})
\end{equation}
with a positive probability of at least
\begin{align}\label{def:p(a)}
    \eta_{(a)}\left(p_\mathrm{out}\right)\coloneqq\min\left\{p_\mathrm{in},1-p_\mathrm{in}\right\}p_\mathrm{out}\left(1- e^{-{\beta}\epsilon_{\mathrm{(a)}}\left(p_{\mathrm{out}}\right)}\right).
\end{align}
\end{theorem}
\begin{proof}
    Recall that a path $R$ consists of three parts $R = (R_\mathrm{(I)}, R_\mathrm{(II)}, R_\mathrm{(III)})$, and that we have denoted its work as \(W = W_\mathrm{(I)} + W_\mathrm{(II)} + W_\mathrm{(III)}\), where random variables $W_\mathrm{(I)}$, $W_\mathrm{(II)}$, $W_\mathrm{(III)}$ are independent. Moreover, the corresponding increment of free energy is denoted by \(\Delta F = \Delta F_\mathrm{(I)} + \Delta F_\mathrm{(II)} + \Delta F_\mathrm{(III)}\). 

%\textit{\textcolor{magenta}{[MS: I have some comment about conjunction of assumptions but unfortunately I do not remeber what it was about. KH maybe you can remember?] \textbf{[HW-Ś: ???]}}}

By assumption, we have $q_0< p_\mathrm{out}$, which allows us to use Lemma \ref{lem:w1}, saying that

\begin{align}\label{eq:w1w3_b1}
    P\left(-W_\mathrm{(I)}-W_\mathrm{(III)}\geq \Delta F_\mathrm{(I)}+\Delta F_\mathrm{(III)}+\epsilon_\mathrm{(III)}\left(p_\mathrm{out}\right)\right)\geq \min\left\{p_\mathrm{in},1-p_\mathrm{in}\right\}p_\mathrm{out}.
\end{align}

Applying further Lemma \ref{lem:w2} with $\epsilon_\mathrm{(II)}\coloneqq\epsilon_\mathrm{(III)}(p_\mathrm{out})/2$, we get
\begin{align}
    P\left(-W_\mathrm{(II)}\geq\Delta F_\mathrm{(II)}-\frac{\epsilon_\mathrm{(III)}\left(p_\mathrm{out}\right)}{2}\right)
    \geq     1- e^{-\tfrac{\beta}{2}\epsilon_{\mathrm{(III)}}\left(p_{\mathrm{out}}\right)}=1- e^{-{\beta}\epsilon_{\mathrm{(a)}}\left(p_{\mathrm{out}}\right)}.
\end{align}
Combining it with Eq.\ \eqref{eq:w1w3_b1} and with the fact that $\Delta F_\mathrm{(I)}+\Delta F_\mathrm{(II)}+\Delta F_\mathrm{(III)}=0$ (cf.\ Eq.\ \eqref{eq:deltaF}), we then obtain
\begin{align}\label{eq:1}
P\left(-W\geq \epsilon_{(a)}\left(p_\mathrm{out}\right)\right)=
    P\left(-W\geq \frac{\epsilon_\mathrm{(III)}\left(p_\mathrm{out}\right)}{2}\right)
    \geq  \min\left\{p_\mathrm{in},1-p_\mathrm{in}\right\}p_\mathrm{out}\left(1- e^{-{\beta}\epsilon_{\mathrm{(a)}}\left(p_{\mathrm{out}}\right)}\right)=\eta_{(a)}\left(p_\mathrm{out}\right),
\end{align}
which completes the proof.

In the end, let us also highlight that $\epsilon_{(a)}(p_\mathrm{out})$ given by \eqref{def:eps(a)} is positive due to Lemma \ref{lem:w1}.
\end{proof}

Let us now consider the opposite case, that is, an arbitrary path $R$ such that $p_\mathrm{out}<q_0$ (\textbf{case (b)}, pointed out in Section \ref{sec:idea}). It is here convenient to consider sub-cases
\begin{itemize}
\item[(b1$^{\prime}$)] $p_{\mathrm{out}} < q_0$ and $q_1 \leq  (p_{\mathrm{out}}+q_0)/2$,
\item[(b2$^{\prime}$)] $p_{\mathrm{out}} < q_0$ and $q_1 > (p_{\mathrm{out}}+q_0)/2$
\end{itemize}
separately. 

Referring firstly to Lemma \ref{lem:w2} and Corollary \ref{corollary:w3_prime_case}, and secondly to Lemmas \ref{lem:w2} and \ref{lem:easy} we obtain the following two results concerning \textbf{cases (b1$^{\prime}$) and (b2$^{\prime}$)}, respectively:

\begin{lemma}\label{lem:(b1)}
    Let $p_\mathrm{in}\in(0,1)$ and $p_\mathrm{out}\in(0,1/2]$. If $p_{\mathrm{out}} < q_0$ and $q_1 \leq  (p_{\mathrm{out}}+q_0)/2$ \textbf{(case (b1$'$))}, 
    then the constants 
\begin{equation}\label{def:eps(b1)}
    \epsilon_\mathrm{(b1)}\left(p_{\mathrm{out}}\right)\coloneqq \frac{1}{2}\overline{\epsilon}_{\mathrm{(I)}}(p_{\mathrm{out}})=\frac{1}{2\beta} \ln \left(\frac{1+e^{\beta E\left(\tfrac{p_{\mathrm{out}}+q_0}{2}\right)}}{1+e^{\beta E_0}}\right)
\end{equation}
(cf. Eq.\ \eqref{def:eps(I)} in Corollary \ref{corollary:w3_prime_case}) and 
    \begin{equation}\label{def:p(b1)}
        \eta_\mathrm{(b1)}\!\left(p_{\mathrm{out}}\right)
        \coloneqq 
        \left(1-p_{\mathrm{out}}\right) p_{\mathrm{in}}
        \left(1 - e^{-{\beta}\epsilon_\mathrm{(b1)}\left(p_{\mathrm{out}}\right)}\right)
    \end{equation}
    %\textcolor{magenta}{[MS: be careful about monotonicty of the above when considering swaps later]}
    are both positive, and moreover
    \begin{align}\label{eq:(b)1}
        P\!\left(-W \geq \epsilon_\mathrm{(b1)}\left(p_{\mathrm{out}}\right)\right)
        \geq \eta_\mathrm{(b1)}\!\left(p_{\mathrm{out}}\right).
    \end{align}
\end{lemma}

\begin{proof}
    If $p_{\mathrm{out}} < q_0$ and $q_1\leq (p_{\mathrm{out}}+q_0)/2$, we can apply Corollary \ref{corollary:w3_prime_case} to conclude that 
    \begin{align}
    P\left(-W_\mathrm{(I)}-W_\mathrm{(III)}\geq \Delta F_\mathrm{(I)}+\Delta F_\mathrm{(III)}+\overline{\epsilon}_\mathrm{(I)}\left(p_{\mathrm{out}}\right)\right)\geq \left(1-p_\mathrm{out}\right) p_\mathrm{in}.
\end{align}
Then, using Lemma \ref{lem:w2} with $\epsilon_\mathrm{(II)}\coloneqq\overline{\epsilon}_\mathrm{(I)}\left(p_{\mathrm{out}}\right)/2$, we obtain
\begin{align}
    P\left(-W_\mathrm{(II)}\geq\Delta F_\mathrm{(II)}-\frac{\overline{\epsilon}_\mathrm{(I)}\left(p_{\mathrm{out}}\right)}{2}\right)
    \geq     1-e^{-\tfrac{\beta}{2}\overline{\epsilon}_\mathrm{(I)}\left(p_{\mathrm{out}}\right)}=1 - e^{-{\beta}\epsilon_\mathrm{(b1)}\left(p_{\mathrm{out}}\right)}.
\end{align}
Combining it with Eq.\ \eqref{eq:w1w3_b1} and with the fact that $\Delta F_\mathrm{(I)}+\Delta F_\mathrm{(II)}+\Delta F_\mathrm{(III)}=0$ (cf.\ Eq.\ \eqref{eq:deltaF}), we finally get
\begin{align}
    P\left(-W\geq \epsilon_\mathrm{(b1)}\left(p_{\mathrm{out}}\right)\right)=
    P\left(-W\geq \frac{\overline{\epsilon}_\mathrm{(I)}\left(p_{\mathrm{out}}\right)}{2}\right)\geq  \left(1-p_\mathrm{out}\right)p_\mathrm{in}
\left(1 - e^{-{\beta}\epsilon_\mathrm{(b1)}\left(p_{\mathrm{out}}\right)}\right)=\eta_\mathrm{(b1)}\left(p_\mathrm{out}\right).
\end{align}
The proof is therefore completed.
\end{proof}

\begin{lemma}\label{lem:(b2)}
    Let $p_\mathrm{in}\in(0,1)$ and $p_\mathrm{out}\in(0,1/2]$. If $p_{\mathrm{out}} < q_0$ and $q_1>(p_{\mathrm{out}}+q_0)/2$ (\textbf{case (b2$^{\prime}$)}), then the constants
    \begin{equation}\label{def:eps(II)final}
        {\epsilon}_{\mathrm{(b2)}}\!\left(p_{\mathrm{out}}\right)
            \coloneqq \frac{1}{2}\left(E\!\left(p_{\mathrm{out}}\right) - E\left(\tfrac{p_{\mathrm{out}}+q_0}{2}\right)
           + \frac{1}{\beta} 
             \ln\!\left(
                 \frac{1 + e^{\beta E\left(\tfrac{p_{\mathrm{out}}+q_0}{2}\right)}}
                      {1 + e^{\beta E(p_{\mathrm{out}})}}
             \right)\right) 
    \end{equation}
    and 
    \begin{equation}\label{def:p(b2)}
        \eta_{(b2)}\left(p_\mathrm{out}\right)\coloneqq
\left(1-p_\mathrm{in}\right)\left(1-p_\mathrm{out}\right)\left(1-e^{-\beta{\epsilon}_{\mathrm{(b2)}}\!\left(p_{\mathrm{out}}\right)}\right)
    \end{equation}
    are both positive, and moreover
    \begin{equation}\label{eq:(b2)}
    P\left(-W\geq {\epsilon}_{\mathrm{(b2)}}\!\left(p_{\mathrm{out}}\right)\right)
    \geq  
    \eta_{(b2)}\left(p_\mathrm{out}\right).
\end{equation}
\end{lemma}
\begin{proof}
Let us first invoke the general result stated at the beginning of Section~\ref{sec:WI_III}, namely Lemma~\ref{lem:easy}, which yields
\begin{equation}\label{eq:W(I)(III)=0}
    P\!\left(-W_{\mathrm{(I)}} - W_{\mathrm{(III)}} \geq 0\right)
    \geq (1 - p_{\mathrm{in}})(1 - p_{\mathrm{out}}).
\end{equation}

Next, note that the considered case (b2$^{\prime}$) in particular implies the inequality $p_{\mathrm{out}} < q_1$, and under this assumption we have
\begin{equation}
    \Delta F_{\mathrm{(II)}}=\int_{E_1}^{E_{N-1}}g(E)\,dE=\int_{E_1}^{E\left(p_\mathrm{out}\right)}\tfrac{e^{-\beta E}}{1+e^{-\beta E}}\,dE > 0.
\end{equation}
Indeed, this follows immediately from the fact that the integral of the Gibbs curve over the interval $[E_1,E(p_{\mathrm{out}})]$ is strictly positive in this regime (provided that $E_1<E(p_{\mathrm{out}})$, which is the case here). Consequently, we may apply Lemma~\ref{lem:w2} with 
\[
    \epsilon_{\mathrm{(II)}} = \frac{1}{2}\Delta F_{\mathrm{(II)}},
\]
and therefore obtain
\begin{equation}\label{eq:estim_on_W(II)_using_DeltaF(II)}
    P\!\left( -W_{\mathrm{(II)}} \geq \tfrac{1}{2}\Delta F_{\mathrm{(II)}} \right)
        \geq 1 - e^{-\tfrac{\beta}{2}\Delta F_{\mathrm{(II)}}}.
\end{equation}
However, since we wish to keep both ${\epsilon}_{\mathrm{(b2)}}$ and 
$p_{(b2)}$ independent of $q_1$, we must provide an estimate of 
$\Delta F_{\mathrm{(II)}}$ that does not involve $q_1$ but remains strictly positive. 
By assumption, we have $q_1 > (p_{\mathrm{out}}+q_0)/2 > p_{\mathrm{out}}$, and therefore also 
\begin{equation}
    E(p_{\mathrm{out}}) > E\left(\frac{p_{\mathrm{out}}+q_0}{2}\right) > E_1.
\end{equation}
Consequently, we obtain 
\begin{equation}\label{eq:estim_DeltaF(II)}
\begin{aligned}
    \frac{1}{2}\Delta F_{\mathrm{(II)}}
        &\geq \frac{1}{2}\int_{E\left(\tfrac{p_{\mathrm{out}}+q_0}{2}\right)}^{E(p_{\mathrm{out}})}
                \frac{e^{-\beta E}}{1 + e^{-\beta E}}\, dE \\
        &= \frac{1}{2}\left(E\!\left(p_{\mathrm{out}}\right) - E\left(\frac{p_{\mathrm{out}}+q_0}{2}\right)
           + \frac{1}{\beta} 
             \ln\!\left(
                 \frac{1 + e^{\beta E\left(\tfrac{p_{\mathrm{out}}+q_0}{2}\right)}}
                      {1 + e^{\beta E(p_{\mathrm{out}})}}
             \right) \right)
        = {\epsilon}_{\mathrm{(b2)}}\!\left(p_{\mathrm{out}}\right)>0.
\end{aligned}
\end{equation}
The right-hand side of \eqref{eq:estim_DeltaF(II)} is strictly positive, as follows from the already mentioned fact that the integral of the Gibbs curve over the interval 
$[\,E((p_{\mathrm{out}}+q_0)/2),\, E(p_{\mathrm{out}})\,]$ is positive whenever $E((p_{\mathrm{out}}+q_0)/2) < E(p_{\mathrm{out}})$. Obviously, we then also get
\begin{equation}\label{eq:estim_probab}
    1-e^{-\tfrac{\beta}{2}\Delta F_\mathrm{(II)}}
        \geq 1-e^{-\beta{\epsilon}_{\mathrm{(b2)}}\left(p_\mathrm{out}\right)}>0.
\end{equation}
Combinig \eqref{eq:estim_on_W(II)_using_DeltaF(II)}, \eqref{eq:estim_DeltaF(II)} and \eqref{eq:estim_probab}, we obtain
\begin{align}\label{eq:W(II)_final}
    \begin{aligned}
    P\!\left( -W_{\mathrm{(II)}} \geq {\epsilon}_{\mathrm{(b2)}}\left(p_\mathrm{out}\right) \right)    
        \geq P\!\left( -W_{\mathrm{(II)}} \geq \tfrac{1}{2}\Delta F_{\mathrm{(II)}} \right)
                    \geq 1 - e^{-\tfrac{\beta}{2}\Delta F_{\mathrm{(II)}}}
            \geq 1-e^{-\beta{\epsilon}_{\mathrm{(b2)}}\left(p_\mathrm{out}\right)}.
    \end{aligned}
\end{align}

Finally, due to \eqref{eq:W(I)(III)=0} and \eqref{eq:W(II)_final}, we get
\begin{equation}%\label{eq:(b)2}
    P\left(-W\geq {\epsilon}_{\mathrm{(b2)}}\left(p_\mathrm{out}\right)\right)
        \geq  \left(1-p_\mathrm{in}\right)\left(1-p_\mathrm{out}\right)\left(1-e^{-\beta{\epsilon}_{\mathrm{(b2)}}\left(p_\mathrm{out}\right)}\right)
                    =\eta_{(b2)}\left(p_\mathrm{out}\right),
\end{equation}
which completes the proof.
\end{proof}

We are now ready to formulate a statement summarizing the analysis of the entire \textbf{scenario~(b)}, which is a direct and straightforward corollary of Lemmas~\ref{lem:(b1)} and~\ref{lem:(b2)}.
\begin{theorem}\label{thm:case(b)altogether}
Let $p_\mathrm{in}\in(0,1)$ and $p_\mathrm{out}\in(0,1/2]$. 
    For constants ${\epsilon}_\mathrm{(b1)}(p_\mathrm{out})$, ${\epsilon}_\mathrm{(b2)}(p_\mathrm{out})$, $p_{(b1)}(p_\mathrm{out})$, and $p_{(b2)}(p_\mathrm{out})$ given by \eqref{def:eps(b1)}, \eqref{def:eps(II)final}, \eqref{def:p(b1)}, and \eqref{def:p(b2)}, define 
    \begin{equation}\label{def:eps(b)}
        \epsilon_\mathrm{(b)}\left(p_\mathrm{out}\right)\coloneqq \min\left\{{\epsilon}_\mathrm{(b1)}\left(p_\mathrm{out}\right), {\epsilon}_\mathrm{(b2)}\left(p_\mathrm{out}\right)\right\},
    \end{equation}
    as well as
    \begin{equation}\label{def:p(b)}
        \eta_\mathrm{(b)}\left(p_\mathrm{out}\right)\coloneqq \min\left\{\eta_\mathrm{(b1)}\left(p_\mathrm{out}\right),\eta_\mathrm{(b2)}\left(p_\mathrm{out}\right)\right\}.
    \end{equation}

    Let $R$ be an arbitrary path such that $q_0>p_\mathrm{out}$ (\textbf{case(b)}). Then the constants $\epsilon_{(b)}(p_\mathrm{out})$ and $p_{(b)}(p_\mathrm{out})$ are both positive and 
the work $W$ of $R$ fulfills 
\begin{equation}
    -W\geq \epsilon_\mathrm{(b)}\left(p_\mathrm{out}\right)
\end{equation}
with a positive probability of at least $\eta_\mathrm{(b)}(p_\mathrm{out})$.
\end{theorem}
%\textcolor{magenta}{[MS: note to myself: check if the below constanmt are posiive with this smaller set of assumtions]}
\begin{proof}
   First of all, note that the constants $\epsilon_\mathrm{(b1)}(p_\mathrm{out})$ and $\epsilon_\mathrm{(b2)}(p_\mathrm{out})$, and thus also $\eta_\mathrm{(b1)}(p_\mathrm{out})$ and $\eta_\mathrm{(b2)}(p_\mathrm{out})$, are positive whenever $p_\mathrm{out}<q_0$, which is precisely the situation covered by scenario~(b). This follows directly from the definitions of these constants. Hence, the constants $\epsilon_\mathrm{(b)}(p_\mathrm{out})$ and $\eta_\mathrm{(b)}(p_\mathrm{out})$ are also positive in this case.

One can then easily observe that the probability
\begin{equation}
    P\!\left(-W \geq \epsilon_\mathrm{(b)}(p_\mathrm{out})\right)
    = P\!\left(
        -W \geq 
        \min\Bigl\{
            \epsilon_\mathrm{(b1)}(p_\mathrm{out}),
            \epsilon_\mathrm{(b2)}(p_\mathrm{out})
        \Bigr\}
    \right)
\end{equation}
can be lower bounded by either
\begin{equation}\label{term1}
    P\!\left(-W \geq \epsilon_\mathrm{(b1)}\left(p_\mathrm{out}\right)\right),
\end{equation}
or
\begin{equation}\label{term2}
    P\!\left(-W \geq \epsilon_\mathrm{(b2)}\left(p_\mathrm{out}\right)\right).
\end{equation}

Depending on which of the two mutually exclusive cases (b1$^{\prime}$) or (b2$^{\prime}$) holds, we can further lower bound either the term~\eqref{term1} -- by $p_{(b1)}(p_\mathrm{out})$ (according to Lemma~\ref{lem:(b1)}), or the term~\eqref{term2} -- by $p_{(b2)}(p_\mathrm{out})$ (according to Lemma~\ref{lem:(b2)}). In either scenario, we therefore obtain
\[
    P\!\left(-W \geq \epsilon_\mathrm{(b)}(p_\mathrm{out})\right)
    \;\geq\;
    \min\!\left\{
        \eta_\mathrm{(b1)}(p_\mathrm{out}),
        \eta_\mathrm{(b2)}(p_\mathrm{out})
    \right\}
    = \eta_\mathrm{(b)}(p_\mathrm{out}).
\]
This completes the proof.

\end{proof}

\subsection{Combining paths into process}\label{sec:paths_into_process}

In the previous sections, we have shown under which conditions and with which probability the work of a single path is negative.  
Now we are ready to combine multiple paths into a process. 
In the theorem below we state the main result of Appendix \ref{sec:appendA}. Let us recall that $q_i$ is given by \eqref{def:qi} for any $i\in\{0,\ldots,N-1\}$. Moreover, let us recall that  the constants $\epsilon_\mathrm{(a)}, \epsilon_\mathrm{(b)},\eta_\mathrm{(a)},\eta_\mathrm{(b)}$ are defined as follows (cf. Eqs. \eqref{def:eps(a)}, \eqref{def:eps(b)}, \eqref{def:p(a)}, \eqref{def:p(b)}):
\begin{align}  &\epsilon_\mathrm{(a)}\left(p\right)\coloneqq \frac{1}{2\beta} \ln \left(\frac{1+e^{\beta E_0}}{1+e^{\beta E\left(p\right)}}\right);\\        &\epsilon_\mathrm{(b)}\left(p\right)\coloneqq
    \frac{1}{2}\min\left\{\frac{1}{\beta} \ln \left(\frac{1+e^{\beta E\left(\tfrac{p+q_0}{2}\right)}}{1+e^{\beta E_0}}\right)\,,\,\int_{E\left(\tfrac{p+q_0}{2}\right)}^{E(p)}
                \frac{e^{-\beta E}}{1 + e^{-\beta E}}\, dE\right\}\\
    %\min\left\{\frac{1}{2}\overline{\epsilon}_\mathrm{(I)}\left(p_\mathrm{out}\right), \overline{\epsilon}_\mathrm{(II)}\left(p_\mathrm{out}\right)\right\}\\
    &\nonumber\hspace{9,2mm} =
    \frac{1}{2}\min\left\{\frac{1}{\beta} \ln \left(\frac{1+e^{\beta E\left(\tfrac{p+q_0}{2}\right)}}{1+e^{\beta E_0}}\right)\,,\, E\!\left(p\right) - E\left(\tfrac{p+q_0}{2}\right)
           + \frac{1}{\beta} 
             \ln\!\left(
                 \frac{1 + e^{\beta E\left(\tfrac{p+q_0}{2}\right)}}
                      {1 + e^{\beta E(p)}}
             \right)\right\};\\
&\eta_\mathrm{(a)}\left(p\right)\coloneqq
    \min\left\{p_\mathrm{in},1-p_\mathrm{in}\right\}p\left(1-e^{-\beta \epsilon_\mathrm{(a)}\left(p\right)}\right)=
    \min\left\{p_\mathrm{in},1-p_\mathrm{in}\right\}p\left(1-  \left(\frac{1+e^{\beta E\left(p\right)}}{1+e^{\beta E_0}}\right)^{1/2} \right);\\
&\eta_\mathrm{(b)}\left(p\right)
    %\coloneqq \min\left\{p_{(b1)}\left(p_\mathrm{out}\right),p_{(b2)}\left(p_\mathrm{out}\right)\right\}\\
    %&\hspace{16,2mm}=
    \coloneqq\min\left\{
        \left(1-p\right)
        p_{\mathrm{in}}
        \left(1 - \left(\frac{1+e^{\beta E_0}}{1+e^{\beta E\left(\tfrac{p+q_0}{2}\right)}}\right)^{1/2}\right)\, ,\right.\\ 
        &\hspace{25mm}\nonumber
        \left.\left(1-p_\mathrm{in}\right)\left(1-p\right)\left(1-e^{-\frac{\beta}{2}\left(E\!\left(p\right) - E\left(\tfrac{p+q_0}{2}\right)
          \right)}
          \left(\frac{1+e^{\beta E\left(p\right)}}{1+e^{\beta E\left(\tfrac{p+q_0}{2}\right)} }\right)^{1/2}
          \right) \right\}
\end{align}
for any $p\in(0,1/2]$. The above can be simplified thanks to substituting Eq.\ \eqref{def:E(p)}, that is the equality 
\begin{equation}
    E(p)=-\frac{1}{\beta}\ln\left(\frac{p}{1-p}\right)=\frac{1}{\beta}\ln\left(\frac{1-p}{p}\right).
\end{equation}
Indeed, we can write
\begin{align}
\label{def:epsilon(a)}    &\epsilon_\mathrm{(a)}\left(p\right)\coloneqq \frac{1}{2\beta}\ln\left(\frac{p}{q_0}\right);\\
\label{def:epsilon(b)}    &\epsilon_\mathrm{(b)}\left(p\right)\coloneqq
    \frac{1}{2\beta}\min\left\{\ln\left(\frac{2q_0}{p+q_0}\right),\ln\left(\frac{2(1-p)}{2-p-q_0}\right)\right\};\\    
\label{def:p(a)_}    &\eta_\mathrm{(a)}\left(p\right)\coloneqq
    \min\left\{p_\mathrm{in},1-p_\mathrm{in}\right\}p\left(1-\left(\frac{q_0}{p}\right)^{1/2}\right);\\   
\label{def:p(b)_}    &\eta_\mathrm{(b)}\left(p\right)
    \coloneqq\min\left\{(1-p)p_\mathrm{in}\left(1-\left(\frac{p+q_0}{2q_0}\right)^{1/2}\right),\left(1-p_\mathrm{in}\right)(1-p)\left(1-\left(\frac{p\left(2-p-q_0\right)}{(1-p)\left(p+q_0\right)}\right)^{1/2}\right)\left(\frac{p+q_0}{2p}\right)^{1/2}\right\}
\end{align}
for any $p\in(0,1/2]$.

%\textcolor{magenta}{[MS: dodać potem komentarz o zakresie p<q, ]}

\begin{theorem}
	\label{thm:main_appendix}
    Let $p_\mathrm{IN}\in(0,1)$. 
    For arbitrary $p_\mathrm{OUT}\in(0,1/2]$ and $q_0\in(0,1/2]$ 
    % (HW-Ś: Zgadza się ten zakres?\textcolor{magenta}{[MS: trzeba jeszcze raz sprawdzić czy tu też nie zakładamy p<q ze względu na A108.]} 
    such that $p_\mathrm{OUT}\not=q_0$ consider one of the scenarios (i)-(iv) listed at the beginning of Section \ref{sec:idea} (which exclude those cases in which a~transition from $\rho \coloneqq (1-p_{\mathrm{IN}}, p_{\mathrm{IN}})$ to $\sigma \coloneqq (1-p_{\mathrm{OUT}}, p_{\mathrm{OUT}})$ can be achieved with no work loss). 
    Then, while transforming the state $\rho \coloneqq (1-p_{\mathrm{IN}}, p_{\mathrm{IN}})$ into the state $\sigma \coloneqq (1-p_{\mathrm{OUT}}, p_{\mathrm{OUT}})$ by operations \textbf{CO2} and \textbf{CO3}, one has to spend at least the following amount of work
	\begin{equation}\label{def:epsilon_final}
        \epsilon\coloneqq \max\left\{\epsilon_\mathrm{(a)}\left(\frac{p_\mathrm{OUT}+q_0}{2}\right)
        \,,\,
        \epsilon_\mathrm{(b)}\left(\frac{p_\mathrm{OUT}+q_0}{2}\right)\right\}  > 0
    \end{equation}
    with positive probability of at least
    \begin{equation}\label{def:p_final}
        \eta\coloneqq \max\left\{\,
        \eta_\mathrm{(a)}\left(\frac{p_\mathrm{OUT}+q_0}{2}\right)\frac{p_\mathrm{OUT}-q_0}{2}\,,\,
        \eta_\mathrm{(b)}\left(\frac{p_\mathrm{OUT}+q_0}{2}\right)\frac{q_0-p_\mathrm{OUT}}{2}
        \,\right\}.
    \end{equation}
\end{theorem}
%\textcolor{magenta}{[MS: In my opinion somthing is wrong here. I do not see exclusion in the assumtiopns of the case where we can do the transformation without the work loss] \textbf{[HW-Ś: Maćku, zerknij, czy teraz już Ci się zgadza.]}}
\begin{proof}
	First, observe that an arbitrary process can be described by a mixture of paths $\{ R_j\}_{j \in I}$, each occurring with probability $\gamma_j$. 
    We will use $p_{\mathrm{in},j}$ and $p_{\mathrm{out},j}$ to denote the starting and ending points of any single path $R_j$. 
    For the process that transforms state $\rho \coloneqq (1-p_{\mathrm{IN}}, p_{\mathrm{IN}})$ into state $\sigma \coloneqq (1-p_{\mathrm{OUT}}, p_{\mathrm{OUT}})$  the mixture of paths $\{ R_i\}_{i \in I}$ has to fulfill 
	\begin{equation}
		\mathop\forall_{j \in I} p_{\mathrm{in},j} = p_{\mathrm{IN}}
	\end{equation}
	and
	\begin{equation}
		\sum_{j \in I} \gamma_j p_{\mathrm{out},j} = p_{\mathrm{OUT}}.
	\end{equation}
    As before, the symbols $W$ and $W_j$, $j \in I$, denote the work of the whole process and of the individual paths (whose mixture describes this process), respectively.
    
	Let $Y$ be the random variable defined as $P(Y = p_{\mathrm{out},j}) = \gamma_j$. Then
	\begin{equation}\label{eq:EY}
        \E Y = \sum_{j \in I} \gamma_j p_{\mathrm{out},j} = p_{\mathrm{OUT}}.
	\end{equation}
 
	Since, for all paths, we have $Y \in [ 0,1 ]$, we can use Markov's inequality in the following form (see Lemma \ref{lem:reverseMarkov}, Eqs.\ \eqref{eq:Markov1} and \eqref{eq:Markov2} for details):
	\begin{equation}
 P(Y<\alpha)\geq 1-\frac{\E Y}{\alpha}\quad\text{and}\quad
		P(Y > \alpha) \geq 1 - \frac{1 - \E Y}{1-\alpha}\quad\text{for any}\quad \alpha>0.
	\end{equation} 
	
 Let us choose 
 $\alpha\coloneqq (p_{\mathrm{OUT}} + q_0)/2$, which in \textbf{case (a)}: $q_0<p_\mathrm{OUT}$ gives us
\begin{equation}
        P\left (Y > \frac{p_{\mathrm{OUT}} + q_0}{2}\right ) 
        \geq 1 - \frac{1 - \E Y}{1-\frac{p_{\mathrm{OUT}} + q_0}{2}} 
        = 1 - \frac{1 - p_{\mathrm{OUT}}}{1-\frac{p_{\mathrm{OUT}} + q_0}{2}} 
        = \frac{p_{\mathrm{OUT}} - q_0}{2 - p_{\mathrm{OUT}} - q_0} 
        \geq \frac{p_{\mathrm{OUT}} - q_0}{2}>0,
	\end{equation}
    where the first equality follows from Eq.\ \eqref{eq:EY}. 
    %\textcolor{magenta}{[MS: note to myself: should we divide the paths into two separate cases and so on?]}
 We have already justified in Theorem \ref{thm:case(a)altogether} that $\epsilon_\mathrm{(a)}(p_{\mathrm{out},j})$ and $\eta_\mathrm{(a)}(p_{\mathrm{out},j})$ are positive for any $j\in I$, provided that $0<q_0<p_{\mathrm{out},j}\leq 1/2$. Let us then draw the attention of the reader to the fact that (within case (a)) for each $j\in I$ such that $1/2\geq p_{\mathrm{out},j} > (p_{\mathrm{OUT}} + q_0)/2>q_0>0$ we have 
 \begin{equation}     \epsilon_\mathrm{(a)}\left(p_{\mathrm{out},j}\right)>\epsilon_\mathrm{(a)}\left(\tfrac{p_\mathrm{OUT}+q_0}{2}\right)>0;\qquad         \eta_\mathrm{(a)}\left(p_{\mathrm{out},j}\right)>\eta_\mathrm{(a)}\left(\tfrac{p_\mathrm{OUT}+q_0}{2}\right)>0.     
 \end{equation}
which follows immediately from the definitions \eqref{def:epsilon(a)} and \eqref{def:p(a)} of $\epsilon_\mathrm{(a)}$ and $\eta_\mathrm{(a)}$, respectively (the functions $p\mapsto\epsilon_\mathrm{(a)}(p)$ and $p\mapsto\eta_\mathrm{(a)}(p)$ are increasing for $p\in(q_0,1/2)$). %Intuitively, we know that the bigger $p_{\mathrm{out},j}$, the bigger the work loss in the third part of the path $R_j$.

As a consequence, we know from Theorem \ref{thm:case(a)altogether} that for each path $R_j$, such that $1/2\geq p_{\mathrm{out},j} > (p_{\mathrm{OUT}} + q_0)/2>q_0>0$, 
 its work loss $-W_j$ is not smaller than
	\begin{equation}
    \epsilon_\mathrm{(a)}\left(\frac{p_{\mathrm{OUT}} + q_0}{2}\right)%=- \frac{\epsilon_\mathrm{(III)}\left(\frac{p_{\mathrm{out}} + q_0}{2}\right)}{2}>0, 
	 \end{equation}  
	with a positive probability of at least
	\begin{equation}
	    \eta_\mathrm{(a)}\left(\frac{p_{\mathrm{OUT}} + q_0}{2}\right)
     %=\min\left\{p_{\mathrm{in}},1-p_{\mathrm{in}}\right\} \left(\frac{p_{\mathrm{out}} + q_0}{2}\right)
%\frac{\epsilon_\mathrm{(III)}^3\left(\frac{p_{\mathrm{out}} + q_0}{2}\right)}{\left(\frac{32}{\beta}+4\epsilon_\mathrm{(III)}\left(\frac{p_{\mathrm{out}} + q_0}{2}\right)\right)\left(\frac{8}{\beta^2}+2\epsilon_\mathrm{(III)}^2\left(\frac{p_{\mathrm{out}} + q_0}{2}\right)\right)}
     ,
	 \end{equation}   
     %\textcolor{magenta}{[MS: we should write here some corespondence between $W$ and $W_j$ since notation is not very consistant] \textbf{[HW-Ś: Tu raczej jest potrzebne połączenie definicji $W$ i $Y$; $Y$ natomiast jest zdefiniowane za pomocą ścieżek $R_j$ tworzących rozważany proces, których prace oznaczamy symbolami $W_j$. -- Dodałam jedno zdanie w pierwszej części dowodu, być może będzie pomocne.]}}
  whence we finally get
 \begin{align}\label{eq:process(a)}
 \begin{aligned}
     P\left(-W\geq  \epsilon_\mathrm{(a)}\left(\frac{p_{\mathrm{OUT}} + q_0}{2}\right)\right)
     &\geq 
     P\left(-W\geq  \epsilon_\mathrm{(a)}\left(\frac{p_{\mathrm{OUT}} + q_0}{2}\right)\;\Big\lvert\;Y>\frac{p_\mathrm{OUT}+q_0}{2}\right)P\left(Y>\frac{p_\mathrm{OUT}+q_0}{2}\right)\\
     &\geq 
    \eta_\mathrm{(a)}\left(\frac{p_{\mathrm{OUT}} + q_0}{2}\right)\frac{p_\mathrm{OUT}-q_0}{2},
\end{aligned}
 \end{align}
 which is the reasonable (that is, non-negative) estimation, 
 provided that $q_0<p_\mathrm{OUT}$. \vspace{2mm}
 
 In \textbf{case (b)}: $p_\mathrm{OUT}<q_0$, we proceed analogously as in case (a). First of all, we have
 \begin{equation}
     P\left(Y<\frac{p_\mathrm{OUT}+q_0}{2}\right)\geq 1-\frac{p_\mathrm{OUT}}{\frac{p_\mathrm{OUT}+q_0}{2}} 
     =\frac{q_0-p_\mathrm{OUT}}{p_\mathrm{OUT}+q_0}>0.
 \end{equation} 
 %\textcolor{magenta}{[MS: we beliebethat should be inequality istead of equality in the equation above] \textbf{[HW-Ś: Faktycznie chyba był błąd rachunkowy. Czy teraz się zgadza?]}}
 Let us also recall that it is claimed in Theorem \ref{thm:case(b)altogether} that $\epsilon_\mathrm{(b)}(p_{\mathrm{out},j})$ and $\eta_\mathrm{(b)}(p_{\mathrm{out},j})$ are positive for any $j\in I$, provided that $1/2\geq q_0>p_{\mathrm{out},j}>0$. 
 Moreover, referring to the definitions \eqref{def:epsilon(b)} and \eqref{def:p(b)_} of $\epsilon_\mathrm{(b)}$ and $\eta_\mathrm{(b)}$, respectively, one can observe that (within case (b)) for each $j\in I$ such that $0<p_{\mathrm{out},j}<(p_\mathrm{OUT}+q_0)/2\leq 1/2$ we have 
  \begin{equation}     \epsilon_\mathrm{(b)}\left(p_{\mathrm{out},j}\right)>\epsilon_\mathrm{(b)}\left(\tfrac{p_\mathrm{OUT}+q_0}{2}\right)>0;\qquad           \eta_\mathrm{(b)}\left(p_{\mathrm{out},j}\right)>\eta_\mathrm{(b)}\left(\tfrac{p_\mathrm{OUT}+q_0}{2}\right)>0     
 \end{equation}
%Intuitively, we see that, under the assumptions of case (b2), the smaller $p_{\mathrm{ out},j}$, the greater the work loss in the middle part of path $R_j$ (the value of $p_{\mathrm{out},j}$ is not influencing the possible work loss in the first part of the path).
(the functions $p\mapsto\epsilon_\mathrm{(b)}(p)$ and $p\mapsto\eta_\mathrm{(b)}(p)$ are decreasing for $p\in(0,q_0)$ what can be shown for example by using symbolic computation software). 
% {\color{blue}(the functions $p\mapsto\epsilon_\mathrm{(b)}(p)$ and $p\mapsto\eta_\mathrm{(b)}(p)$ are decreasing for $p\in(0,q_0,1/2)$)}. 
Theorem  \ref{thm:case(b)altogether} then yields that for each path $R_j$ such that $0<p_{\mathrm{out},j}<(p_\mathrm{OUT}+q_0)/2\leq 1/2$, its work loss $-W_j$ is not smaller than
\begin{equation}
    \epsilon_\mathrm{(b)}\left(\frac{p_\mathrm{OUT}+q_0}{2}\right)
\end{equation}
with a positive probability of at least
\begin{equation}
	    \eta_\mathrm{(b)}\left(\frac{p_{\mathrm{OUT}} + q_0}{2}\right).
     \end{equation}
As a consequence, we finally get
\begin{align}\label{eq:process(b)}
\begin{aligned}
    P\left(-W\geq \epsilon_\mathrm{(b)}\left(\frac{p_\mathrm{OUT}+q_0}{2}\right)\right)
    &\geq
    P\left(-W\geq \epsilon_\mathrm{(b)}\left(\frac{p_\mathrm{OUT}+q_0}{2}\right)\;\Big\lvert\;Y<\frac{p_\mathrm{OUT}+q_0}{2}\right)P\left(Y<\frac{p_\mathrm{OUT}+q_0}{2}\right)\\
     &\geq 
    \eta_\mathrm{(b)}\left(\frac{p_{\mathrm{OUT}} + q_0}{2}\right)\frac{q_0-p_\mathrm{OUT}}{2},
\end{aligned}
\end{align}
which is the reasonable (that is, non-negative) estimation, 
 provided that $p_\mathrm{OUT}<q_0$. 

 To complete the proof is stuffices to observe that for arbitrairly fixed $p_\mathrm{OUT}$ and $q_0$ such that $p_\mathrm{OUT}\neq q_0$ only one of the values $\epsilon_\mathrm{(a)}((p_\mathrm{OUT}+q_0)/2)$ or $\epsilon_\mathrm{(b)}((p_\mathrm{OUT}+q_0)/2)$ can be positive (which is a consequence of a similar observation made after defining $\epsilon_{\mathrm{(I)}}$ and $\epsilon_{\mathrm{(III)}}$ in Eqs.~\eqref{eq:eps3} and~\eqref{eq:eps1}). To be more precise, $\epsilon_\mathrm{(a)}((p_\mathrm{OUT}+q_0)/2)>0$ if and only if $p_\mathrm{OUT}>q_0$, while  $\epsilon_\mathrm{(b)}((p_\mathrm{OUT}+q_0)/2)>0$ if and only if $p_\mathrm{OUT}<q_0$. Thus, $\epsilon$ given by \eqref{def:epsilon_final} can be expressed as
 \begin{equation}
     \epsilon=
\begin{cases}
\epsilon_\mathrm{(a)}\left(\frac{p_\mathrm{OUT}+q_0}{2}\right), & \text{when}\quad p_\mathrm{OUT}>q_0\quad\text{(case(a))},\\[2mm]
\epsilon_\mathrm{(b)}\left(\frac{p_\mathrm{OUT}+q_0}{2}\right), & \text{when}\quad  p_\mathrm{OUT}<q_0\quad\text{(case(b))}.
\end{cases}
 \end{equation}
 Using the same arguments, we also get
 \begin{equation}
     \eta=
\begin{cases}
\eta_\mathrm{(a)}\left(\frac{p_\mathrm{OUT}+q_0}{2}\right)\frac{p_\mathrm{OUT}-q_0}{2}, &  \text{when}\quad  p_\mathrm{OUT}>q_0\quad\text{(case(a))},\\[2mm]
\eta_\mathrm{(b)}\left(\frac{p_\mathrm{OUT}+q_0}{2}\right)\frac{q_0-p_\mathrm{OUT}}{2}, &  \text{when}\quad  p_\mathrm{OUT}<q_0\quad\text{(case(b))}.
\end{cases}
 \end{equation}
 for $\eta$ defined as in \eqref{def:p_final}.

 Combining Eqs.\ \eqref{eq:process(a)} and \eqref{eq:process(b)} then gives the desired assertion of the theorem.

\end{proof}

\section{Characterization of transitions possible by means of \texorpdfstring{$\mcoww$}{MCO-zero}}
\label{sec:appendix-with-swaps}

In this part of the appendix, we provide an almost complete characterization (excluding only the case where $\rho$ is a~pure ground state, i.e.\ when $\rho_\mathrm{IN} = 0$) of the transitions that can be implemented to transform $\rho$ into $\sigma$ at no work cost by means of Coarse Operations without auxiliary systems, i.e., by operations from the class $\mcoww$.  
This requires analyzing the most general setting in which \textbf{CO5}-type operations are allowed.
  
As discussed in Section~\ref{subsec:co5}, this entails considering paths that may include swaps.  
Such paths acquire a new feature, that is, a \emph{reflection} about the point $1/2$.  
This reflection can occur only when the excited energy level is set to $0$, that is, when the Gibbs curve takes the value $1/2$.  
Furthermore, we have argued that, after applying the \emph{Path Shrinking} procedure to a given path $R$, any \emph{swap down}  operation can occur exclusively in part~$(I)$ of that path. This follows from the fact that any thermalization brings the path below, or exactly to, the level $1/2$, and a subsequent swap can only move it upward (hence a \emph{swap down} can only occur before the first thermalization, that is, in part~$(I)$ of the path).

In general, a swap may occur between any two neighboring thermalizations. Crucially, multiple swaps between the same pair of thermalizations can always be simplified: an even number of swaps is equivalent to no swap at all, whereas an odd number of swaps reduces to a single effective swap. As a consequence, after applying the \emph{Path Shrinking} procedure, any two swaps are necessarily interlaced by thermalizations. Combining this observation with the argument above completes the reasoning and justifies that a \emph{swap down} can occur exclusively in part~$(I)$ of the path.

Now, we shall show that, in a certain sense, the new setup can be reduced to the old one. 
To this end, we will argue that Lemmas~\ref{lem:w2}--\ref{lem:w1} and Corollary \ref{corollary:w3_prime_case}, which constitute the building blocks for the remaining statements in Appendix~\ref{sec:appendA}, continue to hold (up to a few subtle details that do not affect the general argument, nor the final conclusion presented in Appendix~\ref{sec:appendA}).

Before presenting the analogues of Lemmas~\ref{lem:w2}--\ref{lem:w1} and Corollary \ref{corollary:w3_prime_case}, let us first examine how the shape of a~single path \( R \) may change in the presence of swaps.

As before, we discuss each part~(I), (II) and~(III) of a path \( R \) separately:
%Part I
\begin{itemize}[align=left]
\item[\textit{(I)}] \textit{Swap performed in part (I) (i.e., before the first thermalization).}
% Firstly, we will look at the situation when the swap operation is performed before the first thermalization (in the $R_{\mathrm{I}}$ part of the path). 
We can distinguish here three cases depending on the starting point of the path. 
The first in which $q_0 \leq p_{\mathrm{in}}\leq 1/2$ is presented in the Fig.\ \ref{fig:swap-RI-RIII}(i). 
The second one in which $p_{\mathrm{in}} > 1/2$ is displayed in Fig.\ \ref{fig:swap-RI-RIII}(ii). 
Finally, the third when $p_{\mathrm{in}} < q_0\leq1/2$ is depicted in the Fig.\ \ref{fig:swap-RI-RIII}(iii). 

\renewcommand{\thesubfigure}{\roman{subfigure}}
\begin{figure}
    \subfloat[]{
        \includegraphics[width=0.23\linewidth]{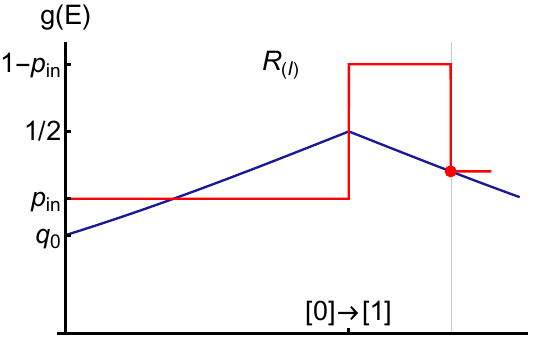}
    }
    \subfloat[]{
        \includegraphics[width=0.23\linewidth]{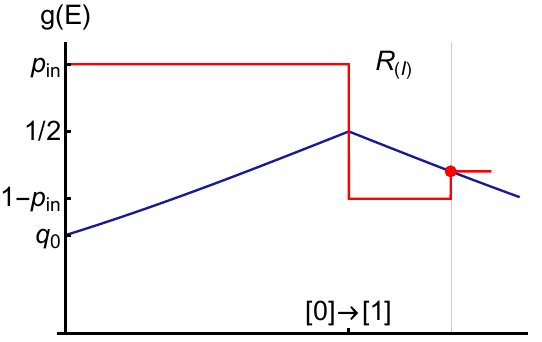}
    }
    \subfloat[]{
        \includegraphics[width=0.23\linewidth]{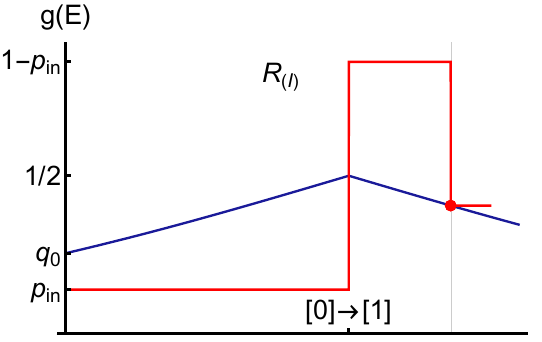}
    }
    \subfloat[]{
        \includegraphics[width=0.23\linewidth]{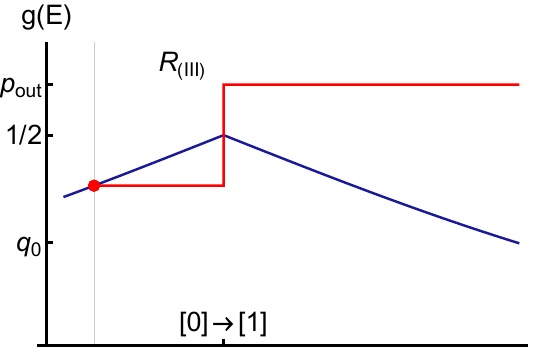}
    }
    \caption{
        Additional possibilities for the shapes of individual parts of a path \( R \) that arise due to the presence of swaps are illustrated. 
        For part~(I), we distinguish three possible cases:  
        (i) the case in which \( R_{\mathrm{(I)}} \) satisfies \( q_0 \leq p_{\mathrm{in}} \leq 1/2 \);  
        (ii) the case \( p_{\mathrm{in}} > 1/2 \);  
        (iii) the case \( p_{\mathrm{in}} < q_0 \).  
        For part~(III), there is a single relevant case:  
        (iv) \( p_{\mathrm{out}} \geq 1/2 \).
    }
	\label{fig:swap-RI-RIII}
\end{figure}

%\renewcommand{\thesubfigure}{\alph{subfigure}}
%\begin{figure}
%    \subfloat[]{
%        \includegraphics[width=0.23\linewidth]{figures/FigAdditionalPossibilities1.pdf}
%    }
%    \subfloat[]{
%        \includegraphics[width=0.23\linewidth]{figures/FigAdditionalPossibilities2.pdf}
%    }
%    \subfloat[]{
%        \includegraphics[width=0.23\linewidth]{figures/FigAdditionalPossibilities3.pdf}
%    }
%    \subfloat[]{
%        \includegraphics[width=0.23\linewidth]{figures/FigAdditionalPossibilities4.pdf}
%    }
%    \caption{\textcolor{magenta}{OLD image to be removed after verification of the new one. }
%        Additional possibilities for the shapes of individual parts of a path \( R \) that arise due to the presence of swaps are illustrated. 
%For part~(I), we distinguish three possible cases:  
%(i) the case in which \( R_{\mathrm{(I)}} \) satisfies \( q_0 \leq p_{\mathrm{in}} \leq 1/2 \);  
%(ii) the case \( p_{\mathrm{in}} > 1/2 \);  
%(iii) the case \( p_{\mathrm{in}} < q_0 \).  
%For part~(III), there is a single relevant case:  
%(iv) \( p_{\mathrm{out}} \geq 1/2 \).
%        {\color{red}[Ważne! Powinniśmy unikać oznaczeń (a) i~(b), które są ściśle przypisane do scenariuszy pojawiających się w całej pracy. 
%Proponuję w ich miejsce użyć oznaczeń (i)--(iv). Notacja: zamienić $q_\mathrm{out}$ i $p_\beta$ odpowiednio na $p_\mathrm{out}$ i $q_0$.]}
%    }
%	\label{fig:swap-RI-RIII}
%\end{figure}

%Part II
\item[\textit{(II)}]
\textit{Swap operation is performed in part (II).} 
Swaps may occur between any two thermalisations in part (II). An exemplary swap, taking place between the $i$-th and $(i+1)$-st thermalisation for some $i \in \{1,\ldots, N-2\}$, is shown in Fig.\,\ref{fig:W(II)swap}.

%Part III
\item[\textit{(III)}]
\textit{Swap is performed in part (III) (i.e., after the last thermalisation).}
In this setting, there is only one case that needs to be considered, namely when 
\( p_{\mathrm{out}} > \tfrac12 \) (which means that swaps are not possible in scenario~(b), as defined in Section~\ref{sec:idea}). 
This is because any swap performed after the last thermalisation necessarily moves the path above 
\( \tfrac12 \). 
Moreover, performing the swap at this stage is the only way to reach such a final state. 
This situation is illustrated in Fig.\,\ref{fig:swap-RI-RIII}(iv).
\end{itemize}

% arguments for all parts

\subsection{Work of \texorpdfstring{$R_{\mathrm{(II)}}$}{R(II)} for a path \texorpdfstring{$R$}{R} that possibly includes swaps}\label{sec:swaps(II)}

Let us recall that 
\(W_\mathrm{(II)} = \sum_{i=1}^{N-2} W_i\). 
For clarity of notation, we will use the notion of $R_i$ for $i \in \{1, \ldots, N-2\}$ to denote the single segment of $R_\mathrm{(II)}$ between the $i$-th and the $(i+1)$-st thermalization. 

After allowing for \textbf{CO5} operations, we shall remember that any path $R$ may now have both:  segments $R_i$ with swaps and those without them. More formally, we write
\begin{equation}
    \{1, \ldots, N-2\} = I_\mathrm{swaps} \cup I_\mathrm{no\text{-}swaps},
\end{equation}
where $I_\mathrm{swaps}$ denotes the set of indices $i \in \{1, \ldots, N-2\}$ corresponding to the segments $R_i$ on which a swap is performed, while $I_\mathrm{no\text{-}swaps}$ denotes the set of indices of those segments $R_i$ on which no swap is performed.

Our goal is to show that, even in the presence of swaps, the general Lemma~\ref{lem:w2} still applies. 

To this end, for any $i \in I_\mathrm{swaps}$, we introduce a random variable $W_i^{\mathrm{swap}}$ that describes the work performed when a~swap is applied between the $i$-th and the $(i+1)$-st thermalisations in part~(II) (see Fig. \ref{fig:W(II)swap}). Its distribution is given by
\begin{align}\label{def:W^swap}
    \begin{aligned}
        &P\left(W_i^{\mathrm{swap}}=-\Delta E_{i2}\right)=1-q_i\quad\text{with}\quad \Delta E_{i2}=E_{i+1}-0=E_{i+1}\\
        &P\left(W_i^{\mathrm{swap}}=-\Delta E_{i1}\right)=q_i\quad\text{with}\quad \Delta E_{i1}=0-E_{i}=-E_i.
    \end{aligned}
\end{align}
This means that whenever the ground (respectively, excited) level is occupied at the beginning of the $i$-th segment of the path, which occurs with probability $1-q$ (respectively, $q$), the swap -- that is, the reflection about the point $1/2$ -- causes the excited (respectively, ground) level to be occupied afterwards (cf. Fig. \ref{fig:W(II)swap}). Consequently, according to the definition of work, we obtain $W_i^{\mathrm{swap}} = -\Delta E_{i2}$ (respectively, $W_i^{\mathrm{swap}} = -\Delta E_{i1}$). Obviously $\Delta E_i=\Delta E_{i1}+\Delta E_{i2}=E_{i+1}-E_i$.

% \begin{figure}
%     {\includegraphics[width=0.9\linewidth]{figures/W(II)swap.pdf
%     }}
%     \caption{\color{red}(1) An exemplary $i$th segment $R_i$ of a path $R$, including a swap (\emph{swap up}), with the corresponding average work $\langle W_i^{\mathrm{swap}} \rangle$, where $i \in \{1,\ldots,N-2\}$. 
% (2) The $i$th segment of an auxiliary path corresponding to the original one depicted in~(1), in the sense that the energy levels remain the same, but the swap is removed. The corresponding average work $\langle W_i^{A} \rangle$ is indicated by the shaded green region. 
% (3) The $i$th segment of an auxiliary path corresponding to the original one depicted in~(1), in the sense that the energy levels remain the same, but the swap is replaced by a standard thermalization. The corresponding average work $\langle W_i^{A} \rangle$ is indicated by the shaded pink region. \textbf{[HW-S: Ważne! Tam chyba powinny być średnie prace $\langle W_i\rangle$ zamiast po prostu zmiennych losowych $W_i$. Nadmiar tekstu z obrazka można usunąć.]}}
% \mh{Tak, to srednie prace - bo to sa te prostokąciki - zmiana energii razy prawdopodobienstwo to sreadnia. }
%     \label{fig:W(II)swap}
% \end{figure}

\begin{figure}
    \centering
    \includegraphics[width=1\linewidth]{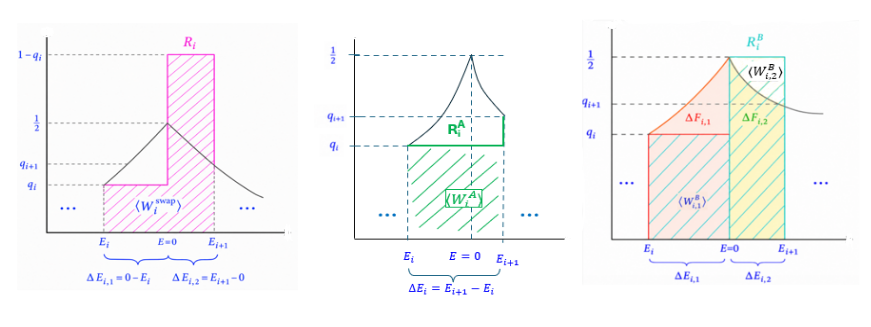}
    \caption{
    (1) An exemplary $i$th segment $R_i$ of a path $R$, including a swap (\emph{swap up}), with the corresponding average work $\langle W_i^{\mathrm{swap}} \rangle$, where $i \in \{1,\ldots,N-2\}$. 
(2) The $i$th segment of an auxiliary path corresponding to the original one depicted in~(1), in the sense that the energy levels remain the same, but the swap is removed. The corresponding average work $\langle W_i^{A} \rangle$ is indicated by the shaded green region. 
(3) The $i$th segment of an auxiliary path corresponding to the original one depicted in~(1), in the sense that the energy levels remain the same, but the swap is replaced by a standard thermalization. The corresponding average work $\langle W_i^{A} \rangle$ is indicated by the shaded pink region. 
    }
    \label{fig:W(II)swap}
\end{figure}
% \begin{figure}
%     \centering
%     \includegraphics[width=1\linewidth]{fig11-better.png}
%     \caption{
%     (1) An exemplary $i$th segment $R_i$ of a path $R$, including a swap (\emph{swap up}), with the corresponding average work $\langle W_i^{\mathrm{swap}} \rangle$, where $i \in \{1,\ldots,N-2\}$. 
% (2) The $i$th segment of an auxiliary path corresponding to the original one depicted in~(1), in the sense that the energy levels remain the same, but the swap is removed. The corresponding average work $\langle W_i^{A} \rangle$ is indicated by the shaded green region. 
% (3) The $i$th segment of an auxiliary path corresponding to the original one depicted in~(1), in the sense that the energy levels remain the same, but the swap is replaced by a standard thermalization. The corresponding average work $\langle W_i^{A} \rangle$ is indicated by the shaded pink region. 
% % \textbf{[HW-S: Ważne! Tam chyba powinny być średnie prace $\langle W_i\rangle$ zamiast po prostu zmiennych losowych $W_i$. Nadmiar tekstu z obrazka można usunąć.]}
%     }
%     \label{fig:W(II)swap}
% \end{figure}

Now, to present our main argument, it is convenient to introduce two auxiliary random variables, $W_i^{\mathrm{A}}$ and $W_i^{\mathrm{B}}$ (see Fig. \ref{fig:W(II)swap}).  

They relate to $W_i^{\mathrm{swap}}$ in the following way:  
$W_i^{\mathrm{A}}$ describes the work performed in the $i$-th segment of the path after modifying it so that no swap is applied, whereas $W_i^{\mathrm{B}}$ describes the work performed when the swap is replaced by an ordinary thermalisation at $E=0$.  
(Of course, we may write $W_i^{\mathrm{B}} = W_{i1}^{\mathrm{B}} + W_{i2}^{\mathrm{B}}$, where $W_{i1}^{\mathrm{B}}$ and $W_{i2}^{\mathrm{B}}$ are independent random variables describing the work performed by the modified path over the segments $[E_i,0]$ and $[0,E_{i+1}]$, respectively.) 
The corresponding distributions are then given by
\begin{align}
    &\label{def:W^A}
    P\left(W_i^{\mathrm{A}}=0\right)=1-q_i,\quad\; P\left(W_i^{\mathrm{A}}=-\Delta E_i\right)=q_i\\
    &\label{def:W^B}
    P\left(W_i^{\mathrm{B}}=0\right)=\frac{1-q_i}{2},\quad\; P\left(W_i^{\mathrm{B}}=-\Delta E_{i2}\right)=\frac{1-q_i}{2},\quad\;
    P\left(W_i^{\mathrm{B}}=-\Delta E_{i1}\right)=\frac{q_i}{2},\quad\;
    P\left(W_i^{\mathrm{B}}=-\Delta E_{i1}-\Delta E_{i2}=-\Delta E_i\right)=\frac{q_i}{2}.
\end{align}
%We are now ready to explain why the assertion of Lemma~\ref{lem:w2} remains valid even after extending the setting to paths that include swaps. The key part of our justification is provided in the forthcoming remark.

We are now ready to explain why the assertion of Lemma~\ref{lem:w2} remains valid even after extending the setting to paths that include swaps, as addressed in the forthcoming lemma. 
\begin{lemma}\label{lem:w2_swaps}
    Let $\epsilon _\mathrm{(II)} > 0$, and consider a path $R$ that may include an arbitrary number of swaps between thermalizations in part (II). With probability of at least
    \begin{equation}
        1-e^{-\beta \epsilon_\mathrm{(II)}},
    \end{equation}
    we have
    \begin{equation}
        -W_{\mathrm{(II)}} \geq \Delta F_{\mathrm{(II)}} - \epsilon_{\mathrm{(II)}}.
    \end{equation}
\end{lemma}
\begin{proof}
In order to repeat the proof of Lemma~\ref{lem:w2} for an arbitrary path that may include swaps, it suffices to establish an analogue of Eq.~\eqref{eq:jarzynski_(II)}.  
The remainder of the proof can then be rewritten directly from the proof of Lemma~\ref{lem:w2}.  
To this end, observe that
\begin{align}\label{eq:Jarzynski(II)swaps}
\begin{aligned}
    \left\langle e^{\beta W_\mathrm{(II)}} \right\rangle
        &= \left\langle 
            e^{\beta\left( \sum_{i\in I_\mathrm{swaps}} W_i^{\mathrm{swap}}
            + \sum_{i\in I_\mathrm{no\text{-}swaps}} W_i \right)}
           \right\rangle 
        = \left\langle 
            \left( \prod_{i\in I_\mathrm{swaps}} e^{\beta W_i^{\mathrm{swap}}} \right)
            \left( \prod_{i\in I_\mathrm{no\text{-}swaps}} e^{\beta W_i} \right)
           \right\rangle \\
        &= \left( \prod_{i\in I_\mathrm{swaps}} 
                \left\langle e^{\beta W_i^{\mathrm{swap}}} \right\rangle \right)
           \left( \prod_{i\in I_\mathrm{no\text{-}swaps}} 
                \left\langle e^{\beta W_i} \right\rangle \right),
\end{aligned}
\end{align}
where the last equality follows from the fact that the family   
\(
\{ W_i \}_{i\in I_\mathrm{no\text{-}swaps}}
\cup 
\{ W_i^{\mathrm{swap}} \}_{i\in I_\mathrm{swaps}}
\)
consists of mutually independent random variables.

Now, recalling the definitions \eqref{def:W^swap}--\eqref{def:W^B}, for any $i\in I_\mathrm{swaps}$ we obtain 
\begin{equation}
    \left\langle e^{\beta W_i^{\mathrm{swap}}} \right\rangle
        = e^{-\beta \Delta E_{i2}} (1 - q_i)
          + e^{-\beta \Delta E_{i1}} q_i,
\end{equation}
and likewise
\begin{align}
\begin{aligned}
    \left\langle e^{\beta W_i^{\mathrm{A}}} \right\rangle
        &= e^{0}\,(1 - q_i)
           + e^{-\beta(\Delta E_{i1} + \Delta E_{i2})} q_i,\\
    \left\langle e^{\beta W_i^{\mathrm{B}}} \right\rangle
        &= \tfrac{1}{2}\left(
                e^{0}(1 - q_i)
              + e^{-\beta \Delta E_{i2}}(1 - q_i)
              + e^{-\beta \Delta E_{i1}} q_i
              + e^{-\beta(\Delta E_{i1} + \Delta E_{i2})} q_i
           \right).
\end{aligned}
\end{align}
This immediately yields the identity
\begin{equation}
    \left\langle e^{\beta W_i^{\mathrm{swap}}} \right\rangle
        = 2\,\left\langle e^{\beta W_i^{\mathrm{B}}} \right\rangle
          - \left\langle e^{\beta W_i^{\mathrm{A}}} \right\rangle,\quad i\in I_\mathrm{swaps}.
\end{equation}
Further, note that the equality 
\(W_i^{\mathrm{B}} = W_{i1}^{\mathrm{B}} + W_{i2}^{\mathrm{B}}\),
as well as the independence of the random variables 
\(W_{i1}^{\mathrm{B}}\) and \(W_{i2}^{\mathrm{B}}\), imply the following:
\begin{equation}
    \left\langle e^{\beta W_i^{\mathrm{swap}}} \right\rangle
        = 2\left\langle 
            e^{\beta W_{i1}^{\mathrm{B}}}\,
            e^{\beta W_{i2}^{\mathrm{B}}}
          \right\rangle
          - \left\langle e^{\beta W_i^{\mathrm{A}}} \right\rangle
        = 2\left\langle e^{\beta W_{i1}^{\mathrm{B}}} \right\rangle
          \left\langle e^{\beta W_{i2}^{\mathrm{B}}} \right\rangle
          - \left\langle e^{\beta W_i^{\mathrm{A}}} \right\rangle,\quad i\in I_\mathrm{swaps}.
\end{equation}
Applying the Jarzynski identity \eqref{eq:jarzynski} (for a single step, between thermalisations, and in the absence of swaps) leads, in turn, to
\begin{equation}\label{eq:after_Jarzynski_swaps}
    \left\langle e^{\beta W_i^{\mathrm{swap}}} \right\rangle
        = 2 e^{-\beta \Delta F_{i1}} e^{-\beta \Delta F_{i2}}
          - e^{-\beta \Delta F_i}
        = e^{-\beta \Delta F_i}\quad \text{for all}\quad i\in I_\mathrm{swaps}.
\end{equation}
% \textbf{\textit{\color{red}[HW-Ś: If we include an appropriate elaboration of the Jarzynski identity, we can justify why it applies to all types of steps in our setting, including those involving swaps. 
% This would allow us to omit the above reasoning entirely and obtain the desired conclusion immediately. 
% For now, however, I keep the current explanation, as it is useful for our understanding and provides a clear, explicit argument.
% ]}}
% \mh{MOzna i po recenzach 
% tak zrobimy :) teraz nie mam sily :(}

Obviously, we also have
\begin{equation}\label{eq:after_Jarzynski}
    \left\langle e^{\beta W_i} \right\rangle
        = e^{-\beta \Delta F_i} \qquad \text{for all}\; i\in I_\mathrm{no\text{-}swaps},
\end{equation}
which is precisely the statement of \eqref{eq:jarzynski}.

As a consequence of Eqs. \eqref{eq:Jarzynski(II)swaps}, \eqref{eq:after_Jarzynski_swaps} and \eqref{eq:after_Jarzynski}, we finally obtain
\begin{align}
\begin{aligned}
    \left\langle e^{\beta W_\mathrm{(II)}} \right\rangle
        &= \left( \prod_{i\in I_\mathrm{swaps}} 
                \left\langle e^{\beta W_i^{\mathrm{swap}}} \right\rangle \right)
           \left( \prod_{i\in I_\mathrm{no\text{-}swaps}} 
                \left\langle e^{\beta W_i} \right\rangle \right)
        =\left( \prod_{i\in I_\mathrm{swaps}} 
                e^{-\beta \Delta F_i} \right)
           \left( \prod_{i\in I_\mathrm{no\text{-}swaps}} 
                e^{-\beta \Delta F_i} \right)\\
        &=\left(e^{-\beta \sum_{i\in I\mathrm{swaps}}\Delta F_i}\right)\left(e^{-\beta \sum_{i\in I\mathrm{no\text{-}swaps}}\Delta F_i}\right)
        =e^{-\beta\sum_{i=1}^{N-2}\Delta F_i}=e^{-\beta\Delta F_\mathrm{(II)}},
\end{aligned}
\end{align}
which completes the proof.
\end{proof}

% \textbf{\textit{\color{purple}[Piotr: jedna uwaga tutaj chyba jest potrzebna dla jasności: Note that the swap operation corresponds to a permutation of energy levels, 
% and therefore does not change the partition function. 
% As a result, the Jarzynski equality continues to hold for such steps. - bo swap to nie jest mapa zachowująca stan Gibbsa]}}
% \mh{tutaj partition function i Gibbs to niemal to samo. swap nie zmienia Gibbsa, bo go robimy tylko kiedy jest degeneracja - energie sa takie same, i wtedy Gibbs to równa mieszanka, więc swap jej nie zmienia. Trzeba w tekscie glownym dopisac, ale juz nie tutaj. }

\subsection{Work of \texorpdfstring{$R_{\mathrm{(I)}}$}{R(II)} and \texorpdfstring{$R_{\mathrm{(III)}}$}{R(III)} for a path \texorpdfstring{$R$}{R} that possibly includes swaps}
\label{sec:swaps(I)(III)}

We have already discussed, at the beginning of Appendix~\ref{sec:appendix-with-swaps}, the possible scenarios involving swaps in parts~(I) and~(III) of an arbitrary path~$R$. 
We now need to understand how these new scenarios affect the original cases~(a) and~(b) (consisting of two mutually exclusive subcases, namely (b1$^{\prime}$) and (b2$^{\prime}$)), introduced in Section~\ref{sec:idea}, and used throughout Appendix~\ref{sec:appendA}. 
To this end, and as already mentioned above, we will prove analogues of Lemmas~\ref{lem:easy}--\ref{lem:w1} and Corollary \ref{corollary:w3_prime_case}, which describe the behaviour of parts~(I) and~(III) of a path, now adapted to the extended setup that allows for swaps.

%\begin{figure}
%    {\includegraphics[width=0.9\linewidth]{figures/scenario(b2)swaps.pdf
%    }}
%    \caption{\color{red}Układ obrazków wymaga jeszcze przemyślenia -- wydaje się, że wystarczy jeden rysunek, który uchwyci zarówno idee przedstawione tutaj, jak i te z Fig.~\ref{fig:swap-RI-RIII}, aczkolwiek jestem za tym, żeby porobić wersje adekwatne dla różnych scenariuszy (a), (b1$^\prime$), (b2$^\prime$).}
%    \label{fig:scenario(b2)swaps}
%\end{figure}

To address \textbf{case~(b2$^{\prime}$)}, i.e.\ $p_\mathrm{out}<q_0$ and  
$q_1>(q_0+p_\mathrm{out})/2$% (see Fig. \ref{fig:scenario(b2)swaps})
, in the final theorem, it suffices to establish the following general statement (cf.\ Lemma~\ref{lem:easy} and the subsequent analysis in Lemma~\ref{lem:(b2)} and Theorems \ref{thm:case(b)altogether}, \ref{thm:main_appendix}), which remains valid in all possible scenarios, whether involving swaps or not:
\begin{lemma}\label{lem:easy_swaps}
    Let $p_\mathrm{in},p_\mathrm{out}\in(0,1)$. Suppose that $R$ is an arbitrary path that may possibly include swaps; in particular, swaps may occur in parts~(I) and~(III). 
Then, with probability of at least $1-p_\mathrm{in}$ we have $-W_\mathrm{(I)}\ge 0$, and with probability of at least $1-p_\mathrm{out}$ we have $-W_\mathrm{(III)}\ge 0$.
\end{lemma}

\begin{proof}
    Let us discuss parts~(I) and~(III) separately.

First, note that when there is no swap before the first thermalization, Lemma~\ref{lem:easy} applies directly, and there is nothing to prove for part (I). Otherwise, independently of the type of a~swap applied (all three possibilities are depicted in Fig.~\ref{fig:swap-RI-RIII}(i)--(iii)), with probability $(1-p_\mathrm{in})$ the ground state is occupied at the beginning of part~(I) over the segment $[E_0,0]$ (up to the swap), and afterwards the excited level is occupied over the segment $[0,E_1]$.
Importantly, the Gibbs curve is decreasing along this second segment (cf. Fig. \ref{fig:swap-RI-RIII}(i)--(iii)). Consequently, with probability $(1-p_\mathrm{in})$ we have
\begin{equation}
    -W_\mathrm{(I)}
    = 0 + (E_1-0)
    = E_1 > 0,
\end{equation}
which implies
\begin{equation}
    P\!\left(-W_\mathrm{(I)} \ge 0\right) \ge 1 - p_\mathrm{in}.
\end{equation}

For part~(III), we may again invoke Lemma~\ref{lem:easy} if no swap occurs after the last thermalization (note that the assumption $p_\mathrm{out}<q_0$, which defines scenario~(b), excludes the possibility of a swap in part~(III)). Otherwise (which is only possible within scenario~(a)), with probability $(1-p_\mathrm{out})$ the ground state is occupied at the beginning of part~(III) over the segment $[E_{N-1},0]$ (up to the swap), and afterwards the excited level is occupied over the segment $[0,E_0]$ (see Fig.~\ref{fig:swap-RI-RIII}(iv)). Once more, the Gibbs curve is decreasing along this second segment.
Consequently, with probability $(1-p_\mathrm{out})$ we have
\begin{equation}
    -W_\mathrm{(III)}
    = 0 + (E_0-0)
    = E_0 > 0,
\end{equation}
which implies
\begin{equation}
    P\!\left(-W_\mathrm{(III)} \ge 0\right) \ge 1 - p_\mathrm{out}.
\end{equation}

The proof is therefore completed.

\end{proof}

Let us now proceed to \textbf{case~(b1$^{\prime}$)}, i.e.\ $p_\mathrm{out}<q_0$ and  
$q_1<(q_0+p_\mathrm{out})/2<q_0$ (see Fig.~\ref{fig:scenario(b1)swaps}). 
In order to establish the final theorem for an arbitrary path $R$, even if it involves swaps, we need to understand how the possibly appearing swaps affect the statement of Corollary~\ref{corollary:w3_prime_case}, as well as the subsequent results, namely Lemma~\ref{lem:(b1)} and Theorem \ref{thm:case(b)altogether}. We claim that
\begin{lemma}\label{corollary:w3_prime_case_swaps}
    Let $p_\mathrm{in},p_\mathrm{out}\in(0,1)$.  Suppose that $R$ is an arbitrary path that may possibly include swaps; in particular, swaps may occur in parts~(I) and~(III). If $p_{\mathrm{out}} < q_0$ and $q_1<(p_{\mathrm{out}}+q_0)/2$ (\textbf{case (b1$^{\prime}$)}), then with positive probability of at least $(1-p_{\mathrm{out}})$ we have $-W_{\mathrm{(III)}}\geq \Delta F_{\mathrm{(III)}}$, and with positive probability of at least $\min\{p_{\mathrm{in}},1-p_\mathrm{in}\}$ we have
    \begin{equation}
        -W_{\mathrm{(I)}}\geq \Delta F_{\mathrm{(I)}}+\overline{\epsilon}_{\mathrm{(I)}}\left(p_{\mathrm{out}}\right),
    \end{equation}
    where $\overline{\epsilon}_{\mathrm{(I)}}(p_{\mathrm{out}})$ is given by Eq. \eqref{def:eps(I)}, that is,  
     \begin{equation}\label{def:eps(I)_B}
        \overline{\epsilon}_{\mathrm{(I)}}\left(p_{\mathrm{out}}\right)\coloneqq
    \frac{1}{\beta} \ln \left(\frac{1+e^{\beta E\left(\tfrac{p_{\mathrm{out}}+q_0}{2}\right)}}{1+e^{\beta E_0}}\right)
    =\frac{1}{\beta}\ln\left(\frac{2q_0}{p_{\mathrm{out}}+q_0}\right),
    \end{equation}
    where the second equality follows from Eq. \eqref{def:E(p)}. 
    Moreover, $\overline{\epsilon}_{\mathrm{(I)}}(p_{\mathrm{out}}) > 0$.
\end{lemma}

% \begin{figur  e}
%     {\includegraphics[width=0.9\linewidth]{figuresNew/skanowanie0130.pdf}}
%     \caption{\color{red}An exemplary stage~(I) of a path $R$ within scenario~(b1$^{\prime}$) (which in particular yields $E_1>E_0$) in the presence of a~swap. Geometrically, we have
% \(
%     0 \;\leq\; \int_{E_0}^{E_1} g(E)\, dE \;\leq\; E_1.
% \)
% Indeed, over stage~(I), the Gibbs curve $g(E)$ on the interval $[E_0,E_1]$ remains -- regardless of the type of swap and the particular shape of the path $R$ -- between the values $0$ and $1/2$. Hence, the area under the curve cannot exceed the area of the rectangle of height $1$ and width $E_1-0=E_1$. 
% }
%     \label{fig:scenario(b1)swaps}
% \end{figure}
% \begin{figure}
%     \centering
%     \includegraphics[width=1\linewidth]{FigPrzedOstatniaKH26.png}
%     \caption{\textcolor{magenta}{[MS: tu jeszcze jeden szczegół, lewa całka powinna być równa 0 a nie różna, moja próba poprawki poniżej]}}
%     \label{fig:placeholder}
% \end{figure}

% \begin{figure}
%     \centering
%     \includegraphics[width=0.6\linewidth]{LekkiLiftingObrazekRóżowoZielony.png}
%     \caption{Enter Caption}
%     \label{fig:placeholder}
% \end{figure}
\begin{figure}
    \centering
    \includegraphics[width=0.5\linewidth]{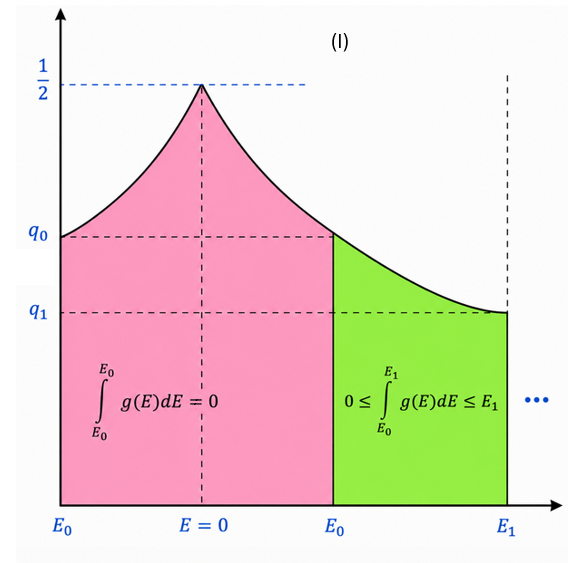}
    \caption{
    An exemplary stage~(I) of a path $R$ within scenario~(b1$^{\prime}$) (which in particular yields $E_1>E_0$) in the presence of a~swap. Geometrically, we have
\(
    0 \;\leq\; \int_{E_0}^{E_1} g(E)\, dE \;\leq\; E_1.
\)
Indeed, over stage~(I), the Gibbs curve $g(E)$ on the interval $[E_0,E_1]$ remains -- regardless of the type of swap and the particular shape of the path $R$ -- between the values $0$ and $1/2$. Hence, the area under the curve cannot exceed the area of the rectangle of height $1$ and width $E_1-0=E_1$. 
    }
    \label{fig:scenario(b1)swaps}
\end{figure}
\begin{proof}
First, note that there is nothing to prove for part~(III), since the undertaken assumption  
\(p_{\mathrm{out}} < q_0\) excludes the possibility of a swap after the last thermalization.  
Therefore, the analysis of \( -W_{\mathrm{(III)}}\) proceeds exactly as in the proofs of  
Lemma~\ref{lem:w3} and Corollary~\ref{corollary:w3_prime_case}.

In part~(I), in contrast, all three swap scenarios (i)--(iii) depicted in  
Fig.~\ref{fig:swap-RI-RIII} may occur.  
Recall that the work loss \(-W_{\mathrm{(I)}}\) is positive if and only if the Gibbs curve is decreasing along the segment \([E_0,E_1]\).  
It is therefore natural to focus on the case in which the swap maps the ground state to the excited state (\([0]\to[1]\)), which happens with probability \((1-p_{\mathrm{in}})\).  
In this situation we have
\begin{equation}
    -W_{\mathrm{(I)}} = 0 + (E_1 - 0) = E_1 > 0 .
\end{equation}

Moreover,
\begin{equation}
    \Delta F_{\mathrm{(I)}} 
        = \int_{E_0}^{E_1} g(E)\, dE
        = E_1 - E_0 - \frac{1}{\beta}\ln\!\left(\frac{1+e^{\beta E_1}}{1+e^{\beta E_0}}\right),
\end{equation}
and (keeping in mind that we have here $E_0<E_1$) this quantity -- although positive -- is strictly smaller than \(E_1\) (see Fig.~\ref{fig:scenario(b1)swaps}, where this inequality is also motivated geometrically).

Consequently, for any path \(R\) that includes a swap in part~(I), with probability \((1-p_{\mathrm{in}})\),
\begin{equation}
\begin{aligned}
    -W_{\mathrm{(I)}} &= E_1 = E_1+\Delta F_{\mathrm{(I)}}-\Delta F_{\mathrm{(I)}}\\
        &= \Delta F_{\mathrm{(I)}} 
            + E_0 
            + \frac{1}{\beta}\ln\!\left(\frac{1+e^{\beta E_1}}{1+e^{\beta E_0}}\right) 
        > \Delta F_{\mathrm{(I)}} + \epsilon_{\mathrm{(I)}}(q_1),
\end{aligned}
\end{equation}
where \(\epsilon_{\mathrm{(I)}}(q_1)\) is the quantity defined in Eq.~\eqref{def:eps(I)(q1)} of Lemma~\ref{lem:w3}.

Wishing to obtain a general estimate, that is, for a path \(R\) that may either include a swap in its first part or not, we combine the above considerations with the statement of Lemma~\ref{lem:w3} to obtain
\begin{equation}\label{eq:eq}
    P\!\left(
        -W_{\mathrm{(I)}} \ge \Delta F_{\mathrm{(I)}} + \epsilon_{\mathrm{(I)}}(q_1)
    \right)
    \ge \min\{p_{\mathrm{in}},\, 1 - p_{\mathrm{in}}\}.
\end{equation}

Thus, to complete the proof, it suffices to repeat the argument given in the proof of Corollary~\ref{corollary:w3_prime_case}, including the part establishing the positivity of $\overline{\epsilon}_{\mathrm{(I)}}(p_{\mathrm{out}})$.
\end{proof}

It now remains to deal with \textbf{case~(a)}: $q_0 < p_{\mathrm{out}}$ (see Fig.~\ref{fig:scenario(a)swaps}), which amounts to proving an analogue of Lemma~\ref{lem:w1}, valid for all types of paths, including those that involve swaps. The forthcoming statement can then be used to formulate a general version of Theorem~\ref{thm:case(a)altogether}. We have
\begin{lemma}\label{lem:w1_swaps}
    Let $p_\mathrm{in},p_\mathrm{out}\in(0,1)$. Suppose that $R$ is an arbitrary path that may possibly include swaps; in particular, swaps may occur in parts~(I) and~(III). If $q_0 < p_\mathrm{out}$, then, with positive probability of at least $\min\{p_{\mathrm{in}},1-p_{\mathrm{in}}\}$ we have $ -W_\mathrm{(I)} \geq \Delta F_\mathrm{(I)}$, and with positive probability of at least $p_\mathrm{out}$ we have 
    \begin{equation}
		  -W_\mathrm{(III)} 
		\geq \Delta F_\mathrm{(III)} +\epsilon_\mathrm{(III)}(p_{\mathrm{out}}).
    \end{equation}
    where 
%{\color{red}
%\begin{align}\label{def:eps3_out_B}
%\overline{\epsilon}_\mathrm{(III)}(p_{\mathrm{out}})
%= 
%\begin{cases}
%\frac{1}{\beta}
%\ln\!\left(
%\frac{1 + e^{\beta E_0}}
%{1 + e^{\beta E\left(p_{\mathrm{out}}\right)}}
%\right)
%    &\text{if}\quad p_\mathrm{out}\in(0,1/2]\quad\text{(the case of no swap in part~(III))},\\[2mm]
%E\left(1 - p_{\mathrm{out}}\right)+
%\frac{1}{\beta}
%\ln\!\left(
%\frac{1 + e^{\beta E_0}}
%{1 + e^{\beta E\left(1-p_{\mathrm{out}}\right)}}
%\right)
%    &\text{if}\quad p_\mathrm{out}\in(1/2,1)\quad\text{(the case of a swap in part~(III))}
%\end{cases}
%\end{align}
%(note that $E(p)$ is given by Eq.~\eqref{def:E(p)} for any $p \in (0, 1/2]$, and hence $E(1 - p_\mathrm{out})$ is well-defined whenever $p_\mathrm{out} \in (1/2, 1)$).
%} 
${\epsilon}_\mathrm{(III)}(p_{\mathrm{out}})$ is given by Eq. \eqref{def:eps3_out}, and -- after applying Eq. \eqref{def:E(p)} -- it can be written as
\begin{equation}\label{def:eps3_out_B}
    {\epsilon}_\mathrm{(III)}(p_{\mathrm{out}})
    =\frac{1}{\beta}\ln\left(\frac{p_{\mathrm{out}}}{q_0}\right).
\end{equation}
    Moreover, ${\epsilon}_\mathrm{(III)}(p_{\mathrm{out}}) > 0$.
\end{lemma}

% \begin{figure}
%     {\includegraphics[width=0.9\linewidth]{figuresNew/skanowanie0131.pdf}}
%     \caption{\color{red}
%     Scenario~(a) allows for a swap (necessarily a \emph{swap up}) after the last thermalization at the point $(E_{N-1}, q_{N-1})$. In this exceptional case, we have
% \(
%     p_\mathrm{out} = 1 - q_{N-1} > 1/2,
% \)
% rather than the standard relation \(p_\mathrm{out} = q_{N-1}\). The figure also illustrates that, within scenario~(a) and under the swap scenario \([0] \to [1]\) (whose probability is equal to $1-q_{N-1}=p_\mathrm{out}$), the quantity \( -W_{\mathrm{(III)}} \) is positive in part~(III). 
% }
%     \label{fig:scenario(a)swaps}
% \end{figure}
\begin{figure}
    \centering
    \includegraphics[width=0.5\linewidth]{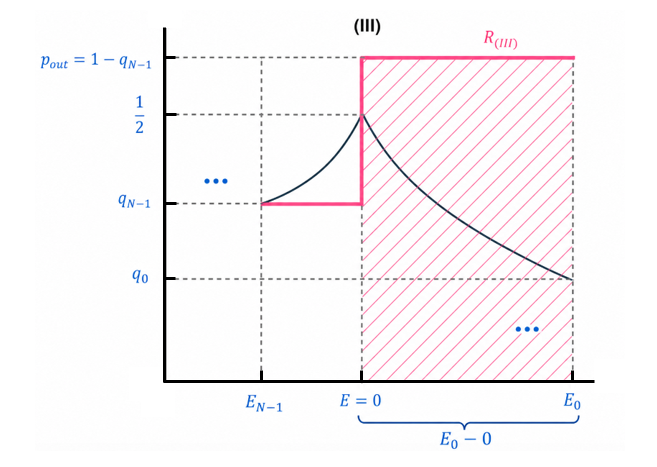}
    \caption{
    Scenario~(a) allows for a swap (necessarily a \emph{swap up}) after the last thermalization at the point $(E_{N-1}, q_{N-1})$. In this exceptional case, we have
\(
    p_\mathrm{out} = 1 - q_{N-1} > 1/2,
\)
rather than the standard relation \(p_\mathrm{out} = q_{N-1}\). The figure also illustrates that, within scenario~(a) and under the swap scenario \([0] \to [1]\) (whose probability is equal to $1-q_{N-1}=p_\mathrm{out}$), the quantity \( -W_{\mathrm{(III)}} \) is positive in part~(III). 
    }
    \label{fig:scenario(a)swaps}
\end{figure}
% \begin{figure}
%     \centering
%     \includegraphics[width=0.5\linewidth]{OstatniObrazekKH-v2.png}
%     \caption{
%     Scenario~(a) allows for a swap (necessarily a \emph{swap up}) after the last thermalization at the point $(E_{N-1}, q_{N-1})$. In this exceptional case, we have
% \(
%     p_\mathrm{out} = 1 - q_{N-1} > 1/2,
% \)
% rather than the standard relation \(p_\mathrm{out} = q_{N-1}\). The figure also illustrates that, within scenario~(a) and under the swap scenario \([0] \to [1]\) (whose probability is equal to $1-q_{N-1}=p_\mathrm{out}$), the quantity \( -W_{\mathrm{(III)}} \) is positive in part~(III). 
% %\textcolor{magenta}{[MS: moim zdaniem rysunek jest generalnie ok tylko 1/2 powino być dokłądnie pośrodku $q_{N-1}$ oraz $p_{out}$ ale jeśli bedzie problem z tą poprawką to może można tak zostawić bo oznaczenia są ok]}
% }
%     \label{fig:scenario(a)swaps}
% \end{figure}

\begin{proof}
In part~(I), we need to consider two mutually exclusive cases:
\begin{itemize}
    \item $q_1 < q_0$, and thus $E_0 < E_1$,
    \item $q_0 < q_1$, and thus $E_1 < E_0$.
\end{itemize}
In the first case, we may simply repeat the argument from Eq.~\eqref{eq:eq}, namely,
\begin{equation}
    P\!\left(-W_{\mathrm{(I)}} \ge \Delta F_{\mathrm{(I)}}\right)
    \ge 
    P\!\left(-W_{\mathrm{(I)}} \ge \Delta F_{\mathrm{(I)}} + \epsilon_{\mathrm{(I)}}(q_1)\right)
    \ge 
    \min\left\{p_{\mathrm{in}},\, 1 - p_{\mathrm{in}}\right\},
\end{equation}
where we note that $\epsilon_{\mathrm{(I)}}(q_1)$, defined in~\eqref{def:eps(I)(q1)}, is positive whenever $q_1 < q_0$.

    In the second of the above cases, we know, on the one hand, that for a path without a swap in part~(I) we obtain (due to Lemma~\ref{lem:w1})
\begin{equation}
    P\!\left(-W_{\mathrm{(I)}} \ge \Delta F_{\mathrm{(I)}}\right)
    = \min\left\{p_{\mathrm{in}},\, 1 - p_{\mathrm{in}}\right\}.
\end{equation}
On the other hand, if a swap occurs, we may argue analogously to the proof of the general Lemma~\ref{lem:easy_swaps} to conclude that
\begin{equation}
    P\!\left(-W_{\mathrm{(I)}} \ge 0\right) \ge 1 - p_{\mathrm{in}}.
\end{equation}
Combining this with the fact that for $E_1 < E_0$ we have
\begin{equation}
    \Delta F_{\mathrm{(I)}} = \int_{E_0}^{E_1} g(E)\, dE < 0,
\end{equation}
and thus we finally arrive at
\begin{equation}
    P\!\left(-W_{\mathrm{(I)}} \ge \Delta F_{\mathrm{(I)}}\right)
    \ge \min\left\{p_{\mathrm{in}},\, 1 - p_{\mathrm{in}}\right\}.
\end{equation}

The proof for part~(I) is therefore completed.\vspace{2mm}

Let us therefore proceed to analysing part~(III), which in scenario~(a) may involve a swap (see Fig.~\ref{fig:swap-RI-RIII}). Note that in the case of a swap after the last thermalization (cf. Fig. \ref{fig:scenario(a)swaps}) we exceptionally have 
\(
    p_{\mathrm{out}} = 1 - q_{N-1}>1/2
\)
(rather than \(p_{\mathrm{out}} = q_{N-1}\), as in the standard situation). Allowing for a swap from the ground to the excited level \([0]\to[1]\), which therefore occurs with probability \((1 - q_{N-1})=p_{\mathrm{out}}\), we obtain
\begin{equation}
    -W_{\mathrm{(III)}} = 0 + (E_0 - 0) = E_0.
\end{equation}
Moreover,
\begin{equation}
    \Delta F_{\mathrm{(III)}}
    = \int_{E_{N-1}}^{E_0} g(E)\, dE
    = E_0 - E_{N-1}
      - \frac{1}{\beta}
      \ln\!\left( \frac{1 + e^{\beta E_0}}{1 + e^{\beta E_{N-1}}} \right),
\end{equation}
which gives us
\begin{equation}
    -W_\mathrm{(III)}=E_0+\Delta F_\mathrm{(III)}-\Delta F_\mathrm{(III)}
    =\Delta F_\mathrm{(III)}+E_{N-1}+\frac{1}{\beta}\ln\left(\frac{1+e^{\beta E_0}}{1+e^{\beta E_{N-1}}}\right).
\end{equation}
Whenever $q_0 < 1 - p_\mathrm{out} < 1/2 < p_\mathrm{out}$, which in particular implies 
$E_0 > E_{N-1} = E(1 - p_\mathrm{out}) > 0$ (note that $E(p)$ is given by Eq.~\eqref{def:E(p)} for any $p \in (0, 1/2]$, and hence $E(1 - p_\mathrm{out})$ is well-defined whenever $p_\mathrm{out} \in (1/2, 1)$), we can readily conclude that the quantity
\begin{equation}
    E_{N-1} + \frac{1}{\beta}\ln\left(\frac{1 + e^{\beta E_0}}{1 + e^{\beta E_{N-1}}}\right)
\end{equation}
is strictly positive. 

In the opposite case, namely when $1 - p_\mathrm{out} < q_0 < 1/2 < p_\mathrm{out}$, which implies 
$0 \leq E_0 < E_{N-1}=E(1 - p_\mathrm{out})$, we arrive at the same conclusion. Indeed, in this regime
\begin{equation}
    E_{N-1} + \frac{1}{\beta}\ln\left(\frac{1 + e^{\beta E_0}}{1 + e^{\beta E_{N-1}}}\right)
    = -W_{\mathrm{(III)}} - \Delta F_\mathrm{(III)},
\end{equation}
and the second term on the right-hand side of the above equality is -- in this case -- strictly positive ($-\Delta F_\mathrm{(III)}>0$), while the first is non-negative ($-W_\mathrm{(III)}\ge 0$).

As a consequence, in the case of a swap after the last thermalization, we obtain
\begin{equation}
    E_{N-1} + \frac{1}{\beta}\ln\left(\frac{1 + e^{\beta E_0}}{1 + e^{\beta E_{N-1}}}\right)=E\left(1-p_\mathrm{out}\right) + \frac{1}{\beta}\ln\left(\frac{1 + e^{\beta E_0}}{1 + e^{\beta E\left(1-p_\mathrm{out}\right)}}\right)>0
\end{equation}
and
\begin{equation}
P\!\left(
        -W_{\mathrm{(III)}} \ge
        \Delta F_{\mathrm{(III)}}
        + E\left(1-p_\mathrm{out}\right) + \frac{1}{\beta}\ln\left(\frac{1 + e^{\beta E_0}}{1 + e^{\beta E\left(1-p_\mathrm{out}\right)}}\right)
    \right)
    \ge p_{\mathrm{out}}.
\end{equation}
Applying Eq.~\eqref{def:E(p)}, we can simplify the above to the following form:
\begin{equation}
P\!\left(
        -W_{\mathrm{(III)}} \ge
        \Delta F_{\mathrm{(III)}}
        +{\epsilon}_{\mathrm{(III)}}(p_{\mathrm{out}})\right)
    \ge p_{\mathrm{out}}\quad\text{with}\quad {\epsilon}_{\mathrm{(III)}}(p_{\mathrm{out}}):=\frac{1}{\beta}\ln\left(\frac{p_\mathrm{out}}{q_0}\right).
\end{equation}

Taking into account the assertion of Lemma~\ref{lem:w1} (valid for those paths that do not involve swaps), and referring again to Eq. \eqref{def:E(p)}, we finally obtain the following estimate for an arbitrary path \(R\) (such that \(q_0 < p_{\mathrm{out}}\)) that may or may not involve a swap in its last part:
\begin{equation}
    P\!\left(
        -W_{\mathrm{(III)}} \ge
        \Delta F_{\mathrm{(III)}} +
        {\epsilon}_{\mathrm{(III)}}(p_{\mathrm{out}})
    \right)\geq p_{\mathrm{out}}.
\end{equation}
%\begin{align}
%\begin{aligned}
%    &P\!\left(
%        -W_{\mathrm{(III)}} \ge
%        \Delta F_{\mathrm{(III)}} +
%        \overline{\epsilon}_{\mathrm{(III)}}(p_{\mathrm{out}})
%    \right)\\
%    &\hspace{2mm}=\begin{cases}
% P\!\left(
%        -W_{\mathrm{(III)}} \ge
%        \Delta F_{\mathrm{(III)}} +
%            \frac{1}{\beta}
%            \ln\!\left(
%                \frac{1 + e^{\beta E_0}}
%                     {1 + e^{\beta E(p_{\mathrm{out}})}}
%            \right)\right)
%            &\text{if}\quad p_\mathrm{out}\in(0,1/2]\quad\text{(the case of no swap in part~(III))},\\[2mm]
%P\!\left(
%        -W_{\mathrm{(III)}} \ge
%        \Delta F_{\mathrm{(III)}} +
%        E\left(1-p_\mathrm{out}\right) + \frac{1}{\beta}\ln\left(\frac{1 + e^{\beta E_0}}{1 + e^{\beta E\left(1-p_\mathrm{out}\right)}}\right)
%    \right)&\text{if}\quad p_\mathrm{out}\in(1/2,1)\quad\text{(the case of a swap in part~(III))}.
%\end{cases}\\[2mm]
%    &\hspace{2mm}\ge
%    {\color{blue}p_{\mathrm{out}}}.
%\end{aligned}
%\end{align}
The positivity of ${\epsilon}_{\mathrm{(III)}}(p_{\mathrm{out}})$ in the case of $q_0<p_\mathrm{out}$ follows either from the assertion of Lemma~\ref{lem:w1} and the above analysis, or directly from the definition \eqref{def:eps3_out_B}.
The proof is now complete.

\end{proof}

\subsection{Work of the entire path \texorpdfstring{$R$}{R} that possibly includes swaps}
Having established the general building blocks, namely Lemmas~\ref{lem:w2_swaps}--\ref{lem:w1_swaps}, which hold for all types of paths (whether involving swaps or not), we can use them in place of Lemmas~\ref{lem:w2}--\ref{lem:w1} and Corollary~\ref{corollary:w3_prime_case} to obtain the general versions of Theorems~\ref{thm:case(a)altogether} and \ref{thm:case(b)altogether}, namely:

\begin{theorem}\label{thm:case(a)altogether_swaps}
For any $p\in(0,1)$ let
\begin{equation}\label{def:epsisilon(a)_swaps}
    {\epsilon}_\mathrm{(a)}\left(p\right)\coloneqq
\frac{1}{2}\epsilon_\mathrm{(III)}(p)=\frac{1}{2\beta}\ln\left(\frac{p}{q_0}\right)
\end{equation}
and 
\begin{equation}\label{def:p(a)_swaps}
    {\eta}_\mathrm{(a)}\left(p\right)\coloneqq
    \min\left\{p_\mathrm{in},1-p_\mathrm{in}\right\}p\left(1- e^{-\beta {\epsilon}_\mathrm{(a)}\left(p\right)} \right)=\min\left\{p_\mathrm{in},1-p_\mathrm{in}\right\}p\left(1-\left(\frac{q_0}{p}\right)^{1/2}\right).
\end{equation}

    Consider an arbitrary path $R$, either involving swaps or not, such that $q_0 < p_\mathrm{out}$ (\textbf{case (a)}). Then the constants ${\epsilon}_\mathrm{(a)}(p_\mathrm{out})$ and ${\eta}_\mathrm{(a)}(p_\mathrm{out})$ are both positive, and the work $W$ of $R$ satisfies
\begin{equation}
    -W \ge {\epsilon}_\mathrm{(a)}(p_\mathrm{out})
\end{equation}
with positive probability of at least ${\eta}_\mathrm{(a)}(p_\mathrm{out})$.
\end{theorem}

\begin{theorem}\label{thm:case(b)altogether_swaps}
    For any $p\in(0,1/2]$ let
    \begin{align}\label{def:epsisilon(b)_swaps}
    \begin{aligned}
        \epsilon_\mathrm{(b)}\left(p\right)&\coloneqq
    \frac{1}{2}\min\left\{\frac{1}{\beta} \ln \left(\frac{1+e^{\beta E\left(\tfrac{p+q_0}{2}\right)}}{1+e^{\beta E_0}}\right)\,,\, E\!\left(p\right) - E\left(\tfrac{p+q_0}{2}\right)
           + \frac{1}{\beta} 
             \ln\!\left(
                 \frac{1 + e^{\beta E\left(\tfrac{p+q_0}{2}\right)}}
                      {1 + e^{\beta E(p)}}
             \right)\right\}\\
             &=\frac{1}{2\beta}\min\left\{\ln\left(\frac{2q_0}{p+q_0}\right),\ln\left(\frac{2(1-p)}{2-p-q_0}\right)\right\}
    \end{aligned}
    \end{align}
    and
    \begin{align}\label{def:p(b)_swaps}
    \begin{aligned}
        \overline{\eta}_\mathrm{(b)}\left(p\right)
    &\coloneqq\min\left\{\left(1 - p\right)
        \, \min\!\left\{p_{\mathrm{in}},1-p_{\mathrm{in}}\right\}
        \left(1 - \left(\frac{1+e^{\beta E_0}}{1+e^{\beta E\left(\tfrac{p+q_0}{2}\right)}}\right)^{1/2}\right)\, ,\right.\\ 
        &\hspace{25mm}
        \left.\left(1-p_\mathrm{in}\right)\left(1-p\right)\left(1-e^{-\frac{\beta}{2}\left(E\!\left(p\right) - E\left(\tfrac{p+q_0}{2}\right)
          \right)}
          \left(\frac{1+e^{\beta E\left(p\right)}}{1+e^{\beta E\left(\tfrac{p+q_0}{2}\right)} }\right)^{1/2}
          \right) \right\}\\
          &=\min\left\{(1-p)\min\!\left\{p_{\mathrm{in}},1-p_{\mathrm{in}}\right\}\left(1-\left(\frac{p+q_0}{2q_0}\right)^{1/2}\right),\left(1-p_\mathrm{in}\right)(1-p)\left(1-\left(\frac{p\left(2-p-q_0\right)}{(1-p)\left(p+q_0\right)}\right)^{1/2}\right)\left(\frac{p+q_0}{2p}\right)^{1/2}\right\}.
    \end{aligned}
    \end{align}
    Consider an arbitrary path $R$, either involving swaps or not, such that $q_0>p_\mathrm{out}$ (\textbf{case(b)}). Then the constants $\epsilon_\mathrm{(b)}(p_\mathrm{out})$ and $\overline{\eta}_\mathrm{(b)}(p_\mathrm{out})$ are both positive and 
the work $W$ of $R$ fulfills 
\begin{equation}
    -W\geq \epsilon_\mathrm{(b)}\left(p_\mathrm{out}\right)
\end{equation}
with a positive probability of at least $\overline{\eta}_\mathrm{(b)}(p_\mathrm{out})$.
\end{theorem}

\subsection{\textbf{Combining paths into a process (allowing all three types of transformations \textbf{CO2}, \textbf{CO3}, and \textbf{CO5})}
}

Recall that the constants 
${\epsilon}_\mathrm{(a)}, \epsilon_\mathrm{(b)}, {\eta}_\mathrm{(a)}, \overline{\eta}_\mathrm{(b)}$ 
are defined by Eqs.~\eqref{def:epsisilon(a)_swaps}--\eqref{def:p(b)_swaps}.\vspace{2mm}

Summarizing the above results from both Appendices~\ref{sec:appendA} and~\ref{sec:appendix-with-swaps}, we can finally formulate the main general result 
(by which we mean that all types of operations \textbf{CO2}, \textbf{CO3}, and \textbf{CO5} can be applied while transforming the state 
$\rho \coloneqq (1-p_{\mathrm{IN}}, p_{\mathrm{IN}})$ into the state 
$\sigma \coloneqq (1-p_{\mathrm{OUT}}, p_{\mathrm{OUT}})$), 
whose proof is almost identical to that of Theorem~\ref{thm:main_appendix}, 
with the only difference that Theorems~\ref{thm:case(a)altogether_swaps} and~\ref{thm:case(b)altogether_swaps} 
are applied instead of Theorems~\ref{thm:case(a)altogether} and~\ref{thm:case(b)altogether}, respectively.

\begin{theorem}
	\label{thm:main_appendix_swaps}
Let $p_\mathrm{IN}\in(0,1)$. For arbitrary $p_\mathrm{OUT}\in(0,1)$ and $q_0\in(0,1/2]$ such that $p_\mathrm{OUT}\not=q_0$ consider one of the scenarios (i)-(iv) listed at the beginning of Section \ref{sec:idea} (which exclude those cases in which a~transition from $\rho \coloneqq (1-p_{\mathrm{IN}}, p_{\mathrm{IN}})$ to $\sigma \coloneqq (1-p_{\mathrm{OUT}}, p_{\mathrm{OUT}})$ can be achieved with no work loss). 
Then, while transforming state $\rho \coloneqq (1-p_{\mathrm{IN}}, p_{\mathrm{IN}})$ into state $\sigma \coloneqq (1-p_{\mathrm{OUT}}, p_{\mathrm{OUT}})$ by operations \textbf{CO2, CO3} and \textbf{CO5}, one has to spend at least the following amount of work
	\begin{equation}\label{def:epsilon_final_B}
{\epsilon}\coloneqq \max\left\{{\epsilon}_\mathrm{(a)}\left(\frac{p_\mathrm{OUT}+q_0}{2}\right)
\,,\,
\epsilon_\mathrm{(b)}\left(\frac{p_\mathrm{OUT}+q_0}{2}\right)\right\}  > 0
	\end{equation}
 with positive probability of at least
 \begin{equation}\label{def:p_final_B}
     \overline{\eta}\coloneqq \max\left\{\,
     {\eta}_\mathrm{(a)}\left(\frac{p_\mathrm{OUT}+q_0}{2}\right)\frac{p_\mathrm{OUT}-q_0}{2}\,,\,
     \overline{\eta}_\mathrm{(b)}\left(\frac{p_\mathrm{OUT}+q_0}{2}\right)\frac{q_0-p_\mathrm{OUT}}{2}
     \,\right\}.
 \end{equation}
\end{theorem}

\begin{remark}
Note, that if $p_\mathrm{OUT}=1$
then even by Thermal Operations, 
one cannot obtain it from any state,
so even more by $\mcoww$, which are 
a subset.  
    
\end{remark}
% xxx
% ---------------------------------------------------\newline\newline

% \textit{\textbf{{\color{red}[HW-Ś: This comes from an old version of the text -- it should be verified to ensure it is still up to date!
% ] [HW-Ś: Rozpisać niewyjaśnie przypadki z $p_\mathrm{OUT}\in\{0,1\}$ -- to już zrobione w głównym tekście.]}}}

To finish the characterization, we show below what states can be transformed by $\mcoww$.
\begin{fact}
\label{fact:trivial}
    It is possible via $\mcoww$ to transform state $\rho \coloneqq (1-p_{\mathrm{IN}}, p_{\mathrm{IN}})$ into state $\sigma \coloneqq (1-p_{\mathrm{OUT}}, p_{\mathrm{OUT}})$ if $p_{\mathrm{IN}} \geq p_{\mathrm{OUT}} \geq q_0$ or if $p_{\mathrm{IN}} \leq p_{\mathrm{OUT}} \leq q_0$.
\end{fact}
This is possible because such a state can be obtained as a mixture of the identity operation (with no thermalization and no swap) and an operation containing a single thermalization at the very beginning of the process (again with no swaps). 
To understand why this does not contradict the previous theorems, note that in this case part~(III) of the path may be trivial. Consequently, there is no negative work contribution arising from this part.
%\cred
%TODO argumentation 

%\blk
\begin{fact}
\label{fact:pure-exc}
    It is possible   via $\mcoww$ to transform state $\rho \coloneqq (0, 1)$ into any state.
\end{fact}
Indeed, as in the previous remark, any state more excited than the Gibbs state can be obtained simply by partial thermalization. Moreover, the ground state can be reached for free (i.e., without any work expenditure) by lowering the energy level to zero, performing a swap, and then raising the empty level back to \(E_0\). Starting from the ground state, any state less excited than the Gibbs state can subsequently be obtained via partial thermalization.

Note that this does not contradict our previous no-go results, since in the present situation part~(III) of the path is absent. Indeed, part~(III) necessarily appears only when thermalization is performed after changing the Hamiltonian, and in the case \(p_{\mathrm{IN}} = 1\) this step can be avoided.

%Then we can just bound it by the path that has thermalization at the very beginning which implies that there is no $R_{\mathrm{I}}$ part.

\blk

%\cred 
%tell here, that only one lemma is affected by swap - namely the one about variance.
%\blk

%\textcolor{red}{maybe better divide above on into II and III parts?}

\section{Lemmas from probability theory}

In the final appendix, we present simple consequences of the Markov and Chebyshev inequalities. 
\begin{lemma}\label{lem:reverseMarkov}
	Let $Y$ be a random variable taking values in $[0,1]$. 
    Then, for any $\alpha>0$, 
	\begin{equation}
	\label{eq:Markov1}
		P(Y<\alpha) \geq 1 - \frac{\E Y}{\alpha}
	\end{equation}
	and
	\begin{equation}
	\label{eq:Markov2}
		P(Y>\alpha) \geq 1 - \frac{1 - \E Y}{1-\alpha}.
	\end{equation}
\end{lemma}
\begin{proof}
    The first inequality is a direct application of Markov's inequality, that is, 
    $P(Y \geq \alpha) \leq \E Y / \alpha$. 
	For the second one, let us define the random variable $X = 1 - Y$. 
    It is easy to see that $X$ also takes values in $[0,1]$ and that $\E X = 1 - \E Y$. 
    Applying Markov's inequality to  $X$, we obtain
	\begin{align}
        P(Y>\alpha)=P(X<1-\alpha)
        =1-P(X\geq 1-\alpha)
        \geq 1 - \frac{\E X}{1-\alpha}
		=1 - \frac{1 - \E Y}{1-\alpha}.
	\end{align}
	The proof is therefore completed.
\end{proof}

\end{document}